\documentclass[letterpaper,twocolumn,10pt]{article}
\usepackage{usenix2019_v3}
\usepackage{xspace}
\usepackage{xparse}
\usepackage{amsmath}
\usepackage{amsfonts}
\usepackage{amsthm}
\usepackage{enumitem}
\usepackage{color}
\usepackage{mathtools}
\usepackage{array}
\usepackage{tikz}
\usetikzlibrary{arrows}
\usetikzlibrary{patterns}

\usepackage{stackrel}
\usepackage{pstricks}
\usepackage[skip=3pt,font=footnotesize]{subcaption}
\usepackage{amssymb}
\usepackage{pgfplots}
\usepackage{booktabs}
\usepackage{expl3}
\usepackage{graphicx}

\newcounter{note}[section]

\newcommand{\secref}[1]{\mbox{Sec.~\ref{#1}}\xspace}

\newcommand{\lineref}[1]{\mbox{line~\ref{#1}}\xspace}
\newcommand{\linesref}[2]{\mbox{lines~\ref{#1}--\ref{#2}}\xspace}
\newcommand{\secsref}[2]{\mbox{Secs.~\ref{#1}--\ref{#2}}\xspace}
\newcommand{\secrefstatic}[1]{\mbox{Section~{#1}}}
\newcommand{\figref}[1]{\mbox{Fig.~\ref{#1}}\xspace}

\newcommand{\figrefstatic}[1]{\mbox{Fig.~{#1}}\xspace}
\newcommand{\tblref}[1]{\mbox{Table~\ref{#1}}\xspace}
\newcommand{\tblrefstatic}[1]{\mbox{Table~{#1}}\xspace}
\newcommand{\tblsrefstatic}[2]{\mbox{Tables~{#1}--{#2}}\xspace}
\newcommand{\appref}[1]{\mbox{App.~\ref{#1}}\xspace}
\newcommand{\eqnref}[1]{\mbox{(\ref{#1})}\xspace}

\newcommand{\propref}[1]{\mbox{Prop.~\ref{#1}}\xspace}
\newcommand{\propsref}[2]{\mbox{Props.~\ref{#1}--\ref{#2}}\xspace}
\newcommand{\msgref}[1]{\mbox{message~\ref{#1}}\xspace}
\newcommand{\dfnrefstatic}[1]{\mbox{Definition~{#1}}\xspace}
\newcommand{\enumref}[1]{\mbox{step~(\ref{#1})}\xspace}

\newtheorem{prop}{Proposition}

\newcommand{\gibibytes}{\ensuremath{\mathrm{GiB}}\xspace}
\newcommand{\gigahertz}{\ensuremath{\mathrm{GHz}}\xspace}
\newcommand{\secs}{\ensuremath{\mathrm{s}}\xspace}

\newcounter{requesterLineNmbr}
\renewcommand{\therequesterLineNmbr}{\ensuremath{\mathsf{r\arabic{requesterLineNmbr}}}}
\newcommand{\requesterLabel}[1]{\refstepcounter{requesterLineNmbr}\label{#1}\therequesterLineNmbr.}

\newcounter{responderLineNmbr}
\renewcommand{\theresponderLineNmbr}{\ensuremath{\mathsf{s\arabic{responderLineNmbr}}}}
\newcommand{\responderLabel}[1]{\refstepcounter{responderLineNmbr}\label{#1}\theresponderLineNmbr.}

\newcounter{messageNmbr}
\renewcommand{\themessageNmbr}{\ensuremath{\mathsf{m\arabic{messageNmbr}}}}
\newcommand{\messageLabel}[1]{\refstepcounter{messageNmbr}\label{#1}\themessageNmbr.~~}

\newlength{\figureheight}
\newcommand{\figurewidth}{\columnwidth}

\definecolor{curve_color}{rgb}{0.129411764705882,0.380392156862745,0.549019607843137}

\newcommand{\requesterTerm}{requester\xspace}
\newcommand{\requestersTerm}{requesters\xspace}
\newcommand{\RequesterTerm}{Requester\xspace}
\newcommand{\responderTerm}{responder\xspace}
\newcommand{\ResponderTerm}{Responder\xspace}
\newcommand{\respondersTerm}{responders\xspace}
\newcommand{\directoryTerm}{directory\xspace}
\newcommand{\DirectoryTerm}{Directory\xspace}
\newcommand{\extractionComplexity}{extraction complexity\xspace}
\newcommand{\accountLocationPrivacy}{account location privacy\xspace}
\newcommand{\trustedForLoginPrivacy}{TLP\xspace}
\newcommand{\untrustedForLoginPrivacy}{ULP\xspace}

\newcommand{\codeExpt}{\ensuremath{\mathtt{Experiment}~}}

\newcommand{\codeIf}{\ensuremath{\mathtt{if}~}}

\newcommand{\codeReturn}{\ensuremath{\mathtt{return}~}}
\newcommand{\codeAbort}{\ensuremath{\mathtt{abort}~}}

\newcommand{\genericRV}{\ensuremath{Y}\xspace}
\newcommand{\genericRVAlt}{\ensuremath{Y'}\xspace}
\newcommand{\genericSet}{\ensuremath{Z}\xspace}

\newcommand{\genericNat}{\ensuremath{z}\xspace}

\newcommand{\genericMatrix}{\ensuremath{\matrixNotation{Z}}\xspace}
\newcommand{\residues}[1]{\ensuremath{\mathbb{Z}_{#1}}\xspace}
\newcommand{\setSize}[1]{\ensuremath{|{#1}|}\xspace}

\newcommand{\nats}[1]{\ensuremath{[{#1}]}\xspace}
\newcommand{\getsr}{\;\stackrel{\$}{\leftarrow}\;}
\newcommand{\boolTrue}{\ensuremath{\mathit{true}}\xspace}
\newcommand{\boolFalse}{\ensuremath{\mathit{false}}\xspace}

\newcommand{\distEqual}{\ensuremath{\mathbin{\;\overset{d}{=}\;}}\xspace}

\newcommand{\testEquiv}[1]{\ensuremath{\mathbin{{\;\overset{?}{\equiv}}_{#1}}}\xspace}
\newcommand{\testIn}{\ensuremath{\mathbin{\;\overset{?}{\in}\;}}\xspace}

\newcommand{\prob}[1]{\ensuremath{\mathbb{P}\left({#1}\right)}\xspace}
\newcommand{\expv}[1]{\ensuremath{\mathbb{E}\left({#1}\right)}\xspace}
\newcommand{\cset}[3]{\ensuremath{#1\{}{#2}\ensuremath{\;#1|} \ifmmode{\;}\fi {#3}\ensuremath{#1\}}\xspace}
\newcommand{\cprob}[3]{\ensuremath{\mathbb{P}#1(}{#2}\ensuremath{\;#1|} \ifmmode{\;}\fi {#3}\ensuremath{#1)}\xspace}
\newcommand{\cexpv}[3]{\ensuremath{\mathbb{E}#1(}{#2}\ensuremath{\;#1|} \ifmmode{\;}\fi {#3}\ensuremath{#1)}\xspace}
\newcommand{\matrixNotation}[1]{\ensuremath{\mathbf{#1}}\xspace}
\newcommand{\matrixComponent}[3]{\ensuremath{\left({#1}\right)_{{#2},{#3}}}\xspace}
\newcommand{\matrixRowIdx}{\ensuremath{i}\xspace}
\newcommand{\matrixColIdx}{\ensuremath{j}\xspace}

\newcommand{\matrixNmbrRows}{\ensuremath{\alpha}\xspace}
\newcommand{\matrixNmbrCols}{\ensuremath{\alpha'}\xspace}
\newcommand{\matrixNmbrColsAlt}{\ensuremath{\alpha''}\xspace}

\newcommand{\encScheme}{\ensuremath{\mathcal{E}}\xspace}
\newcommand{\keygen}{\ensuremath{\mathsf{Gen}}\xspace}
\newcommand{\encrypt}[1]{\ensuremath{\mathsf{Enc}_{#1}}\xspace}

\NewDocumentCommand{\plaintext}{ g g }{\ensuremath{m\IfNoValueF{#1}{\IfNoValueTF{#2}{_{#1}}{_{{#1},{#2}}}}}\xspace}
\newcommand{\ciphertext}[1]{\ensuremath{c_{#1}}\xspace}
\newcommand{\ciphertextIdx}{\ensuremath{k}\xspace}
\newcommand{\plaintextMatrix}{\ensuremath{\matrixNotation{M}}\xspace}
\newcommand{\ciphertextSpace}[1]{\ensuremath{C_{#1}}\xspace}
\newcommand{\secParam}{\ensuremath{\kappa}\xspace}
\newcommand{\privKey}{\ensuremath{\mathit{sk}}\xspace}
\newcommand{\pubKey}{\ensuremath{\mathit{pk}}\xspace}
\newcommand{\encAdd}[1]{\ensuremath{+_{#1}}\xspace}

\newcommand{\encMatrixAdd}[1]{\ensuremath{+_{#1}}\xspace}
\newcommand{\encMatrixScalarMult}[1]{\ensuremath{\ast_{#1}}\xspace}

\newcommand{\encMatrixHadamardMult}[1]{\ensuremath{\circ_{#1}}\xspace}
\NewDocumentCommand{\encSum}{ g g g }{\ensuremath{\displaystyle\operatorname*{\textstyle\sum_{\mathrlap{#1}}}\IfNoValueF{#2}{_{#2}}\IfNoValueF{#3}{^{#3}}\hphantom{_{#1}}}\xspace}
\newcommand{\encZeroTest}[1]{\ensuremath{\mathsf{isZero}_{#1}}\xspace}
\newcommand{\encZeroTestResult}{\ensuremath{z}\xspace}
\newcommand{\encSumIdx}{\ensuremath{k}\xspace}
\newcommand{\encSumIdxAlt}{\ensuremath{k'}\xspace}

\newcommand{\ciphertextMatrix}{\ensuremath{\matrixNotation{C}}\xspace}
\newcommand{\ciphertextMatrixAlt}{\ensuremath{\matrixNotation{C'}}\xspace}

\newcommand{\fieldOrder}{\ensuremath{r}\xspace}
\newcommand{\fieldAdd}{\ensuremath{+}\xspace}
\newcommand{\fieldSum}{\ensuremath{\sum}\xspace}
\newcommand{\fieldMult}{\ensuremath{\times}\xspace}
\newcommand{\fieldAddIdentity}{\ensuremath{\mathbf{0}}\xspace}
\newcommand{\fieldMultIdentity}{\ensuremath{\mathbf{1}}\xspace}
\newcommand{\fieldNegative}{\ensuremath{-}\,\xspace}
\newcommand{\fieldMinus}{\ensuremath{-}\xspace}

\newcommand{\elgGen}{\ensuremath{\mathsf{ElGamalInit}}\xspace}
\newcommand{\elgPrivKey}{\ensuremath{u}\xspace}
\newcommand{\elgPubKey}{\ensuremath{U}\xspace}
\newcommand{\elgEphemeralPrivKey}[1]{\ensuremath{v_{#1}}\xspace}
\newcommand{\elgEphemeralPubKey}[1]{\ensuremath{V_{#1}}\xspace}
\newcommand{\elgCiphertext}[1]{\ensuremath{W_{#1}}\xspace}
\newcommand{\elgGroup}{\ensuremath{G}\xspace}
\newcommand{\elgGroupGenerator}{\ensuremath{g}\xspace}
\newcommand{\elgGroupExponent}{\ensuremath{y}\xspace}

\newcommand{\timeBound}{\ensuremath{t}\xspace}
\newcommand{\timeBoundAlt}{\ensuremath{t'}\xspace}
\newcommand{\queryBound}{\ensuremath{q}\xspace}
\newcommand{\indcpaAdversary}{\ensuremath{A}\xspace}
\newcommand{\indcpaExperiment}[2]{\ensuremath{\mathsf{Expt}^{\textsc{cpa}\mbox{-}{#1}}_{#2}}\xspace}
\newcommand{\indcpaAdvantage}[1]{\ensuremath{\mathsf{Adv}^{\textsc{cpa}}_{#1}}\xspace}
\newcommand{\indcpaLROracle}[3]{\ensuremath{\encrypt{#1}(\textrm{LR}({#2},{#3},\indcpaBit))}\xspace}
\newcommand{\requesterAlgorithm}[2]{\ensuremath{\requester_{{#1}\mbox{-}{#2}}}\xspace}
\newcommand{\responderAdversary}[1]{\ensuremath{B_{#1}}\xspace}
\newcommand{\responderAdversaryState}[1]{\ensuremath{\phi_{#1}}\xspace}

\newcommand{\responderExperiment}[2]{\ensuremath{\mathsf{Expt}^{\textsc{pmt}\mbox{-}{#1}}_{#2}}\xspace}
\newcommand{\responderAdvantage}[1]{\ensuremath{\mathsf{Adv}^{\textsc{pmt}}_{#1}}\xspace}
\newcommand{\passwordChoiceBit}{\ensuremath{b}\xspace}
\newcommand{\requesterOutput}{\ensuremath{b'}\xspace}
\newcommand{\responderAdversaryBit}{\ensuremath{b''}\xspace}
\newcommand{\indcpaBit}{\ensuremath{\hat{b}}\xspace}
\newcommand{\indcpaAdversaryBit}{\ensuremath{\check{b}}\xspace}

\newcommand{\genericBucketCapacity}{\ensuremath{\chi}\xspace}
\newcommand{\cuckooFilter}{\ensuremath{\matrixNotation{X}}\xspace}
\newcommand{\cuckooFingerprintCtext}{\ensuremath{f}\xspace}
\newcommand{\cuckooFingerprintCtextMatrix}{\ensuremath{\matrixNotation{F}}\xspace}
\newcommand{\cuckooFingerprintFn}{\ensuremath{\mathsf{fprint}}\xspace}
\newcommand{\cuckooFingerprintRange}{\ensuremath{F}\xspace}
\newcommand{\cuckooHashFn}{\ensuremath{\mathsf{hash}}\xspace}
\newcommand{\cuckooNmbrBuckets}{\ensuremath{\beta}\xspace}

\newcommand{\cuckooBucketOne}{\ensuremath{i_1}\xspace}
\newcommand{\cuckooBucketTwo}{\ensuremath{i_2}\xspace}

\newcommand{\cuckooBucketCapacity}{\genericBucketCapacity}
\newcommand{\cuckooPMT}{\textsc{cuckoo-pmt}\xspace}
\newcommand{\bloomPMT}{\textsc{bloom-pmt}\xspace}

\newcommand{\zipfShape}[1]{\ensuremath{\lambda_{#1}}\xspace}
\newcommand{\zipfRank}{\ensuremath{k}\xspace}
\newcommand{\pwdDist}[1]{\ensuremath{D_{#1}}\xspace}
\newcommand{\passwordArray}[1]{\ensuremath{\mathsf{pwd}[{#1}]}\xspace}
\newcommand{\password}[1]{\ensuremath{\pi_{#1}}\xspace}

\newcommand{\leakedPwd}{\ensuremath{\password{\mathrm{leaked}}}\xspace}

\newcommand{\pmtSetSize}{\ensuremath{\ell}\xspace}
\newcommand{\pmtFPR}{\ensuremath{p}\xspace}
\newcommand{\pmtSet}{\ensuremath{Z}\xspace}
\NewDocumentCommand{\pmtSetElmt}{ g }{\ensuremath{e\IfNoValueF{#1}{_{#1}}}\xspace}
\newcommand{\requester}{\ensuremath{R}\xspace}
\newcommand{\responder}[1]{\ensuremath{S_{#1}}\xspace}
\NewDocumentCommand{\responderIdx}{ g }{\ensuremath{i\IfNoValueF{#1}{_{#1}}}\xspace}
\newcommand{\responderIdxAlt}{\ensuremath{i'}\xspace}
\NewDocumentCommand{\queryMatrix}{ g }{\ensuremath{\matrixNotation{Q}\IfNoValueF{#1}{_{#1}}}\xspace}
\newcommand{\resultMatrix}{\ensuremath{\matrixNotation{R}}\xspace}

\newcommand{\accountId}[1]{\ensuremath{a_{#1}}\xspace}
\NewDocumentCommand{\siteId}{ g g }{\ensuremath{s_{#1}\IfNoValueF{#2}{.{#2}}}\xspace}
\NewDocumentCommand{\siteIdAlt}{ g }{\ensuremath{s'\IfNoValueF{#1}{.{#1}}}\xspace}
\newcommand{\loginId}{\ensuremath{L}\xspace}
\newcommand{\loginIdx}{\ensuremath{l}\xspace}

\newcommand{\attackWidth}{\ensuremath{w}\xspace}
\newcommand{\nmbrResponders}[1]{\ensuremath{n_{#1}}\xspace}
\newcommand{\suspiciousArray}[1]{\ensuremath{\mathsf{susp}[{#1}]}\xspace}
\newcommand{\pwdRule}{\ensuremath{\textsc{susp}}\xspace}
\newcommand{\afaRule}{\ensuremath{\textsc{susp}^+}\xspace}

\newcommand{\mrHydeTerm}{Mr.\,Hyde\xspace}
\newcommand{\drJekyllTerm}{Dr.\,Jekyll\xspace}
\newcommand{\jekyllhydeTerm}{Jekyll-Hyde\xspace}
\newcommand{\mrHyde}{\ensuremath{\mathcal{H}}\xspace}
\newcommand{\drJekyll}{\ensuremath{\mathcal{J}}\xspace}

\newcommand{\credStuffer}{\ensuremath{\mathcal{C}}\xspace}
\newcommand{\accessedSet}{\ensuremath{\mathsf{accessed}}\xspace}
\newcommand{\detectedSet}{\ensuremath{\mathsf{detected}}\xspace}
\newcommand{\collectionFlag}{\ensuremath{\mathsf{collectionFlag}}\xspace}

\newcommand{\attemptedFlag}{\ensuremath{\mathsf{attemptedFlag}}\xspace}
\newcommand{\detectedFlag}{\ensuremath{\mathsf{detectedFlag}}\xspace}
\newcommand{\afaSet}[1]{\ensuremath{\mathsf{has2FA}_{#1}}\xspace}
\newcommand{\adsCollect}{\ensuremath{\mathsf{col}}\xspace}
\newcommand{\adsCount}{\ensuremath{\mathsf{cnt}}\xspace}
\newcommand{\adsLow}{\ensuremath{\mathit{lo}}\xspace}
\newcommand{\adsHigh}{\ensuremath{\mathit{hi}}\xspace}
\newcommand{\adsPR}[1]{\ensuremath{\mathsf{DR}_\mathsf{ads}^{#1}}\xspace}
\newcommand{\adsFPR}[1]{\ensuremath{\mathsf{FDR}_\mathsf{ads}^{#1}}\xspace}
\newcommand{\adsTPR}[1]{\ensuremath{\mathsf{TDR}_\mathsf{ads}^{#1}}\xspace}
\newcommand{\csd}{\ensuremath{\mathsf{csd}}\xspace}
\newcommand{\csdFPR}{\ensuremath{\mathsf{FDR}_{\csd}}\xspace}
\newcommand{\csdTPR}{\ensuremath{\mathsf{TDR}_{\csd}}\xspace}
\NewDocumentCommand{\adsDetected}{ g g }{\ensuremath{\mathsf{abnormal}^{#1}\IfNoValueF{#2}{({#2})}}\xspace}
\newcommand{\ads}{\ensuremath{\mathsf{ads}}\xspace}

\newcommand{\countingPhase}{\textit{counting phase}\xspace}
\newcommand{\countingPhases}{\textit{counting phases}\xspace}
\newcommand{\CountingPhase}{\textit{Counting phase}\xspace}

\newcommand{\collectingPhase}{\textit{collecting phase}\xspace}
\newcommand{\collectingPhases}{\textit{collecting phases}\xspace}
\newcommand{\CollectingPhase}{\textit{Collecting phase}\xspace}

\begin{document}

\title{Detecting Stuffing of a User's Credentials at Her Own Accounts}

\author{
  {\rm Ke Coby Wang}\\
  Department of Computer Science\\
  University of North Carolina at Chapel Hill
  \and
  {\rm Michael K.\ Reiter}\\
  Department of Computer Science\\
  University of North Carolina at Chapel Hill
}

\date{}
\maketitle
\thispagestyle{plain}
\pagestyle{plain}

\begin{abstract}
  We propose a framework by which websites can coordinate to detect
  credential stuffing on individual user accounts.  Our detection
  algorithm teases apart normal login behavior (involving password
  reuse, entering correct passwords into the wrong sites, etc.) from
  credential stuffing, by leveraging modern anomaly detection and
  carefully tracking suspicious logins.  Websites coordinate using a
  novel private membership-test protocol, thereby ensuring that
  information about passwords is not leaked; this protocol is highly
  scalable, partly due to its use of cuckoo filters, and is more
  secure than similarly scalable alternatives in an important measure
  that we define.  We use probabilistic model checking to estimate our
  credential-stuffing detection accuracy across a range of operating
  points.  These methods might be of independent interest for their
  novel application of formal methods to estimate the usability
  impacts of our design.  We show that even a minimal-infrastructure
  deployment of our framework should already support the combined
  login load experienced by the airline, hotel, retail, and consumer
  banking industries in the U.S.
\end{abstract}  

\section{Introduction}
\label{sec:intro}

In the past decade, massive numbers of website account credentials
have been compromised via password database breaches, phishing, and
keylogging.  According to a report by Shape
Security~\cite{shape2018:spill},\footnote{We recognize
  that this and other reports produced by companies that market
  credential-stuffing defenses might exaggerate the risks or costs of
  credential stuffing.  We are unaware of more objective sources with
  which to corroborate or refute their claims, however.}  2.3 billion
credentials were reported compromised in 2017 alone.  Such compromised
username-password pairs place those users' \textit{other} accounts in
jeopardy, since people tend to reuse their passwords across different
websites (e.g.,~\cite{das2014:tangled, gao2018:forgetting,
  ion2015:expert, pearman2017:habitat, shay2010:reuse}).  As such,
automatically attempting leaked username-password pairs at a wide
array of sites compromises vast numbers of accounts, a type of attack
termed \textit{credential stuffing}.  Credential stuffing is now a
dominant method of account takeover~\cite{shape2018:spill} and is
remarkably commonplace; e.g., Akamai estimates it observed 30 billion
credential-stuffing attempts in 2018~\cite{akamai2019:report}.
Credential stuffing imposes \textit{actual} losses estimated at
\$300M, \$400M, \$1.7B and \$6B on the hotel, airline, consumer
banking, and retail industries, respectively, per
year~\cite[\tblrefstatic{2}]{shape2018:spill}.  A survey of 538 IT
security practitioners who are responsible for the security of their
companies' websites revealed a total annualized cost of credential
stuffing across their organizations, excluding fraud, of \$3.85M,
owing to costs of prevention, detection, and remediation; downtime;
and customer churn~\cite[\tblsrefstatic{1}{3}]{ponemon2018:report}.

Despite the prominence of credential stuffing, users are remarkably
resistant to taking steps to defend themselves against it.  Thomas et
al.~\cite{thomas2017:credential} report that less than 3.1\% of users
who suffer account hijacks enable two-factor authentication after
recovering their accounts.  Users are similarly resistant to stopping
password reuse even despite specific warnings when doing so, leading
Golla et al.\ to conclude that ``notifications alone appear
insufficient in solving password
reuse''~\cite{golla2018:notifications}.  And though password managers
would seem to enable users to more easily avoid password reuse, users
are reluctant to adopt them.  In a 2016 survey of 1040 American
adults, only 12\% reported ever using password management software,
and only 3\% said this is the password technique they rely on
most~\cite{smith2017:americans}.  In a 2019 Google/Harris Poll survey
of 3000 U.S.\ adults, still only 24\% reported using a password
manager~\cite{google2019:security}.

Conceding that the factors that enable credential stuffing to succeed
today are likely to persist for the foreseeable future, we propose a
framework by which websites could cooperate to detect
active credential-stuffing attacks on a per-user
basis.  Developing such a framework is not straightforward, in part
because the exact behaviors that such a framework should detect are
difficult to define.  Anecdotally, users sometimes engage in behaviors
that might appear quite similar to a credential-stuffing attack, e.g.,
submitting the same small handful of passwords to multiple sites in
the course of logging into each, if she is unsure of which password
she set at which site.  A framework to detect credential stuffing on a
user's accounts will need to tease apart behaviors that the user might
normally undertake from actual credential abuse.

To do so, our framework leverages the following technique.  Anomaly
detection systems (ADS) now exist by which a site can differentiate
login attempts by the legitimate user from those by attackers, even
sophisticated ones, with moderately good accuracy, using features
\textit{other} than the password entered
(e.g.,~\cite{freeman2016:ads}).  A site in our framework leverages
this capability to track \textit{suspicious} login attempts locally,
namely abnormal attempts (per the ADS) using an incorrect password
or, for sites requiring second-factor authentication for abnormal
login attempts, such attempts for which the second-factor
authentication fails (even if the password is correct).  Then, our
framework enables a site (the \textit{\requesterTerm}) receiving a
login attempt that it deems abnormal to query other sites (the
\textit{\respondersTerm}) where this user has accounts, to determine
the number of them at which this same password has been submitted in
suspicious login attempts.  If this number is larger than a threshold,
then the \requesterTerm deems this login attempt to be credential
abuse---even if the password is correct.

Of course, such an approach raises concerns.  First, it risks false
detections, and lacking datasets of how legitimate users submit login
attempts---both correct and incorrect ones---across their many
accounts, the false detection rate seems hopeless to evaluate.
Second, measuring the true detection rate of this scheme would require
knowledge of how attackers conduct credential-stuffing attacks today
(again, we are aware of no such datasets) and, more importantly, how
attackers would respond if our framework were deployed by a collection
of websites.  Finally, since both the \requesterTerm's query and a
\responderTerm's suspicious-password set will contain sensitive
passwords, supporting these queries has the potential to leak
sensitive data to the \requesterTerm or \responderTerm.

We address these concerns as follows.  To estimate the true and false
detection rates of our design, we formulate experiments in the form of
Markov decision processes (MDPs), in which the adversary's choices in
the experiment determine a probability of the adversary achieving a
specified goal in our framework.  In the true-detection-rate MDP, the
adversary corresponds to a credential stuffer, and we leverage
probabilistic model checking to calculate the true detection rate for
the \textit{best} adversarial strategy, yielding what we believe is a
conservative estimate of our true detection rate in practice.  The
false-detection-rate MDP casts the ``adversary'' as the legitimate
user who knows \textit{how} she chooses her passwords (i.e., the
distribution) but who cannot recall which one she set at which
website.  Again, we allow the ``adversary'' (forgetful user) arbitrary
flexibility to submit login attempts, toward the ``goal'' of ensuring
that she will be detected as a credential stuffer when eventually
entering her correct password at a designated website.  We use
probabilistic model checking to find the \textit{best} strategy for
this ``adversary'', which we believe serves as a conservative estimate
of our false-detection rate.

To protect passwords while allowing queries to
suspicious-password sets, we develop a new private membership-test
(PMT) protocol that ensures that \respondersTerm do not learn the
\requesterTerm's query or the protocol result (no matter how they
misbehave) and that limits the information about the \responderTerm's
suspicious-password set that is leaked to the \requesterTerm.  We
quantify the suspicious-password-set leakage in terms of a measure we
call \textit{extraction complexity}, which informally is the number of
protocol runs a \responderTerm can tolerate before succumbing to an
\textit{offline} attack on its set.  We show that our protocol
improves over previous communication-efficient PMT protocols
substantially in this measure.

Finally, we present an implementation of our framework by which a
\requesterTerm leverages a \textit{\directoryTerm} to contact
\respondersTerm where a user holds accounts.  We evaluate performance
of our design in two privacy contexts, one where the \directoryTerm is
trusted to hide the \requesterTerm (i.e., where the user is currently
active) from \respondersTerm, and one where it is not and so the
\requesterTerm contacts the \directoryTerm using
Tor~\cite{dingledine2004:tor}.  We show that even with just one
\directoryTerm machine, various configurations of our design can
already support the typical login load experienced by the airline,
hotel, retail, and consumer banking industries in the U.S., combined.

To summarize, our contributions are as follows:
\begin{itemize}[nosep,leftmargin=1em,labelwidth=*,align=left]
\item We develop a framework by which websites can coordinate to
  detect active credential stuffing on a user's accounts, and we
  estimate the true and false detection rates of this
  algorithm using probabilistic model checking
  (\secsref{sec:stuffing}{sec:directory}).

\item We instantiate this framework with a new PMT protocol that
  ensures security against a malicious \requesterTerm or
  \responderTerm, including improving on other communication-efficient
  designs in a security measure (\textit{extraction complexity}) that
  is important in our context (\secref{sec:pmt}).

\item We report the performance of an implementation of our framework
  under two privacy configurations, in experiments ranging up to 256
  \respondersTerm (\secref{sec:performance}).  Our results indicate
  that even with minimal infrastructure, our design should scale to
  accommodate real login loads experienced by major sectors of the
  U.S. economy.
\end{itemize}

\section{Related Work}
\label{sec:related}

\paragraph*{Interfering with password reuse}
A user's reuse of the same passwords across her accounts is the
impetus for credential stuffing.  Password reuse is widespread
(e.g.,~\cite{brown2004:reuse, riley2006:reuse, shay2010:reuse,
  das2014:tangled, ion2015:expert, pearman2017:habitat,
  wang2018:domino}) and is very resistant to warnings to avoid
it---even reactive warnings triggered by a detected
reuse~\cite{golla2018:notifications}.  Most closely related to our
work is a recent proposal by which websites could coordinate to
actively interfere with a user's attempt to reuse the same or similar
passwords across those sites~\cite{wang2019:reuse}.  While we borrow
ingredients of this design (see \secref{sec:stuffing:assumptions} and
\secref{sec:directory}), our work targets credential stuffing
directly, \textit{without} interfering with a user's password reuse
across sites or assuming that it does not occur.  These different
goals lead to a fundamentally different design, requiring novel
underlying cryptographic protocols (\secref{sec:pmt}) and wholly novel
detection algorithms (\secref{sec:stuffing}).

\paragraph*{Detecting user selection of compromised or popular passwords}
It is now common (and recommended~\cite{nist2017:800-63B}) for sites
to cross-reference account passwords against known-leaked passwords,
either for their own users (e.g.,~\cite{collins2016:facebook}) or as a
service for others (e.g.,~\url{https://haveibeenpwned.com}).  Thomas
et al.~\cite{thomas2019:stuffing} and Li et al.~\cite{li2019:stuffing}
proposed improvements to these types of services that leak less
information to or from the service.  Pal et
al.~\cite{pal2019:similarity} developed personalized password strength
meters that warn users when selecting passwords similar to ones
previously compromised, particularly their own.  More distantly
related are services that track password popularity and enable a
website to detect if one of its users selects a popular
password~\cite{schechter2010:popularity, naor2019:popular}.  In
contrast to these works, our work detects credential stuffing of a
user's password before its compromise is reported (which today takes
an \textit{average} of 15 months~\cite{shape2018:spill}) and
irrespective of its popularity.

\paragraph*{Discovering compromised accounts}
Several techniques have been proposed to discover compromised
accounts.  For example, to detect the breach of its password
database, a site might list several site-generated \textit{honey
  passwords} in the database alongside the valid password for each
account~\cite{bojinov2010:kamouflage, juels2013:honeywords,
  erguler2016:flatness}.  Any submission of a honey password in a
login attempt then discloses the breach of the password database.
Similarly, \textit{honey accounts} for a user can be set up at
websites where she does not have an account, specifically for
detecting any attempt to log into them with the password of one of her
actual accounts~\cite{deblasio2017:tripwire}.  Both of these can be
used in conjunction with our framework but do not supersede it, as
attackers holding a user's correct password for one account (e.g., as
obtained from phishing the user) can use it to attempt logins at the
user's actual accounts at other websites without either of these
techniques detecting it.  This is precisely the type of attack that we
seek to detect here.

\paragraph*{Detecting guessing attacks}
Herley and Schechter~\cite{herley2019:scale} provide an algorithm by
which a large-volume website can estimate the likelihood that a login
is part of a guessing attack, based on the assumption that these
malicious logins are a small fraction of total logins.  They point out
that this assumption ``is of course not true for an attacker
exploiting password re-use or other non-guessing approach'', which is
our interest here.  Schechter et al.~\cite{schechter2019:stopguessing}
suggest features for distinguishing benign login attempts from
guessing attacks, though these features should not be characteristic
of credential stuffing.

\section{Detecting Credential Stuffing}
\label{sec:stuffing}

In this section, we present our framework to detect credential
stuffing on a user's accounts.  Our framework detects credential
stuffing on a per-user basis, and so is agnostic to whether the user
is the only one being subjected to credential stuffing (e.g., after
one of her passwords was phished) or whether she is one of many (e.g.,
after a password database breach).

\subsection{Assumptions}
\label{sec:stuffing:assumptions}

\paragraph*{Account identifiers}
Our framework assumes the ability to associate the accounts of the
same user across different websites---or more specifically, to do so
at least as well as a credential-stuffing attacker could.  When user
accounts are tied to email addresses confirmed during account
creation, this can generally be done, even despite the email-address
variations for a single email account that are supported by email
service providers~\cite{wang2019:reuse}.  As such, we will generally
refer to a user's account identifier \accountId{} as being the same
across multiple websites.  We also assume the ability of one website
at which a user has an account to contact (perhaps anonymously) other
websites where the user has an account.  In our design, this ability
is supported using a logical \textit{\directoryTerm} service, as will
be detailed in \secref{sec:directory}.

\paragraph*{Password management}
\textit{Nothing in our framework requires that a site store passwords
  in the clear or in a reversible fashion}, and most existing best
practices (e.g., using expensive hash functions to compute password
hashes~\cite{spafford1992:opus, biryukov2016:argon2}) can be applied
within our scheme.  In particular, while in
\secref{sec:stuffing:algorithm} we will use
\siteId{}{\passwordArray{\accountId{}}} to denote the correct password
for account \accountId{} at site \siteId{}, we stress that \siteId{}
need not (and should not) store this password explicitly.

We do make one concession in password-management best practices,
however, similarly to some other defenses
(e.g.,~\cite{gentry2006:resilient, wang2019:reuse}).  In our
algorithm, each site \siteId{} will maintain a set
\siteId{}{\suspiciousArray{\accountId{}}} of hashes of passwords
submitted in \textit{suspicious} login attempts on account
\accountId{}.  (This will be detailed in
\secref{sec:stuffing:algorithm}.)  For one site \siteId{} (the
\textit{\requesterTerm}) to query whether a hash value \pmtSetElmt is
present in \siteIdAlt{\suspiciousArray{\accountId{}}} at another site
\siteIdAlt (the \textit{\responderTerm}), it is necessary that any
value that \siteId{} uses to salt \pmtSetElmt be the same as the value
that \siteIdAlt uses to salt the elements of
\siteIdAlt{\suspiciousArray{\accountId{}}}.  To do so, the salt
corresponding to identifier \accountId{} could be generated
deterministically from \accountId{} or generated randomly and
distributed to a site \siteId{} by the \directoryTerm when \siteId{}
registers as a \responderTerm for \accountId{} (see
\secref{sec:directory:dos}).  Below, we elide these details
and simply write ``$\password{} \in
\siteIdAlt{\suspiciousArray{\accountId{}}}$'' to denote the membership
of a hash \pmtSetElmt in the set
\siteIdAlt{\suspiciousArray{\accountId{}}}, where \pmtSetElmt uniquely
identifies the password \password{} from which it was computed but
that also incorporates this salting.

While less secure than per-site, per-account salting, an attacker's
precomputation (before breaching \siteIdAlt) to aid his search for the
passwords whose hashes are contained in
\siteIdAlt{\suspiciousArray{\accountId{}}} would need to be repeated
per identifier \accountId{}.  Moreover, the need for consistent
salting per account \accountId{} across sites pertains only to the
storage and querying of \siteId{}{\suspiciousArray{\accountId{}}}
sets, and not to the hashing of
\siteId{}{\passwordArray{\accountId{}}} for the purposes of checking
the correctness of a password entered in a login attempt.  That is,
login attempts can be checked against a hash of
\siteId{}{\passwordArray{\accountId{}}} that is salted with a
per-site, per-account salt value.  As we will discuss in
\secref{sec:stuffing:algorithm}, while a hash of
\siteId{}{\passwordArray{\accountId{}}} might end up in
\siteId{}{\suspiciousArray{\accountId{}}} in certain cases, we do not
expect this to be the common case.  Consequently, we believe that
resorting to consistent salting across sites for the management of
\siteId{}{\suspiciousArray{\accountId{}}} sets per account
\accountId{} is a small concession to make.

\paragraph*{Anomaly detection}
The key assumption we make about sites that participate in our design
is that each one conducts anomaly detection on the login attempts to
its site, based on locally available features such as the time, client
IP address, useragent string, etc.  For a 10\% false-detection rate,
Freeman et al.~\cite{freeman2016:ads} report a $99\%$ true detection
rate for attacker logins to an account \accountId{} from the country
from which the user for account \accountId{} normally logs in (a
so-called \textit{researching} attacker).  Also for a 10\%
false-detection rate, they report a 74\% true-detection rate for the
most advanced attackers they consider, who also initiate login
attempts with the same useragent string as the legitimate user (a
so-called \textit{phishing} attacker).

Here we treat each site's anomaly-detection system (ADS) as a
block-box that takes as input a group of login features and classifies
the login as either normal or abnormal.  We assume that the ADS can be
parameterized (e.g., with a threshold) to tune its true- and
false-detection rates, where a ``detection'' means an abnormal
classification.  Our credential-stuffing detection algorithm will
leverage two parameter settings (e.g., thresholds)
for the ADS, yielding for each login
attempt \loginId an ADS output of the form $\langle
\adsDetected{\adsCollect}{\loginId},
\adsDetected{\adsCount}{\loginId}\rangle$, a pair of boolean values.
The ``\adsCollect'' and ``\adsCount'' qualifiers denote
  the ``collecting'' and ``counting'' phases of our algorithm,
  respectively, which will be explained in
  \secref{sec:stuffing:algorithm}.  The flexibility provided by
  allowing different ADS parameter settings in the two phases of our
  algorithm is important to permit algorithm tuning; see
  \secref{sec:stuffing:eff}.
Because \adsDetected{\adsCollect}{\loginId} and
\adsDetected{\adsCount}{\loginId} are derived from a common ADS, these
indicators are not independent.  We denote their (false-, true-)
detection rates as $(\adsFPR{\adsCollect}, \adsTPR{\adsCollect})$ and
$(\adsFPR{\adsCount}, \adsTPR{\adsCount})$, respectively, and assume
that either $\adsTPR{\adsCollect} \ge \adsTPR{\adsCount}$ and
$\adsFPR{\adsCollect} \ge \adsFPR{\adsCount}$ or $\adsTPR{\adsCount}
\ge \adsTPR{\adsCollect}$ and $\adsFPR{\adsCount} \ge
\adsFPR{\adsCollect}$, as otherwise one parameter setting would be
strictly better than the other.  When the login \loginId is clear
from context, we will generally elide it and simply denote the ADS
output as $\langle \adsDetected{\adsCollect},
\adsDetected{\adsCount}\rangle$.

\paragraph*{Threat model}
We specify our credential-stuffing detection algorithm in
\secref{sec:stuffing:algorithm} assuming sites that cooperate to
detect credential stuffing against a user's accounts at those sites.
  That is, the attacker can submit login attempts at participating
  sites, possibly using passwords it stole, but cannot participate as
  a site in our framework.
In \secsref{sec:directory}{sec:pmt}, however, we address the potential
for \textit{malicious} participating sites.  In particular, we address
user login privacy against participating sites in
\secref{sec:directory:privacy} and the risk of denial-of-service
attacks by participating sites in \secref{sec:directory:dos}.
Finally, we address user account security despite misbehavior of
participating sites in \secref{sec:pmt}.  We make no effort to address
sites that misbehave so as to reduce true detections of our framework,
since these sites could equally well do so by simply not
participating.  The \directoryTerm, introduced in
\secref{sec:directory}, is trusted to not conduct denials-of-service
and to help defend against them (see \secref{sec:directory:dos}), but
is not trusted for security of sites' user accounts.

\subsection{Algorithm}
\label{sec:stuffing:algorithm}

To detect credential stuffing, each website processes each login
attempt in two \textit{phases}, called the \collectingPhase and the
\countingPhase.  To support these phases, each site \siteId{}
maintains a set \siteId{}{\suspiciousArray{\accountId{}}} of (salted
hashes of) passwords used in ``suspicious'' login attempts to account
\accountId{} at site \siteId{}, as discussed below.  Site \siteId{}
assembles \siteId{}{\suspiciousArray{\accountId{}}} in the
\collectingPhases of local login attempts to account \accountId{}, and
queries the \siteIdAlt{}{\suspiciousArray{\accountId{}}} sets at other
sites \siteIdAlt in the \countingPhases of logins to \accountId{} at
\siteId{}.  These queries are performed using \textit{private
  membership tests} (PMTs), which hide \siteId{}'s query when acting
as a \requesterTerm, and hide the contents of
\siteId{}{\suspiciousArray{\accountId{}}} when acting as a
\responderTerm in the protocol.  We defer details of the PMT protocol
to \secref{sec:pmt}.

Our algorithm begins when a site receives a local login attempt.  The
site submits the login features for classification by its ADS,
yielding a classification $\langle \adsDetected{\adsCollect},
\adsDetected{\adsCount}\rangle$.  The site then performs the
\collectingPhase and the \countingPhase, in that order.

\medskip

\paragraph*{\CollectingPhase} In the \collectingPhase of a login attempt
to account \accountId{} at site \siteId{}, if
$\adsDetected{\adsCollect} = \boolTrue$, then \siteId{} applies one of
the following two rules as appropriate,
where \password{} is the submitted password:

\smallskip

\noindent
\begin{tabular}{p{0.125\columnwidth}@{\hspace{1em}}p{0.8\columnwidth}}
  \pwdRule & If \siteId{} does not support second-factor
  authentication for \accountId{},
  then \siteId{} adds \password{} to
  \siteId{}{\suspiciousArray{\accountId{}}} if the password is
  incorrect (i.e., $\password{} \neq
  \siteId{}{\passwordArray{\accountId{}}}$).
\end{tabular}
  
\noindent
\begin{tabular}{p{0.125\columnwidth}@{\hspace{1em}}p{0.8\columnwidth}}
  \afaRule & If \siteId{} supports second-factor authentication
  for \accountId{}, then
  \siteId{} adds \password{} to
  \siteId{}{\suspiciousArray{\accountId{}}}.  If $\password{} =
  \siteId{}{\passwordArray{\accountId{}}}$, then it is subsequently
  removed from \siteId{}{\suspiciousArray{\accountId{}}} only once a
  second-factor challenge issued by \siteId{} to the owner of account
  \accountId{} is completed successfully.
\end{tabular}

\medskip

\paragraph*{\CountingPhase} In the \countingPhase of a login attempt to
account \accountId{} at site \siteId{}, if $\adsDetected{\adsCount} =
\boolTrue$ and if the submitted password \password{} is correct for
account \accountId{} (i.e., $\password{} =
\siteId{}{\passwordArray{\accountId{}}}$), then \siteId{} performs the
role of the \requesterTerm using password \password{} in PMTs.  In
each of these PMTs, another site \siteIdAlt where account \accountId{}
exists performs the role of the \responderTerm with set
\siteIdAlt{}{\suspiciousArray{\accountId{}}}; i.e., \siteIdAlt
interacts with \siteId{} to allow \siteId{} to learn whether
$\password{} \in \siteIdAlt{\suspiciousArray{\accountId{}}}$.  (In
\secref{sec:directory} we will discuss how \siteId{} contacts each
\siteIdAlt.)  Site \siteId{} then detects credential stuffing if
$\setSize{\cset{\big}{\siteIdAlt}{\siteIdAlt \neq \siteId{} \wedge
    \password{} \in \siteIdAlt{\suspiciousArray{\accountId{}}}}} \ge
\attackWidth$, for a specified \textit{attack width} \attackWidth.

\bigskip

Again, the \collectingPhase is performed first, or more precisely,
\siteId{} begins its \countingPhase only once any addition to
\siteId{}{\suspiciousArray{\accountId{}}} in the preceding
\collectingPhase is complete.  Similarly, upon receiving a PMT query
from \siteId{}, site \siteIdAlt defers responding until any additions
to \siteIdAlt{\suspiciousArray{\accountId{}}} in ongoing
\collectingPhases for the same account \accountId{} are completed
locally.

Note that if \pwdRule would add \password{} to
\siteId{}{\suspiciousArray{\accountId{}}}, then \afaRule would, as
well.  In addition, \afaRule allows even
\siteId{}{\passwordArray{\accountId{}}} to be added to
\siteId{}{\suspiciousArray{\accountId{}}} if it is submitted in a
login that is deemed abnormal but the resulting second-factor
challenge is not completed successfully.
We do not expect this to be the norm, however: A rough
estimate assuming $\adsFPR{\adsCollect} = 0.10$ and a second-factor
  failure rate by the correct user of $0.12$ (e.g.,
  see~\cite{doerfler2019:evaluating}) is that a legitimate login
  attempt with the correct password at a site \siteId{} supporting
  \afaRule leaves that password in
  \siteId{}{\suspiciousArray{\accountId{}}} with probability only
  $0.012$.

When evaluating \afaRule, we
assume that an attacker is unable to complete the second-factor
challenge (which is generally true~\cite{doerfler2019:evaluating}),
but that for usability purposes, the site invokes the second-factor
challenge only on logins for which $\adsDetected{\adsCollect} =
\boolTrue$.  Some sites \siteId{} can maintain
\siteId{}{\suspiciousArray{\accountId{}}} according to \pwdRule while
others use \afaRule.  In our evaluations in
\secref{sec:stuffing:eff}, we will consider the impact of
different balances of sites that use \pwdRule versus \afaRule.

\pwdRule and \afaRule indicate when \siteId{} should \textit{add} a
password to \siteId{}{\suspiciousArray{\accountId{}}}, but not when
\siteId{} should \textit{remove} a password from it.  One approach
would be for \siteId{} to remove a password from
\siteId{}{\suspiciousArray{\accountId{}}} if that password is not used
in an attempted login to account \accountId{} for a specified
\textit{expiration time}.  Provided that \siteId{}'s login interface
rate-limits login attempts on \accountId{} (as is
recommended~\cite{nist2017:800-63B}), an upper bound on the capacity
of \siteId{}'s set \siteId{}{\suspiciousArray{\accountId{}}} can be
ensured.  For example, if \siteId{} permits 100 failed login attempts
on a single account in any 30-day
period~\cite[\secrefstatic{8.2.3}]{nist2013:800-63-2}, and if each
password expires from \siteId{}{\suspiciousArray{\accountId{}}} in 30
days since its last use in a login attempt, then
\setSize{\siteId{}{\suspiciousArray{\accountId{}}}} will never exceed
101.  Such a delay should allow ample time for our framework to detect
even a moderately aggressive credential-stuffing attack, or conversely
should dramatically slow a credential-stuffing attack if it is to go
undetected.

Finally, when adding a password \password{} to
\siteId{}{\suspiciousArray{\accountId{}}}, \siteId{} may reduce
\password{} to a canonical form, e.g., converting capital letters at
selected positions to lowercase, or converting a specific digit to a
digit wildcard.  Provided that the rules for this canonicalization are
employed by both \requestersTerm and \respondersTerm, our framework
can then detect stuffing of some passwords \textit{similar} to that
chosen by this user at another site (``credential tweaking'').  Of
course, \siteId{} could also explicitly add selected passwords similar
to \password{}, but at the cost of increasing
\setSize{\siteId{}{\suspiciousArray{\accountId{}}}}.  We do not
consider these extensions further here.

\subsection{Effectiveness}
\label{sec:stuffing:eff}

We now estimate the false- and true-detection rates for the algorithm
in \secref{sec:stuffing:algorithm} across a range of parameter
settings.  In doing so, we seek to demonstrate that our algorithm can
be effective in detecting credential stuffing without imposing
significantly on legitimate users.  We stop short of recommending a
specific course of action when a site detects credential stuffing via
our algorithm, though we will discuss alternatives at the end of this
section.

Evaluating false- and true-detection rates empirically would require
datasets that are unavailable to us.  To evaluate false detections
empirically, we would presumably need datasets that shed light on how
users both set passwords across websites and then try (and sometimes
fail) to log into websites using them.  To evaluate true detections,
we would need datasets of recorded credential-stuffing campaigns,
along with the correct and guessed passwords and (for sites supporting
\afaRule) the results of second-factor challenges.  

In the absence of any foreseeable way of obtaining such datasets, we
instead perform an evaluation using probabilistic model checking.  The
tool we used to perform probabilistic model checking is
Prism~\cite{kwiatkowska2011:prism}, which supports automated analysis
of Markov decision processes (MDP).  Each MDP we design models an
actor interacting with a specific account \accountId, who is either
the legitimate user of \accountId{} or an attacker, across multiple
websites.  To do so, we specify the actor as a set of \textit{states}
and possible \textit{actions}.  When in a state, the actor can choose
from among these actions nondeterministically; the chosen action
determines a probability distribution on the state to which the actor
then transitions.  These state transitions satisfy the \textit{Markov
  property}: informally, the probability of next transitioning to a
specific state depends only on the current state and the actor's
chosen action.  Prism exhaustively searches all decisions an actor can
make to maximize the probability of the actor succeeding in its goal.

Below, we assume that the legitimate user's password choices across
websites are represented by a probability distribution
\pwdDist{\accountId{}}; i.e., $\pwdDist{\accountId{}}(\password{})$ is
the probability with which the user selects \password{} as its
password for any given site.  We abuse notation slightly and also use
\pwdDist{\accountId{}} to denote the set of passwords with non-zero
probability.  For example, we write $\password{} \in
\pwdDist{\accountId{}}$ to indicate that
$\pwdDist{\accountId{}}(\password{}) > 0$;
\setSize{\pwdDist{\accountId{}}} to denote the number of passwords
\password{} for which $\password{} \in \pwdDist{\accountId{}}$; and
$\password{} \getsr \pwdDist{\accountId{}}$ to denote the selection of
a password from distribution \pwdDist{\accountId{}} and its assignment
to \password{}.  As some prior works~\cite{blocki2018:zipf,
  wang2017:zipf}, we model \pwdDist{\accountId{}} as a Zipf
distribution with parameter $\zipfShape{\accountId{}} \ge 0$, so that
the user chooses her \zipfRank-th most probable password ($1 \le
\zipfRank \le \setSize{\pwdDist{\accountId{}}}$) independently with
probability
$(1/\zipfRank^{\zipfShape{\accountId{}}})/(1/1^{\zipfShape{\accountId{}}}
+ 1/2^{\zipfShape{\accountId{}}} + \ldots +
1/\setSize{\pwdDist{\accountId{}}}^{\zipfShape{\accountId{}}})$.

The MDPs below need to synthetically model the distribution of
$\langle \adsDetected{\adsCollect}$, $\adsDetected{\adsCount}\rangle$
pairs for login attempts or sessions thereof, similar to their
distribution in practice (notably, lacking independence).  To do so
for specified detection rates \adsPR{\adsCollect} and
\adsPR{\adsCount}, let $\adsHigh \in \{\adsCollect, \adsCount\}$ and
$\adsLow \in \{\adsCollect, \adsCount\}$ be such that \adsPR{\adsHigh}
and \adsPR{\adsLow} are the larger and smaller of \adsPR{\adsCollect}
and \adsPR{\adsCount}, respectively.  Then, we let
\begin{align*}
  \prob{\adsDetected{\adsHigh} = \boolTrue} & = \adsPR{\adsHigh} \\
  \cprob{\Big}{\adsDetected{\adsLow} = \boolTrue}{\adsDetected{\adsHigh} = \boolTrue}
  & = \adsPR{\adsLow}/\adsPR{\adsHigh} \\
    \cprob{\Big}{\adsDetected{\adsLow} = \boolTrue}{\adsDetected{\adsHigh} = \boolFalse}
  & = 0
\end{align*}
We denote selection of $\langle \adsDetected{\adsCollect}$,
$\adsDetected{\adsCount}\rangle$ according to this distribution in the
experiment descriptions below by the notation $\langle
\adsDetected{\adsCollect}$, $\adsDetected{\adsCount}\rangle$ $\getsr$
$\ads(\adsPR{\adsCollect}$, $\adsPR{\adsCount})$.

\subsubsection{Estimating the false detection rate}
\label{sec:stuffing:eff:fdr}
False detections can arise in our framework; indeed, even the entry of
the \textit{correct} password for an account at a website by the
legitimate user might trigger a credential-stuffing detection if the
user erroneously submitted the same password to other websites (that
use \pwdRule), or even if correctly but without completing a
second-factor challenge from those sites (that use \afaRule).  Here we
leverage probabilistic model checking to conservatively estimate the
probability with which a user induces a false detection.

We express the process by which a user might do so as a MDP
in which the legitimate user is represented by
two parties, to whom we refer here as ``\drJekyllTerm'' (\drJekyll)
and ``\mrHydeTerm'' (\mrHyde).\footnote{In their namesake
  novella~\cite{stevenson2015:strange}, \mrHydeTerm and \drJekyllTerm
  are the evil and good personae, respectively, of the same
  person.}\ Informally, the user's \mrHyde persona knows the
distribution from which the user previously set passwords at various
websites, but does not remember which password the user set at which
site.  \mrHyde attempts a number of logins before turning control over
to the \drJekyll persona, who is challenged to log into another
website, for which he does remember the password.  Still, \drJekyll's
entry of the correct password might be detected as possible credential
stuffing, depending on the actions of \mrHyde before him.  In a
\jekyllhydeTerm experiment, then, we say that \mrHyde wins (and
\drJekyll loses) if \drJekyll's login attempt is (falsely) detected as
credential stuffing, and otherwise \mrHyde loses (and \drJekyll wins).

Since both \drJekyll and \mrHyde represent the legitimate user, we
assume both can complete any second-factor challenges that a website
issues.  Under this assumption, there is no difference between
\pwdRule and \afaRule, and so we do not differentiate sites supporting
\afaRule from those supporting \pwdRule in \jekyllhydeTerm
experiments.  Since users who forget their passwords presumably tend
to attempt multiple password guesses in the same login session (and so
from the same platform and location) until finding the right one,
\mrHyde is classified once by the ADS at site \siteId{}
for all login attempts there.  Finally, we forbid \mrHyde from
attempting logins at a site \siteId{} after he has already submitted
the correct password to \siteId{} (see step~\ref{exp:j-h:hyde}) to
preclude him from trivially winning.  After all, once \mrHyde has
``recalled'' the correct password for \siteId{}, he could artificially
add extra passwords to \siteId{}{\suspiciousArray{\accountId{}}} by
``attempting'' logins with them, thereby unreasonably inflating his
chances of winning.

More precisely, a \jekyllhydeTerm experiment takes as input an account
identifier \accountId{}, the distribution \pwdDist{\accountId{}}, a
number of \respondersTerm \nmbrResponders{\accountId{}}, an integer
\attackWidth, and probabilities \adsFPR{\adsCollect} and
\adsFPR{\adsCount}, and proceeds as follows:
\begin{enumerate}[nosep,labelsep=3pt,labelwidth=*,leftmargin=1.5em,align=left]
\item Sites \siteId{\responderIdx}, for $0 \le \responderIdx \le
  \nmbrResponders{\accountId{}}$, are initialized as follows:
  \begin{itemize}[nosep,leftmargin=1em,labelwidth=*,align=left]
  \item A password
    $\siteId{\responderIdx}{\passwordArray{\accountId{}}} \getsr
    \pwdDist{\accountId{}}$ is selected independently for account
    \accountId{} at website \siteId{\responderIdx}.
  \item The suspicious password set is cleared:
    $\siteId{\responderIdx}{\suspiciousArray{\accountId{}}} \gets
    \emptyset$.
  \item To model the classification of \mrHyde by the ADS at
    \siteId{\responderIdx}, a boolean
    \siteId{\responderIdx}{\collectionFlag} is set to
    \adsDetected{\adsCollect} where $\langle
    \adsDetected{\adsCollect}$, $\adsDetected{\adsCount}\rangle$
    $\getsr$ $\ads(\adsFPR{\adsCollect}$,
    $\adsFPR{\adsCount})$.
  \end{itemize}
\item \mrHyde is given the experiment inputs and performs login
  attempts on \accountId{} at any of $\siteId{1}, \ldots,
  \siteId{\nmbrResponders{\accountId{}}}$, provided that if \mrHyde
  submits the correct password
  \siteId{\responderIdx}{\passwordArray{\accountId{}}} in a login
  attempt at \siteId{\responderIdx}, then this is \mrHyde's last login
  attempt at \siteId{\responderIdx}.  Each incorrect login attempt at
  \siteId{\responderIdx} adds the attempted password to
  \siteId{\responderIdx}{\suspiciousArray{\accountId{}}} if and only
  if \siteId{\responderIdx}{\collectionFlag} is \boolTrue.
  \label{exp:j-h:hyde}
\item Once \mrHyde is done, \drJekyll logs into \siteId{0} using the
  correct password \siteId{0}{\passwordArray{\accountId{}}}.  If
  $\adsDetected{\adsCount} = \boolTrue$ for $\langle
  \adsDetected{\adsCollect}$, $\adsDetected{\adsCount}\rangle$
  $\getsr$ $\ads(\adsFPR{\adsCollect}$, $\adsFPR{\adsCount})$
  and if
  $\setSize{\cset{\big}{\siteId{\responderIdx}}{\siteId{0}{\passwordArray{\accountId{}}}
      \in \siteId{\responderIdx}{\suspiciousArray{\accountId{}}}
      \wedge 1 \le \responderIdx \le \nmbrResponders{\accountId{}}}}
  \ge \attackWidth$, then \mrHyde wins.  Otherwise, \drJekyll wins.
\end{enumerate}
We define \csdFPR
(``\csd'' denotes ``credential-stuffing detection'')
to be the probability with which \mrHyde wins and so
\drJekyll loses, under an optimal strategy for \mrHyde.  We
believe that \csdFPR is a very conservative estimate on the false
detection rate of our framework in practice, in that it reflects the
worst case behavior (in terms of usability) of \mrHyde.  Moreover, by
testing
$\setSize{\cset{\big}{\siteId{\responderIdx}}{\siteId{0}{\passwordArray{\accountId{}}}
    \in \siteId{\responderIdx}{\suspiciousArray{\accountId{}}} \wedge
    1 \le \responderIdx \le \nmbrResponders{\accountId{}}}} \ge
\attackWidth$ only for \drJekyll, i.e., after \mrHyde has filled the
suspicious password sets of sites as much as possible, our
false-detection estimates are even more conservative (notably, ignoring
\mrHyde's logins where he went undetected).

Consider an example with $\nmbrResponders{\accountId{}} = 2$ sites
(\siteId{1} and \siteId{2}), $\setSize{\pwdDist{\accountId{}}} = 2$
passwords (\password{1} and \password{2}), and $\attackWidth = 1$.
Consider the following sequence of choices by \mrHyde: First, \mrHyde
attempts \password{1} at \siteId{2}, which is correct with probability
$\pwdDist{\accountId{}}(\password{1})$ and incorrect with probability
$1 - \pwdDist{\accountId{}}(\password{1}) =
\pwdDist{\accountId{}}(\password{2})$.  In addition, this login
attempt is detected as abnormal if $\siteId{2}{\collectionFlag} =
\boolTrue$, which occurs with probability \adsFPR{\adsCollect}.
Suppose that \password{1} is incorrect and
$\siteId{2}{\collectionFlag} = \boolTrue$, and so
\siteId{2}{\suspiciousArray{\accountId{}}} now contains \password{1}.
Then suppose \mrHyde attempts password \password{2} at \siteId{1}, but
that this guess is incorrect and \mrHyde's login is not deemed
abnormal.  Finally, suppose that \drJekyll gains control and submits
the correct password \password{1} to \siteId{0}, but that his attempt
is detected as abnormal (with probability \adsFPR{\adsCount}).
Because \siteId{0}{ \passwordArray{\accountId{}}} (= \password{1}) is
in at least $\attackWidth = 1$ of the suspicious sets at other sites,
i.e., \siteId{2}{\suspiciousArray{\accountId{}}}, \mrHyde wins.
\mrHyde's choices induce this sequence of events with probability
$\left[\pwdDist{\accountId{}}(\password{2}) \cdot
  \adsFPR{\adsCollect}\right]$ $\cdot$
$\left[\pwdDist{\accountId{}}(\password{1}) \cdot
  (1-\adsFPR{\adsCollect})\right]$ $\cdot$
$\left[\pwdDist{\accountId{}}(\password{1}) \cdot
  \adsFPR{\adsCount}\right]$, and \csdFPR is computed by exhaustively
considering all possible choices by \mrHyde and all event sequences.

\subsubsection{Estimating the true detection rate}
\label{sec:stuffing:eff:tdr}
To evaluate the true-detection rate of our credential-stuffing
algorithm, we use a different type of MDP, in which a
credential-stuffing attacker \credStuffer is given a ``leaked''
password $\leakedPwd \getsr \pwdDist{\accountId{}}$ and allowed to
attempt logins using it at sites $\siteId{1}, \ldots,
\siteId{\nmbrResponders{\accountId{}}}$ where \accountId{} has
accounts.  The attacker knows which sites have second-factor
authentication enabled for abnormal logins
to \accountId{}, as specified by a set
$\afaSet{\accountId{}} \subseteq \{\siteId{1}, \ldots,
\siteId{\nmbrResponders{\accountId{}}}\}$.  Each site
$\siteId{\responderIdx} \in \afaSet{\accountId{}}$ therefore uses \afaRule to manage
\siteId{\responderIdx}{\suspiciousArray{\accountId{}}}, and we assume
that \credStuffer cannot pass a second-factor challenge for
\accountId{}.  Other sites use \pwdRule.  We also allow the attacker
knowledge of the true-detection rates \adsTPR{\adsCollect} and
\adsTPR{\adsCount} of sites' ADS.

As such, our true-detection experiment takes as input an account
identifier \accountId{}, the distribution \pwdDist{\accountId{}}, the
number of \respondersTerm \nmbrResponders{\accountId{}}, the set
\afaSet{\accountId{}}, an integer \attackWidth, and probabilities
\adsTPR{\adsCollect} and \adsTPR{\adsCount}, and proceeds as follows:
\begin{enumerate}[nosep,labelsep=3pt,labelwidth=*,leftmargin=1.5em,align=left]
\item Sites \siteId{\responderIdx}, for $1 \le \responderIdx \le
  \nmbrResponders{\accountId{}}$, are initialized as follows:
  \begin{itemize}[nosep,leftmargin=1em,labelwidth=*,align=left]
  \item A password
    $\siteId{\responderIdx}{\passwordArray{\accountId{}}} \getsr
    \pwdDist{\accountId{}}$ is selected independently for account
    \accountId{} at website \siteId{\responderIdx}.
  \item The suspicious password set is cleared:
    $\siteId{\responderIdx}{\suspiciousArray{\accountId{}}} \gets
    \emptyset$.
  \item A boolean \siteId{\responderIdx}{\attemptedFlag} is
    initialized to \boolFalse.
  \end{itemize}
\item \credStuffer is given \accountId{}, $\leakedPwd \getsr
  \pwdDist{\accountId{}}$, \afaSet{\accountId{}}, \adsTPR{\adsCollect},
  \adsTPR{\adsCount}, \attackWidth, and the opportunity to perform one
  login attempt using \leakedPwd at each of $\siteId{1}, \ldots,
  \siteId{\nmbrResponders{\accountId{}}}$.  On \credStuffer's
  \loginIdx-th login attempt ($\loginIdx = 1, 2, \ldots$), let
  \siteId{\responderIdx{\loginIdx}} denote the site at which this
  attempt occurs.  Then:
  \begin{itemize}[nosep,leftmargin=1em,labelwidth=*,align=left]
  \item Set $\langle \adsDetected{\adsCollect}$,
    $\adsDetected{\adsCount}\rangle$ $\getsr$
    $\ads(\adsTPR{\adsCollect}$, $\adsTPR{\adsCount})$ and
    $\siteId{\responderIdx{\loginIdx}}{\collectionFlag} \gets
    \adsDetected{\adsCollect}$.
  \item \siteId{\responderIdx{\loginIdx}} adds \leakedPwd to
    \siteId{\responderIdx{\loginIdx}}{\suspiciousArray{\accountId{}}}
    if $\adsDetected{\adsCollect} = \boolTrue$ and either $\leakedPwd
    \neq
    \siteId{\responderIdx{\loginIdx}}{\passwordArray{\accountId{}}}$
    (per \pwdRule) or, if $\siteId{\responderIdx{\loginIdx}} \in
    \afaSet{\accountId{}}$, even if $\leakedPwd =
    \siteId{\responderIdx{\loginIdx}}{\passwordArray{\accountId{}}}$
    (per \afaRule, since \credStuffer cannot pass 
    second-factor authentication).
  \item If $\loginIdx > \attackWidth$, then
    \begin{itemize}[nosep,leftmargin=1em,labelwidth=*,align=left]
    \item $\siteId{\responderIdx{\loginIdx}}{\attemptedFlag} \gets \boolTrue$
    \item $\siteId{\responderIdx{\loginIdx}}{\detectedFlag} \gets
      \boolTrue$ if \adsDetected{\adsCount} and, at this point,
      \setSize{\cset{\big}{\siteId{\responderIdxAlt}}{\leakedPwd \in
          \siteId{\responderIdxAlt}{\suspiciousArray{\accountId{}}}
          \wedge \responderIdxAlt \neq \responderIdx{\loginIdx}}}
      $\ge$ \attackWidth.  Otherwise,
      $\siteId{\responderIdx{\loginIdx}}{\detectedFlag} \gets
      \boolFalse$.
    \end{itemize}
  \end{itemize}
  \label{tdr:attack}
\end{enumerate}
When the experiment is finished, define
\begin{align*}
  \accessedSet & = \cset{\Big}{\siteId{\responderIdx}}{\parbox[m]{0.69\columnwidth}{$\siteId{\responderIdx}{\attemptedFlag} \wedge \leakedPwd = \siteId{\responderIdx}{\passwordArray{\accountId{}}}~\wedge\\ (\siteId{\responderIdx}{\collectionFlag} \Rightarrow \siteId{\responderIdx} \not\in \afaSet{\accountId{}})$}} \\
  \detectedSet & = \cset{\big}{\siteId{\responderIdx}}{\siteId{\responderIdx} \in \accessedSet \wedge \siteId{\responderIdx}{\detectedFlag}}
\end{align*}
Then, we define $\csdTPR =
\frac{\expv{\setSize{\detectedSet}}}{\expv{\setSize{\accessedSet}}}$
where this ratio is computed using the adversary's optimal strategy
for minimizing \expv{\setSize{\detectedSet}} among all strategies that
maximize \expv{\setSize{\accessedSet}}.  The condition $\loginIdx >
\attackWidth$ in the last bullet of \enumref{tdr:attack} limits
\accessedSet and \detectedSet to include only sites at which the
attacker succeeded or was detected, respectively, starting with the
$(\attackWidth+1)$-th login attempt, since by design, our algorithm
cannot detect \attackWidth or fewer credential-stuffing login
attempts.  As such, \csdTPR is best interpreted as the true-detection
rate for attacks of width greater than \attackWidth.  We expect that
\csdTPR is very conservative as an estimate of the true detection rate
in practice, since it is computed using the best possible strategy for
\credStuffer, equipped with perfect knowledge of parameters he would
not generally have.

Consider an example with $\nmbrResponders{\accountId{}} = 2$ sites
(\siteId{1} and \siteId{2}), neither of which support second-factor
authentication for \accountId{}; $\setSize{\pwdDist{\accountId{}}} =
2$ passwords (\password{1} and \password{2}); and $\attackWidth = 1$.
Suppose \credStuffer is given \password{1} (with probability
$\pwdDist{\accountId{}}(\password{1})$) as the ``leaked'' password.
\credStuffer{} picks one site, say \siteId{2}, and tries to log in
with \password{1}.  With probability
$\pwdDist{\accountId{}}(\password{2}) \cdot \adsTPR{\adsCollect}$,
\credStuffer{} fails at \siteId{2} with \password{1} being added to
\siteId{2}{\suspiciousArray{\accountId{}}}.  In this event, suppose
\credStuffer{} then submits \password{1} to \siteId{1}, where
\password{1} is correct (with probability
$\pwdDist{\accountId{}}(\password{1})$) and so \siteId{1} is added to
\accessedSet, but where this login attempt is detected as abnormal in
\siteId{1}'s \countingPhase (with probability \adsTPR{\adsCount}).
Since \password{1} appears in $\attackWidth = 1$ of the suspicious
sets at other sites (i.e., in
\siteId{2}{\suspiciousArray{\accountId{}}}), \siteId{1} is added to
\detectedSet.  \credStuffer{}'s choices induce these events with
probability $\pwdDist{\accountId{}}(\password{1})$ $\cdot$
$[\pwdDist{\accountId{}}(\password{2})$ $\cdot$
  $\adsTPR{\adsCollect}]$ $\cdot$
$[\pwdDist{\accountId{}}(\password{1})$ $\cdot$ $\adsTPR{\adsCount}]$.
\csdTPR is then computed by considering all possible choices by
\credStuffer and all event sequences.

\begin{figure*}[t]

  \vspace*{-0.3em}
  \begin{subfigure}[b]{.1\columnwidth}
    \setlength\figureheight{2in}
    \centering
    \hspace*{4.0em}
    \resizebox{!}{2.75em}{\newenvironment{customlegend}[1][]{%
    \begingroup
    \csname pgfplots@init@cleared@structures\endcsname
    \pgfplotsset{#1}%
}{%
    \csname pgfplots@createlegend\endcsname
    \endgroup
}%

\def\addlegendimage{\csname pgfplots@addlegendimage\endcsname}

\begin{tikzpicture}

\begin{customlegend}[
    legend style={{font={\fontsize{10pt}{12}\selectfont}},{draw=none}},
    legend columns=3,
    legend cell align={left},
    legend entries={{$(\adsFPR{\adsCollect}, \adsTPR{\adsCollect}) = (0.05, 0.61)\quad$}, {$(\adsFPR{\adsCollect}, \adsTPR{\adsCollect}) = (0.10, 0.74)\quad$},{$(\adsFPR{\adsCollect}, \adsTPR{\adsCollect}) = (0.20, 0.88)\quad$},{$(\adsFPR{}, \adsTPR{})$~\cite{freeman2016:ads}},{Blind guessing}}]
\addlegendimage{line width=1pt, densely dotted, curve_color}
\addlegendimage{line width=1pt, dash pattern=on 1pt off 3pt on 3pt off 3pt, curve_color}
\addlegendimage{line width=1pt, solid, curve_color}
\addlegendimage{line width=1.5pt, dash pattern=on 1pt off 3pt on 3pt off 3pt, black, black}
\addlegendimage{line width=1.5pt, dotted, black}

\end{customlegend}

\end{tikzpicture}}
  \end{subfigure}
  
  \vspace*{0.25em}
  \hspace*{0em}
  \captionsetup[subfigure]{font=normalsize,labelfont=normalsize, oneside,margin={-0.5in,0in}}
  \begin{subfigure}[b]{0.197\textwidth}
  \setlength\figureheight{1.9in}
    \centering
    \resizebox{!}{11.02em}{\begin{tikzpicture}

\pgfplotsset{every axis/.append style={
                  ylabel={\csdTPR},
                    compat=1.3,
                    x label style={yshift=-1.5em, align=center},
                    label style={font=\small},
                    tick label style={font=\small}  
                    }}
\begin{axis}[
xmin=0, xmax=0.4,
ymin=0.3, ymax=1.0,
width=\figurewidth,
height=\figureheight,
xtick={0.0,0.2,0.4,0.6,0.8,1.0},
xticklabels={0.0,0.2,0.4,0.6,0.8,1.0},
ytick={0.3,0.5,0.7,0.9},
yticklabels={0.3,0.5,0.7,0.9},
tick align=outside,
tick pos=left,
minor tick num=1,
xmajorgrids,
x grid style={lightgray!92.02614379084967!black},
ymajorgrids,
y grid style={lightgray!92.02614379084967!black},
grid=both,
]
\addplot [line width=1pt, densely dotted, mark=*, mark options={scale=0.5}, curve_color]
table {%
9.84373978476361e-15 nan
1.888649568577597e-12 0.0005216194222074089
1.634194063984849e-10 0.003234801482980054
8.404594498185106e-09 0.01313719128302793
2.848793292800541e-07 0.03989077379617589
6.663972586191839e-06 0.09773936230910507
0.0001093626100797433 0.2018057507980536
0.001252518921526494 0.3586967516836818
0.009740481367446544 0.5556609426880594
0.04858882603531809 0.7623513754664278
};
\addplot [line width=1pt, dash pattern=on 1pt off 3pt on 3pt off 3pt, mark=*, mark options={scale=0.5}, curve_color]
table {%
1.007998953959794e-11 nan
9.259773836026651e-10 0.00296804660608796
3.845573190232552e-08 0.01308704005291161
9.52443696077514e-07 0.03921044096717341
1.562148723087143e-05 0.09167686778332151
0.0001780801516982991 0.1804449750995354
0.001440127041300108 0.3102559357263457
0.008280448412106239 0.4732584390392424
0.03344413447121299 0.6480184751906908
0.09277947978450021 0.811043550944416
};
\addplot [line width=1pt, solid, mark=*, mark options={scale=0.5}, curve_color]
table {%
1.032190928854829e-08 nan
4.321049562032821e-07 0.0141165733145604
8.225645491548665e-06 0.04335857395161469
9.413378146526999e-05 0.09634860474463325
0.0007215263672467441 0.1779825456180392
0.003908383052095502 0.2913609394524878
0.01540524522573553 0.4317163569295284
0.04495367247226023 0.5831517829490638
0.0985004959262589 0.7262276797148536
0.1656589752200726 0.8490345379941602
};
\addplot [line width=1.5pt, dash pattern=on 1pt off 3pt on 3pt off 3pt, black]
table {%
0.01                0.28
0.05                0.61
0.10                0.74
0.20                0.88
0.30				0.95
0.40                0.97
0.60                1.00
};
 \addplot [line width=1.5pt, dotted, black]
 table {%
 0                  0
 1                  1
 };
\end{axis}

\end{tikzpicture}}
    \begin{minipage}[t]{16em}
      \vspace*{0.6em}\caption{Baseline}
      \label{fig:roc-phishing:baseline}
    \end{minipage}%
  \end{subfigure}
  \captionsetup[subfigure]{font=normalsize,labelfont=normalsize, oneside,margin={-0.8in,0in}}
  \hspace*{-1.0em}
  \begin{subfigure}[b]{0.191\textwidth}
  \setlength\figureheight{1.9in}
    \centering
    \resizebox{!}{11.32em}{\begin{tikzpicture}
\pgfplotsset{every axis/.append style={
                    compat=1.3,
                    x label style={yshift=-1.5em, align=center},
                    label style={font=\small},
                    tick label style={font=\small}  
                    }}
\begin{axis}[
xmin=0, xmax=0.4,
ymin=0.3, ymax=1.0,
width=\figurewidth,
height=\figureheight,
xtick={0.0,0.2,0.4,0.6,0.8,1.0},
xticklabels={0.0,0.2,0.4,0.6,0.8,1.0},
ytick={0.0,0.2,0.4,0.6,0.8,1.0},
yticklabels={},
tick align=outside,
tick pos=left,
minor tick num=1,
xmajorgrids,
x grid style={lightgray!92.02614379084967!black},
ymajorgrids,
y grid style={lightgray!92.02614379084967!black},
grid=both,
]
\addplot [line width=1pt, densely dotted, mark=*, mark options={scale=0.5}, curve_color]
table {%
1.098632812500001e-14 nan
2.10754386149347e-12 0.0008341487261505076
1.823250306677074e-10 0.00577978461616846
9.374660177025951e-09 0.02502182152159316
3.176615892825976e-07 0.07636729548358656
7.427711111849641e-06 0.1764417360559714
0.0001218242121039507 0.3252944936744792
0.001393999126252267 0.5011870280666136
0.01082485661483491 0.6743537035389232
0.05385016432417706 0.8253424631597747
};
\addplot [line width=1pt, dash pattern=on 1pt off 3pt on 3pt off 3pt, mark=*, mark options={scale=0.5}, curve_color]
table {%
1.125000000000001e-11 nan
1.033124914169313e-09 0.004746357574545335
4.288836573600772e-08 0.0237745478508754
1.061687263355256e-06 0.07588778968875243
1.740161887398721e-05 0.175342168181984
0.0001981928594126702 0.3174054312577149
0.001600710316252678 0.4781752336883961
0.009185940986964916 0.6308449476114515
0.03698577800737643 0.7605149741685504
0.1020481002326944 0.8654232799868881
};
\addplot [line width=1pt, solid, mark=*, mark options={scale=0.5}, curve_color]
table {%
1.152000000000001e-08 nan
4.819199121093755e-07 0.02257454601310993
9.165719780273444e-06 0.08072110392566589
0.0001047698016210938 0.1899661521250132
0.0008018028148535163 0.3371746554904098
0.004333940525493167 0.4922738880103809
0.01703066487985108 0.630072170857406
0.04947473300660159 0.7411570808240348
0.1076777680959229 0.8277310452124949
0.1792666283997658 0.8956261972371486
};
\addplot [line width=1.5pt, dash pattern=on 1pt off 3pt on 3pt off 3pt, black]
table {%
0.01                0.28
0.05                0.61
0.10                0.74
0.20                0.88
0.30				0.95
0.40                0.97
0.60                1.00
};
 \addplot [line width=1.5pt, dotted, black]
 table {%
 0                  0
 1                  1
 };
\end{axis}

\end{tikzpicture}}
    \begin{minipage}[t]{15.5em}
      \vspace*{0.6em}\caption{$\zipfShape{\accountId{}}=0$}
      \label{fig:roc-phishing:zipf}
    \end{minipage}%
  \end{subfigure}
  \hspace*{-2.5em}
   \begin{subfigure}[b]{0.191\textwidth}
  \setlength\figureheight{1.9in}
    \centering
    \resizebox{!}{11.32em}{\begin{tikzpicture}

\pgfplotsset{every axis/.append style={
                    compat=1.3,
                    x label style={yshift=-1.5em, align=center},
                    label style={font=\small},
                    tick label style={font=\small}  
                    }}
\begin{axis}[
xmin=0, xmax=0.4,
ymin=0.3, ymax=1.0,
width=\figurewidth,
height=\figureheight,
xtick={0.0,0.2,0.4,0.6,0.8,1.0},
xticklabels={0.0,0.2,0.4,0.6,0.8,1.0},
ytick={0.0,0.2,0.4,0.6,0.8,1.0},
yticklabels={},
tick align=outside,
tick pos=left,
minor tick num=1,
xmajorgrids,
x grid style={lightgray!92.02614379084967!black},
ymajorgrids,
y grid style={lightgray!92.02614379084967!black},
grid=both,
]
\addplot [line width=1pt, densely dotted, mark=*, mark options={scale=0.5}, curve_color]
table {%
1.053619197914081e-14 nan
2.021206020411093e-12 0.0007516093953591207
1.748574157842211e-10 0.004455833685457988
8.990811320863967e-09 0.01737315518087013
3.046605572618659e-07 0.05073054268287702
7.123913241905954e-06 0.1192803269161564
0.0001168470820514333 0.2353299338692805
0.001337158032399866 0.3987576665591608
0.01038515410147282 0.5907089253568205
0.05168193495792753 0.781724343605827
};
\addplot [line width=1pt, dash pattern=on 1pt off 3pt on 3pt off 3pt, mark=*, mark options={scale=0.5}, curve_color]
table {%
1.078906058664019e-11 nan
9.908087026617161e-10 0.004276703703936957
4.113254821253148e-08 0.01789464114467321
1.018255864210856e-06 0.05127029501426805
1.669063033745341e-05 0.1150344970959065
0.0001901116362028778 0.2168939709430278
0.001535669785996978 0.3560017378975197
0.008815071042367581 0.5184897731373896
0.03551113055131241 0.6819601586894158
0.09808848134944723 0.8279727666474916
};
\addplot [line width=1pt, solid, mark=*, mark options={scale=0.5}, curve_color]
table {%
1.104799804071955e-08 nan
4.621881052760327e-07 0.02034078617817003
8.790936533185761e-06 0.05862579671479218
0.0001004963289492889 0.1241036003357133
0.0007692371061694528 0.2197636699659667
0.004159219967045387 0.3442184724697303
0.01635304433868196 0.486884492081387
0.04755230416581534 0.6298304011588711
0.1036734239310111 0.7574837803985834
0.1731294217738933 0.8636586769353314
};
\addplot [line width=1.5pt, dash pattern=on 1pt off 3pt on 3pt off 3pt, black]
table {%
0.01                0.28
0.05                0.61
0.10                0.74
0.20                0.88
0.30				0.95
0.40                0.97
0.60                1.00
};
 \addplot [line width=1.5pt, dotted, black]
 table {%
 0                  0
 1                  1
 };
\end{axis}

\end{tikzpicture}}
    \begin{minipage}[t]{15.5em}
      \vspace*{0.6em}\caption{$\setSize{\pwdDist{\accountId{}}}=5$}
      \label{fig:roc-phishing:pwds}
    \end{minipage}%
  \end{subfigure}
  \hspace*{-2.5em}
   \begin{subfigure}[b]{0.191\textwidth}
  \setlength\figureheight{1.9in}
    \centering
    \resizebox{!}{11.32em}{
\begin{tikzpicture}

\pgfplotsset{every axis/.append style={
                    compat=1.3,
                    x label style={yshift=-1.5em, align=center},
                    label style={font=\small},
                    tick label style={font=\small}  
                    }}
\begin{axis}[
xmin=0, xmax=0.4,
ymin=0.3, ymax=1.0,
width=\figurewidth,
height=\figureheight,
xtick={0.0,0.2,0.4,0.6,0.8,1.0},
xticklabels={0.0,0.2,0.4,0.6,0.8,1.0},
ytick={0.0,0.2,0.4,0.6,0.8,1.0},
yticklabels={},
tick align=outside,
tick pos=left,
minor tick num=1,
xmajorgrids,
x grid style={lightgray!92.02614379084967!black},
ymajorgrids,
y grid style={lightgray!92.02614379084967!black},
grid=both,
]
\addplot [line width=1pt, densely dotted, mark=*, mark options={scale=0.5}, curve_color]
table {%
9.613030577166449e-28 nan
3.670859832288267e-25 6.078104957119024e-07
6.661855955471085e-23 6.277869305404415e-06
7.640379893729252e-21 4.202950565546093e-05
6.211332634858669e-19 0.0002054217114069168
3.805276183148122e-17 0.0007832632565277775
1.823156475016578e-15 0.002433983295748066
6.99670647990524e-14 0.006368799571569106
2.185019023054405e-12 0.01442067093822985
5.609697241560033e-11 0.02896815713565426
1.191087287656837e-09 0.05284530330667747
2.096764403074673e-08 0.08935485718575287
3.058152798430364e-07 0.1421202968327605
3.681235697045247e-06 0.2141533994983802
3.630493509267302e-05 0.3059808251899079
0.0002900056228578935 0.413972181810161
0.001845485230009867 0.5306538070185257
0.009140344460600454 0.6474156812290042
0.03412798895697681 0.7577496571016393
0.09220418061698914 0.8586837229471741
};
\addplot [line width=1pt, dash pattern=on 1pt off 3pt on 3pt off 3pt, mark=*, mark options={scale=0.5}, curve_color]
table {%
1.007999315048289e-21 nan
1.833176606487475e-19 2.387283547067298e-05
1.585351818634625e-17 0.0001655124315526058
8.670434659479383e-16 0.0007546161745771718
3.364052074886584e-14 0.002554792601345501
9.845584691085849e-13 0.006884251293744148
2.256159034612624e-11 0.01548165364570775
4.14720831390677e-10 0.03015561723782356
6.214525033534059e-09 0.05253980751464238
7.672810046816697e-08 0.08426613477371736
7.857092311480445e-07 0.1273768580114645
6.695567735930744e-06 0.1842418780189654
4.750834366742681e-05 0.256414571083653
0.0002801063707417694 0.3429303289647162
0.001366056389168006 0.4395036166579873
0.005472122562321645 0.5396382734171832
0.01783824175888526 0.6370303664826924
0.04681842188451284 0.7275431002922979
0.09800785944756424 0.80962214617024
0.1634376947507078 0.8834707322284416
};
\addplot [line width=1pt, solid, mark=*, mark options={scale=0.5}, curve_color]
table {%
1.056963889776075e-15 nan
8.652609552179119e-14 0.000642202661213731
3.373220084737491e-12 0.002824585336654528
8.330716002923326e-11 0.008432867986716275
1.46255941547097e-09 0.01937887172132147
1.941627256874934e-08 0.03689343488343644
2.024221951474171e-07 0.06127158019387335
1.698992838502361e-06 0.0924830756746059
1.167773021134327e-05 0.1312084921632541
6.651200370908222e-05 0.1793018634509114
0.0003165034431483854 0.2389379398817096
0.001265355948317119 0.3107983569629607
0.004265764126145553 0.3926811444635957
0.01215592821945201 0.4797603643659162
0.02933897797089187 0.5663806785953649
0.06013067647104636 0.6480627458366822
0.1051596382964069 0.7224701157427316
0.1584170173016118 0.7891469460497708
0.2090319157635741 0.8486844356600578
0.2478134373236372 0.9020052657173742
};
\addplot [line width=1.5pt, dash pattern=on 1pt off 3pt on 3pt off 3pt, black]
table {%
0.01                0.28
0.05                0.61
0.10                0.74
0.20                0.88
0.30				0.95
0.40                0.97
0.60                1.00
};
 \addplot [line width=1.5pt, dotted, black]
 table {%
 0                  0
 1                  1
 };
\end{axis}

\end{tikzpicture}}
    \begin{minipage}[t]{15.5em}
      \vspace*{0.6em}\caption{$\nmbrResponders{\accountId{}}=20$}
      \label{fig:roc-phishing:sites}
    \end{minipage}%
  \end{subfigure}
  \hspace*{-2.5em}
   \begin{subfigure}[b]{0.191\textwidth}
  \setlength\figureheight{1.9in}
    \centering
    \resizebox{!}{11.32em}{\begin{tikzpicture}

\pgfplotsset{every axis/.append style={
                    compat=1.3,
                    x label style={yshift=-1.5em, align=center},
                    label style={font=\small},
                    tick label style={font=\small}  
                    }}
\begin{axis}[
xmin=0, xmax=0.4,
ymin=0.3, ymax=1.0,
width=\figurewidth,
height=\figureheight,
xtick={0.0,0.2,0.4,0.6,0.8,1.0},
xticklabels={0.0,0.2,0.4,0.6,0.8,1.0},
ytick={0.0,0.2,0.4,0.6,0.8,1.0},
yticklabels={},
tick align=outside,
tick pos=left,
minor tick num=1,
xmajorgrids,
x grid style={lightgray!92.02614379084967!black},
ymajorgrids,
y grid style={lightgray!92.02614379084967!black},
grid=both,
]
\addplot [line width=1pt, densely dotted, mark=*, mark options={scale=0.5}, curve_color]
table {%
9.84373978476361e-15 nan
1.888649568577597e-12 0.00194563180362683
1.634194063984849e-10 0.0104447111664725
8.404594498185106e-09 0.0421713566886117
2.848793292800541e-07 0.119733416768775
6.663972586191839e-06 0.257515460160462
0.0001093626100797433 0.383711473919577
0.001252518921526494 0.5098626003901729
0.009740481367446544 0.647172661248576
0.04858882603531809 0.787835458239267
};
\addplot [line width=1pt, dash pattern=on 1pt off 3pt on 3pt off 3pt, mark=*, mark options={scale=0.5}, curve_color]
table {%
1.007998953959794e-11 nan
9.259773836026651e-10 0.0110707646717117
3.845573190232552e-08 0.041094350376789
9.52443696077514e-07 0.11703431407095
1.562148723087143e-05 0.262433705914208
0.0001780801516982991 0.442681307315389
0.001440127041300108 0.573875720263742
0.008280448412106239 0.661551924670336
0.03344413447121299 0.749635325542137
0.09277947978450021 0.836225371151695
};
\addplot [line width=1pt, solid, mark=*, mark options={scale=0.5}, curve_color]
table {%
1.032190928854829e-08 nan
4.321049562032821e-07 0.0526545846065571
8.225645491548665e-06 0.154239811539004
9.413378146526999e-05 0.323183321335247
0.0007215263672467441 0.535139986676171
0.003908383052095502 0.525145593799845
0.01540524522573553 0.778415710255754
0.04495367247226023 0.812357511348438
0.0985004959262589 0.8457783423640159
0.1656589752200726 0.878578868547369
};
\addplot [line width=1.5pt, dash pattern=on 1pt off 3pt on 3pt off 3pt, black]
table {%
0.01                0.28
0.05                0.61
0.10                0.74
0.20                0.88
0.30				0.95
0.40                0.97
0.60                1.00
};
 \addplot [line width=1.5pt, dotted, black]
 table {%
 0                  0
 1                  1
 };
\end{axis}

\end{tikzpicture}}
    \begin{minipage}[t]{15em}
      \vspace*{0.6em}\caption{$\setSize{\afaSet{\accountId{}}}=5$}
      \label{fig:roc-phishing:afa}
    \end{minipage}%
  \end{subfigure}
  \hspace*{-2.5em}
   \begin{subfigure}[b]{0.191\textwidth}
  \setlength\figureheight{1.9in}
    \centering
    \resizebox{!}{11.32em}{\begin{tikzpicture}

\pgfplotsset{every axis/.append style={
                    compat=1.3,
                    x label style={yshift=-1.5em, align=center},
                    label style={font=\small},
                    tick label style={font=\small}  
                    }}
\begin{axis}[
xmin=0, xmax=0.4,
ymin=0.3, ymax=1.0,
width=\figurewidth,
height=\figureheight,
xtick={0.0,0.2,0.4,0.6,0.8,1.0},
xticklabels={0.0,0.2,0.4,0.6,0.8,1.0},
ytick={0.0,0.2,0.4,0.6,0.8,1.0},
yticklabels={},
tick align=outside,
tick pos=left,
minor tick num=1,
xmajorgrids,
x grid style={lightgray!92.02614379084967!black},
ymajorgrids,
y grid style={lightgray!92.02614379084967!black},
grid=both,
]
\addplot [line width=1pt, densely dotted, mark=*, mark options={scale=0.5}, curve_color]
table {%
1.312498637968482e-14 nan
2.518199424770129e-12 0.0005326008837275298
2.178925418646466e-10 0.003302902566832522
1.120612599758014e-08 0.01341376373109171
3.798391057067389e-07 0.0407305795603059
8.885296781589119e-06 0.09979703309456023
0.0001458168134396577 0.2060542929201179
0.001670025228701991 0.3662482622454436
0.01298730848992873 0.5673590677972817
0.06478510138042413 0.7784008781078263
};
\addplot [line width=1pt, dash pattern=on 1pt off 3pt on 3pt off 3pt, mark=*, mark options={scale=0.5}, curve_color]
table {%
1.343998605279725e-11 nan
1.23463651147022e-09 0.003030531797795111
5.127430920310071e-08 0.01336255668560449
1.269924928103352e-06 0.04003592393490352
2.082864964116191e-05 0.0936069071050758
0.0002374402022643989 0.1842438166805782
0.001920169388400144 0.3167876396363737
0.01104059788280832 0.4832217745979632
0.04459217929495066 0.6616609694052316
0.1237059730460003 0.82811815201693
};
\addplot [line width=1pt, solid, mark=*, mark options={scale=0.5}, curve_color]
table {%
1.376254571806439e-08 nan
5.761399416043762e-07 0.01441376433170904
1.096752732206489e-05 0.04427138603480663
0.00012551170862036 0.09837699642346764
0.0009620351563289922 0.1817295465784189
0.005211177402794003 0.297494853967277
0.02054032696764738 0.4408051223385712
0.05993822996301365 0.5954286625900969
0.1313339945683452 0.7415166834983242
0.2208786336267635 0.866908949320353
};
\addplot [line width=1.5pt, dash pattern=on 1pt off 3pt on 3pt off 3pt, black]
table {%
0.01                0.28
0.05                0.61
0.10                0.74
0.20                0.88
0.30				0.95
0.40                0.97
0.60                1.00
};
 \addplot [line width=1.5pt, dotted, black]
 table {%
 0                  0
 1                  1
 };
\end{axis}

\end{tikzpicture}}
    \begin{minipage}[t]{16em}
      \vspace*{0.6em}\caption{\parbox[t]{6.5em}{$\adsFPR{\adsCount}=0.40$ \\ $\adsTPR{\adsCount}=0.97$}}
      \label{fig:roc-phishing:cntads}
    \end{minipage}%
  \end{subfigure}
  
  \vspace*{-4.2em}
  \begin{subfigure}[b]{.43\columnwidth}
    \setlength\figureheight{2in}
    \begin{minipage}[b]{1\textwidth}
      \centering
      \hspace*{24.5em}
      \resizebox{!}{1.5em}{\begin{tikzpicture}
\node at (0,0)[
  scale=1,
  anchor=south,
  text=black,
  rotate=0
]{\csdFPR};
\end{tikzpicture}}\vspace*{-0.4em}
    \end{minipage}
  \end{subfigure}
   \vspace*{2.7em} 
  \caption{\textit{Phishing} attacker.  Baseline: $\setSize{\pwdDist{\accountId{}}}=4$, $\zipfShape{\accountId{}}=1$, 
    $\nmbrResponders{\accountId{}}=10$, $\setSize{\afaSet{\accountId{}}}=0$, $(\adsFPR{\adsCount}, \adsTPR{\adsCount})=(0.30, 0.95)$.}
  \label{fig:roc-phishing}
\end{figure*}

\begin{figure*}[t]

  \begin{subfigure}[b]{.1\columnwidth}
    \setlength\figureheight{2in}
    \centering
    \hspace*{4.0em}
    \resizebox{!}{2.8em}{\newenvironment{customlegend}[1][]{%
    \begingroup
    \csname pgfplots@init@cleared@structures\endcsname
    \pgfplotsset{#1}%
}{%
    \csname pgfplots@createlegend\endcsname
    \endgroup
}%

\def\addlegendimage{\csname pgfplots@addlegendimage\endcsname}

\begin{tikzpicture}

\begin{customlegend}[
    legend style={{font={\fontsize{10pt}{12}\selectfont}},{draw=none}},
    legend columns=3,
    legend cell align={left},
    legend entries={{$(\adsFPR{\adsCollect}, \adsTPR{\adsCollect}) = (0.01, 0.65)\quad$}, {$(\adsFPR{\adsCollect}, \adsTPR{\adsCollect}) = (0.02, 0.80)\quad$},{$(\adsFPR{\adsCollect}, \adsTPR{\adsCollect}) = (0.05, 0.93)\quad$},{$(\adsFPR{}, \adsTPR{})$~\cite{freeman2016:ads}}}]
\addlegendimage{line width=1pt, densely dotted, curve_color}
\addlegendimage{line width=1pt, dash pattern=on 1pt off 3pt on 3pt off 3pt, curve_color}
\addlegendimage{line width=1pt, solid, curve_color}
\addlegendimage{line width=1.5pt, dash pattern=on 1pt off 3pt on 3pt off 3pt, black, black}
\addlegendimage{line width=1.5pt, dotted, black}

\end{customlegend}

\end{tikzpicture}}
  \end{subfigure}
  
  \vspace*{0.5em}
 \hspace*{0em}
  \captionsetup[subfigure]{font=normalsize,labelfont=normalsize, oneside,margin={-0.5in,0in}}
  \begin{subfigure}[b]{0.195\textwidth}
  \setlength\figureheight{1.9in}
    \centering
    \resizebox{!}{11.32em}{\begin{tikzpicture}

\pgfplotsset{every axis/.append style={
                  ylabel={\csdTPR},
                    compat=1.3,
                    x label style={yshift=-1.5em, align=center},
                    label style={font=\small},
                    tick label style={font=\small}  
                    }}
\begin{axis}[
xmin=0, xmax=0.1,
ymin=0.3, ymax=1.0,
width=\figurewidth,
height=\figureheight,
xtick={0.00, 0.05, 0.10},
xticklabels={0.00, 0.05, 0.10},
ytick={0.1,0.3,0.5,0.7,0.9},
ytick={0.1,0.3,0.5,0.7,0.9},
tick align=outside,
tick pos=left,
minor tick num=1,
xmajorgrids,
x grid style={lightgray!92.02614379084967!black},
ymajorgrids,
y grid style={lightgray!92.02614379084967!black},
grid=both,
]
\addplot [line width=1pt, densely dotted, mark=*, mark options={scale=0.5}, curve_color]
table {%
3.359996513199312e-22 nan
3.332658451729938e-19 0.0009627577820225941
1.488115642761195e-16 0.005366100856432965
3.939931239592901e-14 0.01977989716685802
6.851200601367581e-12 0.05515942142506003
8.17942248634949e-10 0.1256565302795579
6.794719785629124e-08 0.2439157383859299
3.883786501480441e-06 0.4115491198356644
0.0001466811359796489 0.6110147828490484
0.003338508085261496 0.8118713487206077
};
\addplot [line width=1pt, dash pattern=on 1pt off 3pt on 3pt off 3pt, mark=*, mark options={scale=0.5}, curve_color]
table {%
3.440636429516095e-19 nan
1.692322639218634e-16 0.006238881127434426
3.74892754007669e-14 0.02357431434187873
4.927109725350221e-12 0.06194289113978368
4.256649309994701e-10 0.1301830180648214
2.527943984812967e-08 0.2353492296352623
1.046716362152842e-06 0.3776469619991263
2.992508584447574e-05 0.5443064637458329
0.000569169144875976 0.7127730006014348
0.006631160685159794 0.8637118631587491
};
\addplot [line width=1pt, solid, mark=*, mark options={scale=0.5}, curve_color]
table {%
3.281246594921204e-15 nan
6.295498561925324e-13 0.02419041755918094
5.447313546616164e-11 0.06510317699025503
2.801531499395035e-09 0.1308420255830325
9.495977642668471e-08 0.2247150391787237
2.22132419539728e-06 0.3483878587715873
3.645420335991443e-05 0.4942164964737625
0.0004175063071754979 0.6444413527706685
0.003246827122482181 0.7812622229580743
0.01619627534510603 0.8962171551184216
};
\addplot [line width=1.5pt, dash pattern=on 1pt off 3pt on 3pt off 3pt, black]
table {%
0.001				0.10
0.005				0.50
0.01                0.65
0.02				0.80
0.05                0.93
0.10                0.99
0.20                1.00
};
 \addplot [line width=1.5pt, dotted, black]
 table {%
 0                  0
 1                  1
 };
\end{axis}

\end{tikzpicture}}
    \begin{minipage}[t]{16em}
      \vspace*{0.6em}\caption{Baseline}
      \label{fig:roc-researching:baseline}
    \end{minipage}%
  \end{subfigure}
  \captionsetup[subfigure]{font=normalsize,labelfont=normalsize, oneside,margin={-0.8in,0in}}
  \hspace*{-0.5em}
  \begin{subfigure}[b]{0.189\textwidth}
  \setlength\figureheight{1.9in}
    \centering
    \resizebox{!}{11.32em}{\begin{tikzpicture}
\pgfplotsset{every axis/.append style={
                    compat=1.3,
                    x label style={yshift=-1.5em, align=center},
                    label style={font=\small},
                    tick label style={font=\small}  
                    }}
\begin{axis}[
xmin=0, xmax=0.1,
ymin=0.3, ymax=1.0,
width=\figurewidth,
height=\figureheight,
xtick={0.00, 0.05, 0.10},
xticklabels={0.00, 0.05, 0.10},
ytick={0.1,0.3,0.5,0.7,0.9},
yticklabels={},
tick align=outside,
tick pos=left,
minor tick num=1,
xmajorgrids,
x grid style={lightgray!92.02614379084967!black},
ymajorgrids,
y grid style={lightgray!92.02614379084967!black},
grid=both,
]
\addplot [line width=1pt, densely dotted, mark=*, mark options={scale=0.5}, curve_color]
table {%
3.750000000000002e-22 nan
3.719374971389772e-19 0.001539596002900323
1.660733518915177e-16 0.009632344223274303
4.396750863174058e-14 0.03785320981227025
7.645124878134763e-12 0.105627980904395
9.126550830175396e-10 0.2252757414643823
7.580663968173857e-08 0.3880742180965074
4.332308220615323e-06 0.5667098592796029
0.0001635756362678881 0.7334608416503727
0.003721008636267528 0.8746410442705561
};
\addplot [line width=1pt, dash pattern=on 1pt off 3pt on 3pt off 3pt, mark=*, mark options={scale=0.5}, curve_color]
table {%
3.840000000000002e-19 nan
1.888639970703126e-16 0.009976919039999976
4.183511819628908e-14 0.04323331584000001
5.497763713537112e-12 0.1208315750400001
4.749085144189792e-10 0.2483144294400002
2.819941160725803e-08 0.4075756185600001
1.167358708070388e-06 0.5689307289600001
3.336292306569925e-05 0.7104610384457144
0.00063420094854837 0.82536765696
0.007380759646320923 0.91668589056
};
\addplot [line width=1pt, solid, mark=*, mark options={scale=0.5}, curve_color]
table {%
3.662109375000005e-15 nan
7.025146204978234e-13 0.03868415387343516
6.077501022256915e-11 0.1225691470398447
3.124886725675317e-09 0.2596715014623633
1.058871964275325e-07 0.4228895137566303
2.475903703949881e-06 0.5782990374882497
4.060807070131691e-05 0.7068334320869716
0.0004646663754174224 0.8063659195211874
0.003608285538278303 0.882689813546556
0.01795005477472569 0.9422949186277142
};
\addplot [line width=1.5pt, dash pattern=on 1pt off 3pt on 3pt off 3pt, black]
table {%
0.001				0.10
0.005				0.50
0.01                0.65
0.02				0.80
0.05                0.93
0.10                0.99
0.20                1.00
};
 \addplot [line width=1.5pt, dotted, black]
 table {%
 0                  0
 1                  1
 };
\end{axis}

\end{tikzpicture}}
    \begin{minipage}[t]{15.5em}
      \vspace*{0.6em}\caption{$\zipfShape{\accountId{}}=0$}
      \label{fig:roc-researching:zipf}
    \end{minipage}%
  \end{subfigure}
  \hspace*{-2.4em}
   \begin{subfigure}[b]{0.189\textwidth}
  \setlength\figureheight{1.9in}
    \centering
    \resizebox{!}{11.32em}{\begin{tikzpicture}

\pgfplotsset{every axis/.append style={
                    compat=1.3,
                    x label style={yshift=-1.5em, align=center},
                    label style={font=\small},
                    tick label style={font=\small}  
                    }}
\begin{axis}[
xmin=0, xmax=0.1,
ymin=0.3, ymax=1.0,
width=\figurewidth,
height=\figureheight,
xtick={0.00, 0.05, 0.10},
xticklabels={0.00, 0.05, 0.10},
ytick={0.1,0.3,0.5,0.7,0.9},
yticklabels={},
tick align=outside,
tick pos=left,
minor tick num=1,
xmajorgrids,
x grid style={lightgray!92.02614379084967!black},
ymajorgrids,
y grid style={lightgray!92.02614379084967!black},
grid=both,
]
\addplot [line width=1pt, densely dotted, mark=*, mark options={scale=0.5}, curve_color]
table {%
3.596353528880061e-22 nan
3.566988354271634e-19 0.00138725239823545
1.592693538709152e-16 0.007376592870402376
4.216629440803551e-14 0.02607322750832342
7.33194438197329e-12 0.06987141097484373
8.752715589103375e-10 0.1526525562131748
7.270193590705917e-08 0.2830322290611604
4.154908882343889e-06 0.4554115884122414
0.0001568800429153053 0.6474122916023206
0.003568818465208868 0.8312574011904688
};
\addplot [line width=1pt, dash pattern=on 1pt off 3pt on 3pt off 3pt, mark=*, mark options={scale=0.5}, curve_color]
table {%
3.682666013573183e-19 nan
1.811262720552512e-16 0.008989699141311536
4.012125889676958e-14 0.03209679878257177
5.272561747357734e-12 0.08050570152652803
4.554577914922211e-10 0.162257314280686
2.704467869163058e-08 0.2808260384462379
1.119570998397059e-06 0.4301106830981204
3.199784042085918e-05 0.5925669601012995
0.0006082757437084467 0.7470633599863163
0.007079683471296426 0.8802715689817694
};
\addplot [line width=1pt, solid, mark=*, mark options={scale=0.5}, curve_color]
table {%
3.512063993046936e-15 nan
6.737353401370312e-13 0.03485634227000667
5.828580526140704e-11 0.08756505175358442
2.996937106954656e-09 0.1675034910376016
1.015535190872886e-07 0.2757991371621294
2.374637747301985e-06 0.4090120986943955
3.894902735047777e-05 0.5540218499278002
0.0004457193441332885 0.6927905971532103
0.003461718033824272 0.8126783044157522
0.01722731165264252 0.910706399019047
};
\addplot [line width=1.5pt, dash pattern=on 1pt off 3pt on 3pt off 3pt, black]
table {%
0.001				0.10
0.005				0.50
0.01                0.65
0.02				0.80
0.05                0.93
0.10                0.99
0.20                1.00
};
 \addplot [line width=1.5pt, dotted, black]
 table {%
 0                  0
 1                  1
 };
\end{axis}

\end{tikzpicture}}
    \begin{minipage}[t]{15.5em}
      \vspace*{0.6em}\caption{$\setSize{\pwdDist{\accountId{}}}=5$}
      \label{fig:roc-researching:pwds}
    \end{minipage}%
  \end{subfigure}
  \hspace*{-2.4em}
   \begin{subfigure}[b]{0.189\textwidth}
  \setlength\figureheight{1.9in}
    \centering
    \resizebox{!}{11.32em}{\begin{tikzpicture}

\pgfplotsset{every axis/.append style={
                    compat=1.3,
                    x label style={yshift=-1.5em, align=center},
                    label style={font=\small},
                    tick label style={font=\small}  
                    }}
\begin{axis}[
xmin=0, xmax=0.1,
ymin=0.3, ymax=1.0,
width=\figurewidth,
height=\figureheight,
xtick={0.00, 0.05, 0.10},
xticklabels={0.00, 0.05, 0.10},
ytick={0.1,0.3,0.5,0.7,0.9},
yticklabels={},
tick align=outside,
tick pos=left,
minor tick num=1,
xmajorgrids,
x grid style={lightgray!92.02614379084967!black},
ymajorgrids,
y grid style={lightgray!92.02614379084967!black},
grid=both,
]
\addplot [line width=1pt, densely dotted, mark=*, mark options={scale=0.5}, curve_color]
table {%
3.359997716827627e-42 nan
6.659055607871237e-39 2.117223450226469e-06
6.269354173946972e-36 1.934974143991131e-05
3.72829988235746e-33 0.0001150353113614777
1.570710941515119e-30 0.0005013849932465897
4.983249138570111e-28 0.001713410818948713
1.235389179344123e-25 0.004800828364542764
2.450683027592313e-23 0.01140950917264305
3.951118498354472e-21 0.02367143661713156
5.22872895983055e-19 0.04401982786394976
5.711203086912368e-17 0.07516783784632508
5.158631705019871e-15 0.1202220362320436
3.84719066561845e-13 0.1823737522376147
2.356769331019875e-11 0.2635793867377395
1.174889924678296e-09 0.3626361642376121
4.696682477658991e-08 0.4742505070076467
1.472340544367913e-06 0.5904161471350163
3.498067639156476e-05 0.7034533410902146
0.0005963468170941496 0.8084906754022171
0.006625522067673435 0.9038771813093204
};
\addplot [line width=1pt, dash pattern=on 1pt off 3pt on 3pt off 3pt, mark=*, mark options={scale=0.5}, curve_color]
table {%
3.523212965920246e-36 nan
3.459312432724864e-33 0.0001094265412316942
1.613694623707729e-30 0.0006279009772555133
4.755348596498721e-28 0.002394459882015698
9.92890824171713e-26 0.006863684349608379
1.56143430704718e-23 0.01588037072403548
1.91914373629018e-21 0.03115879681808664
1.887942762150366e-19 0.053923175215569
1.509904643819992e-17 0.08514046457462499
9.915526641433454e-16 0.126201472855254
5.377052815105064e-14 0.1792326858150611
2.4127774781948e-12 0.2462553845927546
8.946459316754665e-11 0.3274288981226932
2.72798583050395e-09 0.4197487423491627
6.78019557830426e-08 0.5174973785600502
1.354582810924985e-06 0.6143103408145136
2.130374126726999e-05 0.7053667784283733
0.0002556050310533321 0.7883265816203274
0.002228674453725565 0.862884030258817
0.01303768949003786 0.9297738131529472
};
\addplot [line width=1pt, solid, mark=*, mark options={scale=0.5}, curve_color]
table {%
3.20434352572215e-28 nan
1.223619944096089e-25 0.001912444946051473
2.220618651823695e-23 0.007033857261612431
2.546793297909751e-21 0.01791714497341812
2.070444211619556e-19 0.03589014014026248
1.268425394382708e-17 0.06088658709959727
6.077188250055259e-16 0.09213241299321961
2.332235493301747e-14 0.1294799980945445
7.283396743514684e-13 0.1743416634886614
1.869899080520011e-11 0.2292653650101922
3.970290958856122e-10 0.2961463315953574
6.989214676915577e-09 0.3742977828773554
1.019384266143455e-07 0.4599189173117196
1.227078565681749e-06 0.5474193865899337
1.210164503089101e-05 0.6315810471942496
9.666854095263119e-05 0.709081945607522
0.0006151617433366223 0.7787207086673172
0.003046781486866818 0.8407401699638236
0.01137599631899227 0.8960045572360942
0.03073472687232972 0.9454746320433821
};
\addplot [line width=1.5pt, dash pattern=on 1pt off 3pt on 3pt off 3pt, black]
table {%
0.001				0.10
0.005				0.50
0.01                0.65
0.02				0.80
0.05                0.93
0.10                0.99
0.20                1.00
};
 \addplot [line width=1.5pt, dotted, black]
 table {%
 0                  0
 1                  1
 };
\end{axis}

\end{tikzpicture}}
    \begin{minipage}[t]{15.5em}
      \vspace*{0.6em}\caption{$\nmbrResponders{\accountId{}}=20$}
      \label{fig:roc-researching:sites}
    \end{minipage}%
  \end{subfigure}
  \hspace*{-2.4em}
   \begin{subfigure}[b]{0.189\textwidth}
  \setlength\figureheight{1.9in}
    \centering
    \resizebox{!}{11.32em}{\begin{tikzpicture}

\pgfplotsset{every axis/.append style={
                    compat=1.3,
                    x label style={yshift=-1.5em, align=center},
                    label style={font=\small},
                    tick label style={font=\small}  
                    }}
\begin{axis}[
xmin=0, xmax=0.1,
ymin=0.3, ymax=1.0,
width=\figurewidth,
height=\figureheight,
xtick={0.00, 0.05, 0.10},
xticklabels={0.00, 0.05, 0.10},
ytick={0.1,0.3,0.5,0.7,0.9},
yticklabels={},
tick align=outside,
tick pos=left,
minor tick num=1,
xmajorgrids,
x grid style={lightgray!92.02614379084967!black},
ymajorgrids,
y grid style={lightgray!92.02614379084967!black},
grid=both,
]
\addplot [line width=1pt, densely dotted, mark=*, mark options={scale=0.5}, curve_color]
table {%
3.359996513199312e-22 nan
3.332658451729938e-19 0.00359107057778933
1.488115642761195e-16 0.0193241744617303
3.939931239592901e-14 0.06811257650160191
6.851200601367581e-12 0.175844129431334
8.17942248634949e-10 0.353883019294768
6.794719785629124e-08 0.463490101951056
3.883786501480441e-06 0.589479687041826
0.0001466811359796489 0.7219332765516659
0.003338508085261496 0.854692605083229
};
\addplot [line width=1pt, dash pattern=on 1pt off 3pt on 3pt off 3pt, mark=*, mark options={scale=0.5}, curve_color]
table {%
3.440636429516095e-19 nan
1.692322639218634e-16 0.0232709232513189
3.74892754007669e-14 0.0842920222490047
4.927109725350221e-12 0.209991380601951
4.256649309994701e-10 0.400762668380192
2.527943984812967e-08 0.624211968043257
1.046716362152842e-06 0.6966892289948931
2.992508584447574e-05 0.770659254742511
0.000569169144875976 0.843539512278799
0.006631160685159794 0.914505278073696
};
\addplot [line width=1pt, solid, mark=*, mark options={scale=0.5}, curve_color]
table {%
3.281246594921204e-15 nan
6.295498561925324e-13 0.0902298567545478
5.447313546616164e-11 0.230693768714765
2.801531499395035e-09 0.435549483729182
9.495977642668471e-08 0.6643549657216909
2.22132419539728e-06 0.86788099074327
3.645420335991443e-05 0.891140602365766
0.0004175063071754979 0.914205453481912
0.003246827122482181 0.936930887080938
0.01619627534510603 0.9592993442776629
};
\addplot [line width=1.5pt, dash pattern=on 1pt off 3pt on 3pt off 3pt, black]
table {%
0.001				0.10
0.005				0.50
0.01                0.65
0.02				0.80
0.05                0.93
0.10                0.99
0.20                1.00
};
 \addplot [line width=1.5pt, dotted, black]
 table {%
 0                  0
 1                  1
 };
\end{axis}

\end{tikzpicture}}
    \begin{minipage}[t]{15em}
      \vspace*{0.6em}\caption{$\setSize{\afaSet{\accountId{}}}=5$}
      \label{fig:roc-researching:afa}
    \end{minipage}%
  \end{subfigure}
  \hspace*{-2.4em}
   \begin{subfigure}[b]{0.189\textwidth}
  \setlength\figureheight{1.9in}
    \centering
    \resizebox{!}{11.32em}{\begin{tikzpicture}

\pgfplotsset{every axis/.append style={
                    compat=1.3,
                    x label style={yshift=-1.5em, align=center},
                    label style={font=\small},
                    tick label style={font=\small}  
                    }}
\begin{axis}[
xmin=0, xmax=0.1,
ymin=0.3, ymax=1.0,
width=\figurewidth,
height=\figureheight,
xtick={0.00, 0.05, 0.10},
xticklabels={0.00, 0.05, 0.10},
ytick={0.1,0.3,0.5,0.7,0.9},
yticklabels={},
tick align=outside,
tick pos=left,
minor tick num=1,
xmajorgrids,
x grid style={lightgray!92.02614379084967!black},
ymajorgrids,
y grid style={lightgray!92.02614379084967!black},
grid=both,
]
\addplot [line width=1pt, densely dotted, mark=*, mark options={scale=0.5}, curve_color]
table {%
6.719993026398624e-22 nan
6.665316903459876e-19 0.0009724826081033422
2.97623128552239e-16 0.005420303895387057
7.879862479185802e-14 0.01997969410793732
1.370240120273516e-11 0.0557165872980403
1.635884497269898e-09 0.1269257881611694
1.358943957125825e-07 0.2463795337231618
7.767573002960883e-06 0.4157061816521863
0.0002933622719592979 0.6171866493424731
0.006677016170522992 0.8200720694147553
};
\addplot [line width=1pt, dash pattern=on 1pt off 3pt on 3pt off 3pt, mark=*, mark options={scale=0.5}, curve_color]
table {%
6.881272859032191e-19 nan
3.384645278437268e-16 0.00630190012872156
7.497855080153379e-14 0.02381243872917049
9.854219450700442e-12 0.06256857690887241
8.513298619989403e-10 0.1314979980452742
5.055887969625934e-08 0.2377264945810731
2.093432724305684e-06 0.3814615777768954
5.985017168895149e-05 0.5498045088341748
0.001138338289751952 0.7199727278802373
0.01326232137031959 0.8724362254128779
};
\addplot [line width=1pt, solid, mark=*, mark options={scale=0.5}, curve_color]
table {%
6.562493189842408e-15 nan
1.259099712385065e-12 0.02443476521129384
1.089462709323233e-10 0.06576078483864134
5.603062998790071e-09 0.1321636622050834
1.899195528533694e-07 0.2269848880593168
4.442648390794559e-06 0.3519069280521083
7.290840671982885e-05 0.4992085822967299
0.0008350126143509957 0.6509508613845137
0.006493654244964363 0.7891537605637114
0.03239255069021207 0.9052698536549713
};
\addplot [line width=1.5pt, dash pattern=on 1pt off 3pt on 3pt off 3pt, black]
table {%
0.001				0.10
0.005				0.50
0.01                0.65
0.02				0.80
0.05                0.93
0.10                0.99
0.20                1.00
};
 \addplot [line width=1.5pt, dotted, black]
 table {%
 0                  0
 1                  1
 };
\end{axis}

\end{tikzpicture}}
    \begin{minipage}[t]{16em}
		\vspace*{0.6em}\caption{\parbox[t]{6.5em}{$\adsFPR{\adsCount}=0.20$ \\ $\adsTPR{\adsCount}=1.00$}}
      \label{fig:roc-researching:cntads}
    \end{minipage}%
  \end{subfigure}
  
  \vspace*{-4.2em}
  \begin{subfigure}[b]{.43\columnwidth}
    \setlength\figureheight{2in}
    \begin{minipage}[b]{1\textwidth}
      \centering
      \hspace*{24.75em}
      \resizebox{!}{1.5em}{\begin{tikzpicture}
\node at (0,0)[
  scale=1,
  anchor=south,
  text=black,
  rotate=0
]{\csdFPR};
\end{tikzpicture}}\vspace*{-0.4em}
    \end{minipage}
  \end{subfigure}
   \vspace*{2.7em} 
  \caption{\textit{Researching} attacker.  Baseline: $\setSize{\pwdDist{\accountId{}}}=4$, $\zipfShape{\accountId{}}=1$, 
    $\nmbrResponders{\accountId{}}=10$, $\setSize{\afaSet{\accountId{}}}=0$, $(\adsFPR{\adsCount}, \adsTPR{\adsCount})=(0.10, 0.99)$.}
  \label{fig:roc-researching}
\end{figure*}

\subsubsection{Trading off \csdTPR and \csdFPR}
\label{sec:stuffing:eff:rocs}
Given the above MDPs and resulting \csdTPR and \csdFPR measures, we
now explore how they vary together as \attackWidth is varied, for
fixed $(\adsFPR{\adsCollect}, \adsTPR{\adsCollect})$ and
$(\adsFPR{\adsCount}, \adsTPR{\adsCount})$ pairs and parameters
\zipfShape{\accountId{}}, \setSize{\pwdDist{\accountId{}}},
\nmbrResponders{\accountId{}}, and \setSize{\afaSet{\accountId{}}}.  The
$(\adsFPR{\adsCollect}, \adsTPR{\adsCollect})$ and
$(\adsFPR{\adsCount}, \adsTPR{\adsCount})$ pairs we consider were
drawn by inspection from ROC curves published by Freeman et
al.~\cite[\figrefstatic{4b}]{freeman2016:ads} for their ADS, for two
categories of attackers: a \textit{researching} attacker who issues
login attempts from the legitimate user's country (presumably after
researching that user), and a \textit{phishing} attacker who issues
login attempts both from the legitimate user's country and presenting
the same useragent string as the legitimate user would (presumably
after phishing the user).  In particular, phishing attackers are the
most powerful attackers considered by Freeman et al.  The curves
labeled $(\adsFPR{}, \adsTPR{})$ in \figref{fig:roc-phishing} and
\figref{fig:roc-researching} depict the ROC curves reported by Freeman
et al.  for \textit{phishing} and \textit{researching} attackers,
respectively.

\figref{fig:roc-phishing} shows representative ROC curves for a
\textit{phishing} attacker, and \figref{fig:roc-researching} shows
curves for a \textit{researching} attacker.  ``Baseline''
configurations, detailed in each figure's caption, are shown in
\figref{fig:roc-phishing:baseline} and
\figref{fig:roc-researching:baseline}. Each figure to the right of the
baseline shows the effects of strengthening security in one parameter,
starting from the baseline.  So, for example, starting from the
baseline, \figref{fig:roc-phishing:zipf} shows the effects of users
choosing passwords more uniformly (by changing
$\zipfShape{\accountId{}} = 1$ to $\zipfShape{\accountId{}} = 0$).
Similarly, \figref{fig:roc-phishing:pwds} shows the effects, again
starting from the baseline, of a user leveraging five passwords versus
only four (i.e., by changing $\setSize{\pwdDist{\accountId{}}}=4$ to
$\setSize{\pwdDist{\accountId{}}}=5$).

These ROC curves suggest that our credential-stuffing detector can be
highly effective in detecting credential stuffing without impinging
substantially on usability.  Notably, our detector is more effective
than simply using a state-of-the-art ADS~\cite{freeman2016:ads} for a
wide range of parameter settings.

Choosing a good operating point for our design depends on how a
credential-stuffing detection is treated at the detecting site.  An
aggressive response such as locking the account pending a password
reset (performed after two-factor authentication if deployed, or a
different intervention if not) would favor keeping \csdFPR small,
e.g., $\csdFPR < 0.05$.  A less aggressive response, such as invoking
two-factor authentication on every login attempt until the password is
reset, might allow a higher \csdFPR, e.g., $0.05 \le \csdFPR < 0.10$.
Simply warning the user might permit an even higher \csdFPR.

\section{The \DirectoryTerm}
\label{sec:directory}

Our framework in \secref{sec:stuffing} requires one website to run PMT
protocols as a \requesterTerm with other sites where the same user has
accounts.  This capability is similar to that implemented in previous
work~\cite{wang2019:reuse} using a \textit{\directoryTerm} that
stores, per account identifier \accountId{}, an address (possibly a
pseudonym) to contact each site where the account \accountId{} exists.
Assuming a one-round PMT protocol (as in \secref{sec:pmt}), the
\directoryTerm receives a PMT query from a \requesterTerm for an
account \accountId{} and forwards a copy of this query to each site
with the same account.  The \directoryTerm then receives every
\responderTerm's reply, permutes them randomly, and forwards the
responses back to the \requesterTerm in a batch.  By shuffling the
responses, the \directoryTerm ensures that the \requesterTerm learns
only the number of \respondersTerm that returned \boolTrue
(respectively, \boolFalse), not which ones, for good measure.  (The
\directoryTerm learns nothing about the private inputs to/outputs from
the PMT protocol by \requestersTerm and \respondersTerm.)  Since our
goal here is not to innovate in the design of scalable \directoryTerm
services---itself a topic with a long history, with many deployments
that far surpass our needs here, e.g.,~\cite{decandia2007:dynamo,
  noghabi2016:ambry}---we largely adopt this design in our
implementation (see \secref{sec:performance:impl}).  Below we address
two concerns about such a \directoryTerm specifically in our context,
however, namely the potentials for privacy risks and denials of
service.

\subsection{Privacy}
\label{sec:directory:privacy}

Among the design goals adopted in previous work~\cite{wang2019:reuse}
is hiding the identity of the \requesterTerm from the \respondersTerm
and the identity of each \responderTerm from other \respondersTerm and
the \requesterTerm.  The purpose of doing so is hiding where the user
has accounts, a property termed ``\accountLocationPrivacy''.  To this
end, the \requesterTerm and \respondersTerm either trust the
\directoryTerm to hide their identities (as an anonymizing proxy,
cf.,~\cite{boyan1997:anonymizer, gabber1999:lpwa}) or communicate with
the \directoryTerm using Tor~\cite{dingledine2004:tor}.

Unfortunately, \accountLocationPrivacy is impossible in our framework
against an active attacker: an attacker can attempt a login on account
\accountId{} at a site \siteId{} with a truly random password
\password{}, and if \accountId{} exists at \siteId{} and the attempt
is deemed abnormal, \password{} will be added to
\siteId{}{\suspiciousArray{\accountId{}}} under \pwdRule (or
\afaRule).  The attacker can then attempt to use \password{} in the
PMT protocol as a \requesterTerm, thereby learning whether some
\responderTerm returns a \boolTrue result; if so, then apparently
\accountId{} exists at \siteId{}.  This attack is of academic interest
only, since in practice, an attacker could equally easily determine
whether \accountId{} exists at \siteId{} by simply trying to establish
account \accountId{} at \siteId{}; most sites will inform the attacker
if \accountId{} already exists.  Still, our framework only further
renders irrelevant any attempts to hide where the user has accounts.

We thus settle for a weaker notion of privacy here, namely hiding the
identity of the \requesterTerm only, which will at least hide the site
at which the user is presently logging in.  As such, while in our
design the \requesterTerm still communicates to the \directoryTerm
using Tor if it does not trust the \directoryTerm to protect its
identity, there is no point in the \respondersTerm doing so; the
\respondersTerm receive requests directly from the \directoryTerm and
respond directly to it.  We refer to the model in which
\requestersTerm contact the \directoryTerm directly as
\trustedForLoginPrivacy (``trusted for login privacy''), and the model
in which \requestersTerm contact the \directoryTerm using Tor as
\untrustedForLoginPrivacy (``untrusted for login privacy'').

\subsection{Denials of Service}
\label{sec:directory:dos}

Like any critical service (cf., DNS), the \directoryTerm should employ
state-of-the-art defenses against blunt denial-of-service (DoS)
attempts (e.g., request overloading).
  If the \directoryTerm succumbs to such a DoS, then detecting
  credential stuffing will not be possible while the \directoryTerm is
  offline, and a site will incur a delay awaiting a timeout on the
  \directoryTerm for any login attempt with the correct password but
  for which $\adsDetected{\adsCount} = \boolTrue$.
  If the \directoryTerm is responsible for providing the salt for an
  account to each site having that account (see
  \secref{sec:stuffing:assumptions}), then a site \siteId{} with a
  newly created account \accountId{} will also be delayed in
  populating its \siteId{}{\suspiciousArray{\accountId{}}} set until
  the \directoryTerm recovers.

Our main concern here is whether the \directoryTerm introduces DoS
risks based on its particular functionality.  One such DoS risk
is associated with the process by which a website \siteId{} informs
the \directoryTerm that the user with identifier \accountId{} has
registered an account at \siteId{} and so \siteId{} should now be
consulted as a \responderTerm for \accountId{} in the framework of
\secref{sec:stuffing}.  The risk lies primarily in malicious actors
falsifying such registrations, e.g., potentially registering millions
of sites per identifier \accountId{}.

In our envisioned method of deploying our framework, this risk can be
managed.  For example, in
\secref{sec:performance:results:scalability}, we evaluate the
scalability of our design to support the U.S.\ airline, hotel, retail,
and consumer banking industries.  For a deployment by these
industries, the websites permitted to register as a \responderTerm for
an account \accountId{} can be limited to approved members of these
industry consortia.  The \directoryTerm can then limit each approved
member to at most one such registration for \accountId{}.  In doing
so, the \directoryTerm can enforce a limit on the number of site
registrations per account \accountId{}.  Moreover, owing to the
security guarantees of our framework (specifically, see
\secref{sec:pmt:security:requester}), a website has no motivation to
register for an account \accountId{} superfluously, since it learns
nothing as a \responderTerm in the protocol (except that the user for
\accountId{} is active at some website).

That said, if further limiting the registrations for account
\accountId{} is desirable, then the \directoryTerm can leverage the
online presence of the user when creating account \accountId{} at site
\siteId{} to confirm the request for \siteId{} to register as a
\responderTerm for \accountId{} at the \directoryTerm.  For example,
the \directoryTerm can send a confirmation email to the email address
\accountId{}, asking her to confirm that she created an account at
\siteId{}.  The registration attempt at the \directoryTerm is then
deferred until the user confirms it.

Not only do we contend that the \directoryTerm is not particularly
vulnerable to DoS, but it can also help in mitigating other DoS risks
of our framework:
\begin{itemize}[nosep,leftmargin=1em,labelwidth=*,align=left]
\item \textit{Defending \requestersTerm}: The primary
  DoS threat to a \requesterTerm is the possibility that some
  \respondersTerm always return a PMT protocol result indicating
  membership holds, increasing \csdFPR accordingly.  However, the
  \directoryTerm can ``audit'' \respondersTerm by issuing queries as a
  \requesterTerm itself with a truly random password, which should
  garner a \boolFalse result from every \responderTerm.  Any
  \responderTerm whose response generates a \boolTrue result is
  detected as misbehaving.

\item \textit{Defending \respondersTerm}:
  Permitting PMT queries against
  \siteId{}{\suspiciousArray{\accountId{}}} sets raises the
  possibility that an attacker will perform queries repeatedly to
  discover the contents of those sets.  (In particular, recall that
  \afaRule permits \siteId{}{\passwordArray{\accountId{}}} to be added
  to \siteId{}{\suspiciousArray{\accountId{}}}.)  A \responderTerm
  thus should rate-limit PMT queries, just as it would regular login
  attempts, to stem such online dictionary attacks.  However, for
  accounts experiencing an unusually high rate of queries, the
  \directoryTerm can pose CAPTCHAs~\cite{vonahn2004:captchas} back to
  the \requestersTerm as a precondition to forwarding their queries to
  \respondersTerm.  In this way, the limited PMT budgets of
  \respondersTerm can be allocated preferentially to \requestersTerm
  with real users, preventing bots from starving those
  \requestersTerm.
\end{itemize}

\section{Privately Testing Set Membership}
\label{sec:pmt}

An ingredient of our framework in \secref{sec:stuffing} is a protocol
by which a \requesterTerm \siteId{}, having received password
$\password{} = \siteId{}{\passwordArray{\accountId{}}}$ in a login
attempt for account \accountId{}, inquires with a \responderTerm
\siteIdAlt to determine whether $\password{} \in
\siteIdAlt{\suspiciousArray{\accountId{}}}$.  Because $\password{} =
\siteId{}{\passwordArray{\accountId{}}}$, it is important that the
protocol not disclose \password{} to \siteIdAlt.  Moreover, since
\siteIdAlt{\suspiciousArray{\accountId{}}} might contain
\siteIdAlt{\passwordArray{\accountId{}}} (see
\secref{sec:stuffing:algorithm}) or passwords similar to it, the
protocol should not divulge \siteId{}{\suspiciousArray{\accountId{}}}
to \siteId{}.  This specification is met by a \textit{private
  membership test} (PMT) protocol.

\subsection{The Need for a New Protocol}
\label{sec:pmt:need}

Several PMT protocols have been proposed
(e.g.,~\cite{nojima2009:bloom, meskanen2015:private, tamrakar2017:pmt,
  ramezanian2017:private, wang2019:reuse}).  In addition, PMT
protocols are closely related to private set-intersection (PSI;
surveyed by Pinkas et al.~\cite{pinkas2018:psi}) and private
set-intersection cardinality protocols (PSI-CA;
e.g.,~\cite{davidson2017:psi-ca, decristofaro2012:intersection,
  debnath2015:psi-ca, egert2015:union, kissner2005:psi-ca}).  In
particular, having the \requesterTerm in a PSI/PSI-CA protocol prove
in zero knowledge that its input is a set of size one yields a PMT
protocol.

Considering the additional requirements of our framework in
\secsref{sec:stuffing}{sec:directory} somewhat narrows the options for
implementing our PMT, however.  First, because our threat model
permits the \requesterTerm or \responderTerm to misbehave arbitrarily,
we require a protocol that accommodates the malicious behavior of
either party while still protecting the privacy of each party's input
to the protocol.  Second, minimizing rounds of communication in the
protocol is critical for the scalability of our framework, since these
rounds (each with a different website as \responderTerm) will traverse
wide-area links and---in the \untrustedForLoginPrivacy model (see
\secref{sec:directory:privacy})---an anonymous communication channel,
which will add even more overhead to each round.  For the same reason,
we wish to leverage bandwidth-efficient protocols to the extent
possible, and because \respondersTerm may need to respond to
significant numbers of PMT queries (as we will analyze in
\secref{sec:performance:results:scalability}), computational
efficiency for the \responderTerm is a secondary but still important
concern.

To our knowledge, among PSI protocols that are secure against
malicious behaviors (e.g.,~\cite{dachman2009:mal-psi,
  decristofaro2010:mal-psi, freedman2004:mal-psi,
  freedman2016:mal-psi, hazay2012:mal-psi, kamara2014:mal-psi,
  rindal2017:mal-psi, rindal2017:mal-psi2, thomas2019:stuffing,
  li2019:stuffing}), only those of De Cristofaro et
al.~\cite{decristofaro2010:mal-psi} and of Thomas et
al.~\cite{thomas2019:stuffing} and Li et al.~\cite{li2019:stuffing}
execute in one round.  However, the responses in these protocols are
of size $O(\pmtSetSize)$ ciphertexts for a set of size $\pmtSetSize$.
While there are several one-round PSI-CA protocols
(e.g.,~\cite{davidson2017:psi-ca, decristofaro2012:intersection,
  debnath2015:psi-ca, egert2015:union}), we are aware of none that
address malicious parties (without introducing a trusted third party,
cf.,~\cite{decristofaro2012:intersection,debnath2015:psi-ca}).

One strategy to improve performance has been to weaken security in
quantified ways against malicious parties.  For example, for an
integer $\genericBucketCapacity \le \pmtSetSize$, Thomas et
al.~\cite{thomas2019:stuffing} and Li et al.~\cite{li2019:stuffing}
explored protocols in which the \requesterTerm leaks $\log_2
\genericBucketCapacity$ bits of the \requesterTerm's input, in
exchange for reducing the response size to
$O(\pmtSetSize/\genericBucketCapacity)$ ciphertexts.
However, a protocol that gains efficiency by leaking information
only in the other direction (from \responderTerm to \requesterTerm)
is arguably more appropriate for our context, since the
\requesterTerm \siteId{} invokes the protocol with the correct
password, i.e., \siteId{}{\passwordArray{\accountId{}}}.
Ramezanian et al.~\cite{ramezanian2017:private} and Wang \&
Reiter~\cite{wang2019:reuse} proposed protocols whereby the
\responderTerm learns nothing about the \requesterTerm's element, but
the \requesterTerm learns more information about the \responderTerm's
set than just the truth of its membership query.  Specifically, in the
Ramezanian et al.\ protocol~\cite{ramezanian2017:private}, the
\responderTerm leaks its set to the \requesterTerm over
$O(\pmtSetSize/\genericBucketCapacity)$ responses, each of
$O(\genericBucketCapacity \log_2 \frac{1}{\pmtFPR})$ bits in size,
where \pmtFPR is a tunable false positive rate for the membership
test.  The Wang \& Reiter protocol~\cite{wang2019:reuse} leaks the
\responderTerm's set to a malicious \requesterTerm over $O(\pmtSetSize
\log_2 \frac{1}{\pmtFPR})$ responses, each of size only one
ciphertext.  The protocol that we propose here also
allows a malicious \requesterTerm to learn the \responderTerm's set
faster than the ideal---but only after $\Omega(\frac{1}{\pmtFPR})$
responses, much better than the Ramezanian et al.\ and Wang \& Reiter
protocols.  (Below we term this measure the ``extraction complexity''
of the protocol, and justify this claim in
\secref{sec:pmt:security:responder}.)  The request and response sizes
of our protocol are only $O(\pmtSetSize/\genericBucketCapacity)$ and
$O(\genericBucketCapacity)$ ciphertexts, respectively.

\subsection{Partially Homomorphic Encryption}
\label{sec:pmt:crypto}

Our protocol builds on a partially homomorphic encryption scheme
$\encScheme = \langle \keygen, \encrypt{}, \encZeroTest{},
\encAdd{[\cdot]}\rangle$ with these algorithms:

\begin{itemize}[nosep,leftmargin=1em,labelwidth=*,align=left]
\item \keygen is a randomized algorithm that on input $1^\secParam$
  outputs a public-key/private-key pair $\langle\pubKey,
  \privKey\rangle \gets \keygen(1^{\secParam})$.  The value of \pubKey
  identifies a prime \fieldOrder for which the \textit{plaintext
    space} for encrypting with \pubKey is the finite field $\langle
  \residues{\fieldOrder}, \fieldAdd, \fieldMult\rangle$ where
  \fieldAdd and \fieldMult are addition and multiplication modulo
  \fieldOrder, respectively.  For clarity below, we denote the
  additive identity by \fieldAddIdentity, the multiplicative identity
  by \fieldMultIdentity, and the additive inverse of $\plaintext \in
  \residues{\fieldOrder}$ by $\fieldNegative\plaintext$.  \pubKey also
  determines a \textit{ciphertext space} $\ciphertextSpace{\pubKey} =
  \bigcup_{\plaintext \in
    \residues{\fieldOrder}} \ciphertextSpace{\pubKey}(\plaintext)$,
  where $\ciphertextSpace{\pubKey}(\plaintext)$ denotes the
  ciphertexts for plaintext $\plaintext \in \residues{\fieldOrder}$.
\item \encrypt{} is a randomized algorithm that on input public key
  \pubKey and a plaintext $\plaintext \in \residues{\fieldOrder}$,
  outputs a ciphertext $\ciphertext{} \gets
  \encrypt{\pubKey}(\plaintext)$ chosen uniformly at random from
  $\ciphertextSpace{\pubKey}(\plaintext)$.
\item \encAdd{[\cdot]} is a randomized algorithm that, on input a
  public key \pubKey and ciphertexts $\ciphertext{1} \in
  \ciphertextSpace{\pubKey}(\plaintext{1})$ and $\ciphertext{2} \in
  \ciphertextSpace{\pubKey}(\plaintext{2})$, outputs a ciphertext
  $\ciphertext{} \gets \ciphertext{1} \encAdd{\pubKey} \ciphertext{2}$
  chosen uniformly at random from
  $\ciphertextSpace{\pubKey}(\plaintext{1}
  \fieldAdd \plaintext{2})$.
\item \encZeroTest{} is a deterministic algorithm that on input a
  private key \privKey and ciphertext $\ciphertext{}
  \in \ciphertextSpace{\pubKey}$, outputs a boolean
  $\encZeroTestResult \gets \encZeroTest{\privKey}(\ciphertext{})$
  where $\encZeroTestResult = \boolTrue$ iff $\ciphertext{}
  \in \ciphertextSpace{\pubKey}(\fieldAddIdentity)$.
\end{itemize}

Note that our protocol does not require an efficient decryption
capability.  Indeed, the instantiation of this scheme that we
leverage, described in \appref{sec:pmt:elgamal}, does not support
one---though it does support an efficient \encZeroTest{} calculation.

\subsection{Additional Operators}
\label{sec:pmt:operators}

To express our protocol, it will be convenient to define a few
additional operators involving ciphertexts.  These additional
operators can all be expressed using the operators given in
\secref{sec:pmt:crypto}, and so require no new functionality from the
cryptosystem.  Below, ``$\genericRV \distEqual \genericRVAlt$''
denotes that random variables \genericRV and \genericRVAlt are
distributed identically; ``$\genericMatrix
\in (\genericSet)^{\matrixNmbrRows \times \matrixNmbrCols}$'' means
that \genericMatrix is an \matrixNmbrRows-row, \matrixNmbrCols-column
matrix of elements in the set \genericSet; and
``\matrixComponent{\genericMatrix}{\matrixRowIdx}{\matrixColIdx}''
denotes the row-\matrixRowIdx, column-\matrixColIdx element of the
matrix \genericMatrix.
\begin{itemize}[nosep,leftmargin=1em,labelwidth=*,align=left]
\item \encSum{\pubKey} denotes summing a sequence
  using \encAdd{\pubKey}, i.e.,
  \[
    \encSum{\pubKey}{\encSumIdx=1}{\genericNat} \ciphertext{\encSumIdx}
    \distEqual \ciphertext{1} \encAdd{\pubKey} \ciphertext{2} \encAdd{\pubKey}
    \ldots \encAdd{\pubKey} \ciphertext{\genericNat}
  \]

\item If $\ciphertextMatrix \in
  (\ciphertextSpace{\pubKey})^{\matrixNmbrRows \times
  \matrixNmbrCols}$ and $\ciphertextMatrixAlt \in
  (\ciphertextSpace{\pubKey})^{\matrixNmbrRows \times
  \matrixNmbrCols}$, then $\ciphertextMatrix \encMatrixAdd{\pubKey}
  \ciphertextMatrixAlt \in
  (\ciphertextSpace{\pubKey})^{\matrixNmbrRows \times
    \matrixNmbrCols}$ is the result of component-wise addition using
  \encAdd{\pubKey}, i.e., so that
  \[
  \matrixComponent{\ciphertextMatrix \encMatrixAdd{\pubKey}
    \ciphertextMatrixAlt}{\matrixRowIdx}{\matrixColIdx} \distEqual
  \matrixComponent{\ciphertextMatrix}{\matrixRowIdx}{\matrixColIdx}
  \encAdd{\pubKey}
  \matrixComponent{\ciphertextMatrixAlt}{\matrixRowIdx}{\matrixColIdx}
  \]

\item If $\plaintextMatrix \in (\residues{\fieldOrder})^{\matrixNmbrRows
  \times \matrixNmbrCols}$ and $\ciphertextMatrix \in
  (\ciphertextSpace{\pubKey})^{\matrixNmbrRows \times
  \matrixNmbrCols}$, then $\plaintextMatrix
  \encMatrixHadamardMult{\pubKey} \ciphertextMatrix \in
  (\ciphertextSpace{\pubKey})^{\matrixNmbrRows \times
    \matrixNmbrCols}$ is the result of Hadamard (i.e., component-wise)
  ``scalar multiplication'' using repeated application of
  \encAdd{\pubKey}, i.e., so that
  \[
  \matrixComponent{\plaintextMatrix \encMatrixHadamardMult{\pubKey}
    \ciphertextMatrix}{\matrixRowIdx}{\matrixColIdx} \distEqual
  \encSum{\pubKey}{\encSumIdx =
    1}{\matrixComponent{\plaintextMatrix}{\matrixRowIdx}{\matrixColIdx}}
  \matrixComponent{\ciphertextMatrix}{\matrixRowIdx}{\matrixColIdx}
  \]

\item If $\plaintextMatrix \in (\residues{\fieldOrder})^{\matrixNmbrRows \times
  \matrixNmbrCols}$ and $\ciphertextMatrix \in
  (\ciphertextSpace{\pubKey})^{\matrixNmbrCols \times
  \matrixNmbrColsAlt}$, then $\plaintextMatrix
  \encMatrixScalarMult{\pubKey} \ciphertextMatrix \in
  (\ciphertextSpace{\pubKey})^{\matrixNmbrRows \times
    \matrixNmbrColsAlt}$ is the result of standard matrix
  multiplication using \encAdd{\pubKey} and ``scalar multiplication''
  using repeated application of \encAdd{\pubKey}, i.e., so that
  \[
  \matrixComponent{\plaintextMatrix \encMatrixScalarMult{\pubKey}
    \ciphertextMatrix}{\matrixRowIdx}{\matrixColIdx} \distEqual
  \encSum{\pubKey}{\encSumIdx = 1}{\matrixNmbrCols}
  \encSum{\pubKey}{\encSumIdxAlt =
    1}{\matrixComponent{\plaintextMatrix}{\matrixRowIdx}{\encSumIdx}}
  \matrixComponent{\ciphertextMatrix}{\encSumIdx}{\matrixColIdx}
  \]
\end{itemize}

\subsection{Cuckoo Filters}
\label{sec:pmt:cuckoo}

Our PMT protocol, called \cuckooPMT, uses a cuckoo
filter~\cite{fan2014:cuckoo} as an underlying building block.  A
cuckoo filter is a set representation that supports insertion and
deletion of elements, as well as testing membership.  The cuckoo
filter uses a ``fingerprint'' function $\cuckooFingerprintFn:
\{0,1\}^* \rightarrow \cuckooFingerprintRange$ and a hash function
$\cuckooHashFn: \{0, 1\}^* \rightarrow \nats{\cuckooNmbrBuckets}$,
where for an integer \genericNat, the notation ``\nats{\genericNat}''
denotes $\{1, \ldots, \genericNat\}$, and where \cuckooNmbrBuckets is
a number of ``buckets''.  We require that $\cuckooFingerprintRange
\subset \residues{\fieldOrder} \setminus \{\fieldAddIdentity\}$ for
any \fieldOrder determined by $\langle\pubKey, \privKey\rangle \gets
\keygen(1^{\secParam})$, and that members of \cuckooFingerprintRange
can be distinguished from members of $\residues{\fieldOrder} \setminus
\cuckooFingerprintRange$ using a public predicate.  (For example,
defining \cuckooFingerprintRange to be the odd elements of
\residues{\fieldOrder} would suffice.)  For an integer bucket
``capacity'' \cuckooBucketCapacity, the cuckoo filter data structure
is a \cuckooBucketCapacity-row, \cuckooNmbrBuckets-column matrix
\cuckooFilter of elements in \residues{\fieldOrder}, i.e.,
$\cuckooFilter \in (\residues{\fieldOrder})^{\cuckooBucketCapacity
  \times \cuckooNmbrBuckets}$.  Then, the cuckoo filter
contains the element \pmtSetElmt if and only if there exists
$\matrixRowIdx \in \nats{\cuckooBucketCapacity}$ such that either
\begin{align}
  \matrixComponent{\cuckooFilter}{\matrixRowIdx}{\cuckooHashFn(\pmtSetElmt)}
  & = \cuckooFingerprintFn(\pmtSetElmt) & \mbox{or} \label{eqn:cuckooOptionOne}\\
  \matrixComponent{\cuckooFilter}{\matrixRowIdx}{\cuckooHashFn(\pmtSetElmt) \oplus \cuckooHashFn(\cuckooFingerprintFn(\pmtSetElmt))}
  & = \cuckooFingerprintFn(\pmtSetElmt) \label{eqn:cuckooOptionTwo}
\end{align}
Cuckoo filters permit false positives
(membership tests that return \boolTrue for elements not previously
added or already removed) with a probability \pmtFPR that, for fixed
\cuckooBucketCapacity, can be decreased by increasing the size of
\cuckooFingerprintRange~\cite{fan2014:cuckoo}.

\subsection{Protocol Description}
\label{sec:pmt:protocol}

Our protocol is illustrated in \figref{fig:protocol}, where the steps
performed by the \requesterTerm \requester with input \pmtSetElmt are
shown on the left in
\linesref{prot:requester:init}{prot:requester:return} (in addition to
sending \msgref{prot:msg:request}), and the steps performed by the
\responderTerm \responder{} with cuckoo filter \cuckooFilter are shown
on the right in
\linesref{prot:responder:checkCiphertexts}{prot:responder:makeResult}
(in addition to sending \msgref{prot:msg:response}).  The protocol
returns \boolTrue to \requester if \pmtSetElmt is in the cuckoo filter
\cuckooFilter and \boolFalse otherwise.

In our protocol, \requester creates a \cuckooNmbrBuckets-row,
$2$-column matrix \queryMatrix of ciphertexts, where the first column
contains a ciphertext of \fieldMultIdentity in row
$\cuckooHashFn(\pmtSetElmt)$ and ciphertexts of \fieldAddIdentity in
other rows, and where the second column contains a ciphertext of
\fieldMultIdentity in row $\cuckooHashFn(\pmtSetElmt) \oplus
\cuckooHashFn(\cuckooFingerprintFn(\pmtSetElmt))$ and ciphertexts of
\fieldAddIdentity in others
(\lineref{prot:requester:calculateQueryMatrix}).  The \requesterTerm
also generates a ciphertext \cuckooFingerprintCtext of
$\fieldNegative\cuckooFingerprintFn(\pmtSetElmt)$
(\lineref{prot:requester:calculateFingerprintCtext}), and sends this
ciphertext and the matrix \queryMatrix to \responder{}, along with the
public key \pubKey (\msgref{prot:msg:request}).  After checking in
\lineref{prot:responder:checkCiphertexts} that
$\cuckooFingerprintCtext \in \ciphertextSpace{\pubKey}$, and
$\queryMatrix \in (\ciphertextSpace{\pubKey})^{\cuckooNmbrBuckets
  \times 2}$ (and that \pubKey is well-formed, which is left implicit
in \figref{fig:protocol}), \responder{} generates a matrix
$\cuckooFingerprintCtextMatrix \in
(\ciphertextSpace{\pubKey})^{\cuckooBucketCapacity \times 2}$ having a
copy of \cuckooFingerprintCtext in each component
(\lineref{prot:responder:multFingerprintCtext}) and a matrix
$\plaintextMatrix \in (\residues{\fieldOrder})^{\cuckooBucketCapacity
  \times 2}$ of random elements of $\residues{\fieldOrder} \setminus
\{\fieldAddIdentity\}$ (\lineref{prot:responder:randMatrix}).
\responder{} then forms the response matrix $\resultMatrix \gets
\plaintextMatrix \encMatrixHadamardMult{\pubKey}
\left(\left(\cuckooFilter \encMatrixScalarMult{\pubKey}
\queryMatrix\right) \encMatrixAdd{\pubKey}
\cuckooFingerprintCtextMatrix\right)$, which is best understood
component-wise:
$\matrixComponent{\resultMatrix}{\matrixRowIdx}{\matrixColIdx}$ is a
ciphertext of a random element of $\residues{\fieldOrder} \setminus
\{\fieldAddIdentity\}$ if \matrixComponent{\left(\cuckooFilter
  \encMatrixScalarMult{\pubKey} \queryMatrix\right)
  \encMatrixAdd{\pubKey}
  \cuckooFingerprintCtextMatrix}{\matrixRowIdx}{\matrixColIdx} is a
ciphertext of anything other than \fieldAddIdentity, since
\matrixComponent{\plaintextMatrix}{\matrixRowIdx}{\matrixColIdx} is
chosen at random from $\residues{\fieldOrder} \setminus
\{\fieldAddIdentity\}$.  Moreover,
\matrixComponent{\left(\cuckooFilter \encMatrixScalarMult{\pubKey}
  \queryMatrix\right) \encMatrixAdd{\pubKey}
  \cuckooFingerprintCtextMatrix}{\matrixRowIdx}{\matrixColIdx} is an
encryption of \fieldAddIdentity iff \matrixComponent{\cuckooFilter
  \encMatrixScalarMult{\pubKey}
  \queryMatrix}{\matrixRowIdx}{\matrixColIdx} is a ciphertext of
$\cuckooFingerprintFn(\pmtSetElmt)$, since
\matrixComponent{\cuckooFingerprintCtextMatrix}{\matrixRowIdx}{\matrixColIdx}
is a ciphertext of $\fieldNegative\cuckooFingerprintFn(\pmtSetElmt)$.
And \matrixComponent{\cuckooFilter \encMatrixScalarMult{\pubKey}
  \queryMatrix}{\matrixRowIdx}{\matrixColIdx} is a ciphertext of
$\cuckooFingerprintFn(\pmtSetElmt)$ iff either
\eqnref{eqn:cuckooOptionOne} holds (since
\matrixComponent{\queryMatrix}{\cuckooHashFn(\pmtSetElmt)}{1} is an
encryption of \fieldMultIdentity) or \eqnref{eqn:cuckooOptionTwo}
holds (since \matrixComponent{\queryMatrix}{\cuckooHashFn(\pmtSetElmt)
  \oplus \cuckooHashFn(\cuckooFingerprintFn(\pmtSetElmt))}{2} is an
encryption of \fieldMultIdentity).  So, if \requester and \responder{}
behave correctly, the protocol returns \boolTrue to \requester iff
\pmtSetElmt is an element of the cuckoo filter \cuckooFilter.

\newcolumntype{R}{>{$}p{0.36in}<{$}}
\newcolumntype{M}{>{$}p{0.36in}<{$}}
\begin{figure}[t]
  \centering
  \setcounter{requesterLineNmbr}{0}
  \setcounter{responderLineNmbr}{0}
  \setcounter{messageNmbr}{0}
  \framebox{\small
     $\begin{array}{@{}r@{\hspace{2pt}}R@{}r@{}M@{}r@{\hspace{2pt}}l@{}} 
      & \multicolumn{1}{c}{\mathrlap{~~~~~\requester(\pmtSetElmt)}}
      & & & & \multicolumn{1}{c}{\responder{}(\cuckooFilter)} \\[10pt]
      \requesterLabel{prot:requester:init}
      & \mathrlap{\langle\pubKey, \privKey\rangle \gets \keygen(1^{\secParam})}
      \\
      \requesterLabel{prot:requester:calculateFingerprintCtext}
      & \mathrlap{\cuckooFingerprintCtext \gets
        \encrypt{\pubKey}(\fieldNegative\cuckooFingerprintFn(\pmtSetElmt))}
      \\
      \requesterLabel{prot:requester:bucketIndexOne}
      & \mathrlap{\cuckooBucketOne \gets \cuckooHashFn(\pmtSetElmt)}
      \\ 
      \requesterLabel{prot:requester:bucketIndexTwo}
      & \mathrlap{\cuckooBucketTwo \gets \cuckooHashFn(\pmtSetElmt) \oplus 
        \cuckooHashFn(\cuckooFingerprintFn(\pmtSetElmt))}
      \\
      \requesterLabel{prot:requester:calculateQueryMatrix}
      & \mathrlap{\forall\matrixRowIdx \in \nats{\cuckooNmbrBuckets}:} \\
      & \mathrlap{\matrixComponent{\queryMatrix}{\matrixRowIdx}{1} \gets
        \left\{
           \begin{array}{ll}
             \encrypt{\pubKey}(\fieldMultIdentity) &
             \mbox{if~} \matrixRowIdx = \cuckooBucketOne \\
             \encrypt{\pubKey}(\fieldAddIdentity) & \mbox{otherwise}
           \end{array}
        \right.
        }
      \\
      & \mathrlap{\matrixComponent{\queryMatrix}{\matrixRowIdx}{2} \gets
        \left\{
           \begin{array}{ll}
             \encrypt{\pubKey}(\fieldMultIdentity) &
             \mbox{if~} \matrixRowIdx = \cuckooBucketTwo \\
             \encrypt{\pubKey}(\fieldAddIdentity) & \mbox{otherwise}
           \end{array}
        \right.
        }
      \\[20pt]
      & & \messageLabel{prot:msg:request}
      & \mathrlap{\xrightarrow{\makebox[9.75em]{$\pubKey, \cuckooFingerprintCtext, \queryMatrix$}}}
      \\[12pt]
      & & & & \responderLabel{prot:responder:checkCiphertexts}
      & \codeAbort \codeIf \cuckooFingerprintCtext \not\in \ciphertextSpace{\pubKey}
        \vee \queryMatrix \not\in (\ciphertextSpace{\pubKey})^{\cuckooNmbrBuckets \times 2}
      \\
      & & & & \responderLabel{prot:responder:multFingerprintCtext}
      & \forall \matrixRowIdx \in \nats{\cuckooBucketCapacity}, \matrixColIdx \in \nats{2}:
        \matrixComponent{\cuckooFingerprintCtextMatrix}{\matrixRowIdx}{\matrixColIdx} \gets
        \cuckooFingerprintCtext \\
      & & & & \responderLabel{prot:responder:randMatrix}
      & \forall \matrixRowIdx \in \nats{\cuckooBucketCapacity}, \matrixColIdx \in \nats{2}:
        \matrixComponent{\plaintextMatrix}{\matrixRowIdx}{\matrixColIdx} \getsr \residues{\fieldOrder} \setminus \{\fieldAddIdentity\}\!\! 
      \\
      & & & & \responderLabel{prot:responder:makeResult}
      & \resultMatrix \gets \plaintextMatrix \encMatrixHadamardMult{\pubKey}
        \left(\left(\cuckooFilter \encMatrixScalarMult{\pubKey} \queryMatrix\right)
        \encMatrixAdd{\pubKey} \cuckooFingerprintCtextMatrix\right)
      \\[12pt]
      & & \messageLabel{prot:msg:response}
      & \mathrlap{\xleftarrow{\makebox[9.75em]{\resultMatrix}}}
      \\[12pt]
      \requesterLabel{prot:requester:checkCiphertexts}
      & \mathrlap{\codeAbort \codeIf \resultMatrix \not\in (\ciphertextSpace{\pubKey})^{\cuckooBucketCapacity \times 2}}
      \\
      \requesterLabel{prot:requester:return}
      & \mathrlap{\codeReturn \displaystyle{\bigvee_{\matrixRowIdx, \matrixColIdx}}~
        \encZeroTest{\privKey}(\matrixComponent{\resultMatrix}{\matrixRowIdx}{\matrixColIdx})}
      \\
    \end{array}$
  }
  \caption{\cuckooPMT with \requesterTerm \requester and \responderTerm
    \responder{}, described in \secref{sec:pmt}.  Matrix types
    are: $\queryMatrix \in
    (\ciphertextSpace{\pubKey})^{\cuckooNmbrBuckets \times 2}$;
    $\cuckooFilter \in (\residues{\fieldOrder})^{\cuckooBucketCapacity
      \times \cuckooNmbrBuckets}$; $\cuckooFingerprintCtextMatrix \in
    (\ciphertextSpace{\pubKey})^{\cuckooBucketCapacity \times 2}$;
    $\plaintextMatrix \in
    (\residues{\fieldOrder})^{\cuckooBucketCapacity \times 2}$; and
    $\resultMatrix \in
    (\ciphertextSpace{\pubKey})^{\cuckooBucketCapacity \times 2}$.}
  \label{fig:protocol}
\end{figure}

For (an artificially small) example, suppose $\cuckooNmbrBuckets = 3$,
$\cuckooBucketCapacity = 1$, and that the \requesterTerm \requester{}
queries the membership of element \pmtSetElmt such that
$\cuckooBucketOne = \cuckooHashFn(\pmtSetElmt) = 3$ and
$\cuckooBucketTwo = \cuckooHashFn(\pmtSetElmt) \oplus
\cuckooHashFn(\cuckooFingerprintFn(\pmtSetElmt)) = 2$.  The
\responderTerm \responder{} generates a cuckoo filter $\cuckooFilter
\in (\residues{\fieldOrder})^{1 \times 3}$ based on its input
set. Here we assume \pmtSetElmt is in \responder{}'s set, as indicated
by $\matrixComponent{\cuckooFilter}{1}{3} =
\cuckooFingerprintFn(\pmtSetElmt)$.  Then, 
\begin{align*}
\lefteqn{\cuckooFilter
  \encMatrixScalarMult{\pubKey} \queryMatrix} \\
& \distEqual 
\left[\begin{array}{lll}
    \plaintext{1} & \plaintext{2}
    & \cuckooFingerprintFn(\pmtSetElmt)
\end{array}\right]
\encMatrixScalarMult{\pubKey}
\left[\begin{array}{ll}
\encrypt{\pubKey}(\fieldAddIdentity) & \encrypt{\pubKey}(\fieldAddIdentity) \\
\encrypt{\pubKey}(\fieldAddIdentity) & \encrypt{\pubKey}(\fieldMultIdentity) \\
\encrypt{\pubKey}(\fieldMultIdentity) & \encrypt{\pubKey}(\fieldAddIdentity)
  \end{array}\right] \\
& \distEqual \left[\begin{array}{ll}
    \encrypt{\pubKey}(\cuckooFingerprintFn(\pmtSetElmt)) & \encrypt{\pubKey}(\plaintext{2})
    \end{array}\right]
\end{align*}
where \plaintext{1} and \plaintext{2} are elements of
\residues{\fieldOrder} that, barring a collision in the output of
\cuckooFingerprintFn, are not equal to
$\cuckooFingerprintFn(\pmtSetElmt)$.  So, $(\cuckooFilter
\encMatrixScalarMult{\pubKey} \queryMatrix) \encMatrixAdd{\pubKey}
\cuckooFingerprintCtextMatrix \distEqual \left[ \begin{array}{ll}
    \encrypt{\pubKey}(\fieldAddIdentity) &
    \encrypt{\pubKey}(\plaintext{3}) \end{array}\right]$ where
$\plaintext{3} = \plaintext{2} \fieldMinus
\cuckooFingerprintFn(\pmtSetElmt) \neq \fieldAddIdentity$, again
assuming $\plaintext{2} \neq \cuckooFingerprintFn(\pmtSetElmt)$, and
thus $\plaintextMatrix \encMatrixHadamardMult{\pubKey}
\left(\left(\cuckooFilter \encMatrixScalarMult{\pubKey}
\queryMatrix\right) \encMatrixAdd{\pubKey}
\cuckooFingerprintCtextMatrix \right) \distEqual
\left[ \begin{array}{ll} \encrypt{\pubKey}(\fieldAddIdentity) &
    \encrypt{\pubKey}(\plaintext{4}) \end{array}\right]$, where
\plaintext{4} is distributed uniformly in $\residues{\fieldOrder}
\setminus \{\fieldAddIdentity\}$.

\subsection{Security Against a Malicious \RequesterTerm}
\label{sec:pmt:security:responder}

If the \responderTerm follows the protocol, then the \textit{only}
information encoded in each
\matrixComponent{\resultMatrix}{\matrixRowIdx}{\matrixColIdx} is
$\encZeroTest{\privKey}(\matrixComponent{\resultMatrix}{\matrixRowIdx}{\matrixColIdx})$,
as a corollary of the following two propositions.

\begin{prop}
\label{prop:comp:responder-security-one}
If the \responderTerm follows the protocol, then
$\cprob{\big}{\matrixComponent{\resultMatrix}{\matrixRowIdx}{\matrixColIdx}
  \in \ciphertextSpace{\pubKey}(\plaintext)}{\matrixComponent{\resultMatrix}{\matrixRowIdx}{\matrixColIdx}
  \not\in \ciphertextSpace{\pubKey}(\fieldAddIdentity)} =
\frac{1}{\fieldOrder-1}$ for any $\matrixRowIdx \in
\nats{\cuckooBucketCapacity}$, $\matrixColIdx \in \nats{2}$,
and $\plaintext \in \residues{\fieldOrder} \setminus
\{\fieldAddIdentity\}$.
\end{prop}

\begin{proof}
$\matrixComponent{\left(\cuckooFilter
  \encMatrixScalarMult{\pubKey} \queryMatrix\right)
  \encMatrixAdd{\pubKey}
  \cuckooFingerprintCtextMatrix}{\matrixRowIdx}{\matrixColIdx} \in
\ciphertextSpace{\pubKey}$, since by
\lineref{prot:responder:checkCiphertexts}, $\cuckooFingerprintCtext
\in \ciphertextSpace{\pubKey}$ and $\queryMatrix \in
(\ciphertextSpace{\pubKey})^{\cuckooNmbrBuckets \times 2}$.  Moreover,
$\matrixComponent{\left(\cuckooFilter \encMatrixScalarMult{\pubKey}
  \queryMatrix\right) \encMatrixAdd{\pubKey}
  \cuckooFingerprintCtextMatrix}{\matrixRowIdx}{\matrixColIdx} \not\in
\ciphertextSpace{\pubKey}(\fieldAddIdentity)$ since 
$\matrixComponent{\resultMatrix}{\matrixRowIdx}{\matrixColIdx} \not\in
\ciphertextSpace{\pubKey}(\fieldAddIdentity)$ by assumption.  Since
$\matrixComponent{\plaintextMatrix}{\matrixRowIdx}{\matrixColIdx}$ is
drawn uniformly from $\residues{\fieldOrder} \setminus
\{\fieldAddIdentity\}$
(\lineref{prot:responder:randMatrix}), the plaintext of
\matrixComponent{\resultMatrix}{\matrixRowIdx}{\matrixColIdx} is
uniformly distributed in $\residues{\fieldOrder} \setminus
\{\fieldAddIdentity\}$.
\end{proof}

\begin{prop}
\label{prop:comp:responder-security-two}
If the \responderTerm follows the protocol, then
$\cprob{\big}{\matrixComponent{\resultMatrix}{\matrixRowIdx}{\matrixColIdx}
  =
  \ciphertext{}}{\matrixComponent{\resultMatrix}{\matrixRowIdx}{\matrixColIdx}
  \in \ciphertextSpace{\pubKey}(\plaintext)} =
\frac{1}{\setSize{\ciphertextSpace{\pubKey}(\plaintext)}}$ for any
$\matrixRowIdx \in \nats{\cuckooBucketCapacity}$, $\matrixColIdx \in
\nats{2}$,
$\plaintext \in \residues{\fieldOrder}$, and $\ciphertext{} \in
\ciphertextSpace{\pubKey}(\plaintext)$.
\end{prop}

\begin{proof}
This is immediate since \encAdd{\pubKey} ensures that for
$\ciphertext{1} \in \ciphertextSpace{\pubKey}(\plaintext{1})$ and
$\ciphertext{2} \in \ciphertextSpace{\pubKey}(\plaintext{2})$,
$\ciphertext{1} \encAdd{\pubKey} \ciphertext{2}$ outputs a ciphertext
$\ciphertext{}$ chosen uniformly at random from
$\ciphertextSpace{\pubKey}(\plaintext{1} \fieldAdd \plaintext{2})$.
\end{proof}

\propref{prop:comp:responder-security-one} and
\propref{prop:comp:responder-security-two} are also true for the
protocol of Wang \& Reiter~\cite{wang2019:reuse} (henceforth called
\bloomPMT), in that each protocol execution leaks to the
\requesterTerm only one yes/no answer about the \responderTerm's set
representation, regardless of the actions of the \requesterTerm.  A
critical distinction exists between our protocol and \bloomPMT,
however, in that \bloomPMT permits a malicious \requesterTerm to craft
queries so that the yes/no answer can be expected to carry a (full)
bit of information to the \requesterTerm about the \responderTerm's
set. We capture this information leakage using \textit{extraction
  complexity}, which is the expected number of queries for a malicious
\requesterTerm to extract the \responderTerm's set representation,
thereby enabling the \requesterTerm to conduct offline attacks on the
set.  More precisely, for a fixed \responderTerm set \pmtSet, the
\textit{\extractionComplexity} of a PMT protocol is the expected
number of protocol runs required for a malicious \requesterTerm to
extract enough information from an honest \responderTerm to locally
determine $\pmtSetElmt \testIn \pmtSet$ for any \pmtSetElmt with the
same accuracy as the PMT provides.

\bloomPMT enables a malicious \requesterTerm to learn any single bit
in the Bloom-filter representation of the \responderTerm's set.  Since
to accommodate a set of size \pmtSetSize with false-positive rate of
\pmtFPR for membership tests, a Bloom filter uses $O(\pmtSetSize
\log_2 \frac{1}{\pmtFPR})$ bits, this is the extraction complexity for
\bloomPMT; after this many queries, the malicious \requesterTerm knows
enough to conduct an \textit{offline} attack on set members.
In contrast:

\begin{prop}
  \label{prop:cuckooExtraction}
  The \extractionComplexity of \cuckooPMT is $\Omega(\frac{1}{\pmtFPR})$.
\end{prop}

\begin{proof}
Suppose the \responderTerm behaves according to the protocol, and for
each $\matrixRowIdx \in \nats{\cuckooNmbrBuckets}$, $\matrixColIdx \in
\nats{2}$, denote by $\plaintext{\matrixRowIdx}{\matrixColIdx} \in
\residues{\fieldOrder}$ the plaintext such that
$\matrixComponent{\queryMatrix}{\matrixRowIdx}{\matrixColIdx}
\in \ciphertextSpace{\pubKey}(\plaintext{\matrixRowIdx}{\matrixColIdx})$.
Similarly, denote by $\plaintext{\cuckooFingerprintCtext} \in
\residues{\fieldOrder}$ the plaintext such that
$\cuckooFingerprintCtext
\in \ciphertextSpace{\pubKey}(\plaintext{\cuckooFingerprintCtext})$.
A corollary of
\propsref{prop:comp:responder-security-one}{prop:comp:responder-security-two}
is then that in one PMT response, the \requesterTerm learns only the
result
\begin{align}
  \left(
  \fieldSum_{\encSumIdx = 1}^{\cuckooNmbrBuckets}
  \left(
  \matrixComponent{\cuckooFilter}{\matrixRowIdx}{\encSumIdx}
  \fieldMult
  \plaintext{\encSumIdx}{\matrixColIdx}
  \right)
  \right)
  \fieldAdd
  \plaintext{\cuckooFingerprintCtext}
  & \testEquiv{\fieldOrder} \fieldAddIdentity
\label{eqn:linear-congruence}
\end{align}
for each $\matrixRowIdx \in \nats{\cuckooBucketCapacity}$,
$\matrixColIdx \in \nats{2}$, i.e., a total of
$2\cuckooBucketCapacity$ linear congruence-mod-\fieldOrder tests,
where the \plaintext{\encSumIdx}{\matrixColIdx} and
\plaintext{\cuckooFingerprintCtext} values are chosen by the
\requesterTerm.  Even if \cuckooFilter represents a set consisting of
only a single element \pmtSetElmt chosen so that
$\cuckooFingerprintFn(\pmtSetElmt)$ is uniformly distributed in
\cuckooFingerprintRange, confirming the presence of
$\cuckooFingerprintFn(\pmtSetElmt)$ in \cuckooFilter requires, in
expectation, testing $\setSize{\cuckooFingerprintRange}/2$ linear
congruences and so performing
$\setSize{\cuckooFingerprintRange}/4\cuckooBucketCapacity$ PMT
queries.  Since
$\setSize{\cuckooFingerprintRange}/2\cuckooBucketCapacity \ge
1/\pmtFPR$ to retain the false-positive rate
\pmtFPR~\cite[\secrefstatic{5.1}]{fan2014:cuckoo}, \cuckooPMT has an
extraction complexity of $\Omega(\frac{1}{\pmtFPR})$ queries.
\end{proof}

The lower bound in \propref{prop:cuckooExtraction} is very coarse, in
that it applies even for a cuckoo filter \cuckooFilter storing a
single element---not to mention one storing many.  Moreover, there are
a number of measures that can make extraction even more difficult for
a malicious \requesterTerm at minimal expense to the \responderTerm.
\begin{itemize}[nosep,leftmargin=1em,labelwidth=*,align=left]
\item The \responderTerm can permute each column of \cuckooFilter
  independently after each execution of \cuckooPMT, since the query
  matrix \queryMatrix produced by a correct \requesterTerm will select
  the same elements from \cuckooFilter regardless of this permuting.
  Interpreting the results of multiple malicious PMT queries will
  become more difficult, however.
\item The \responderTerm can select any
  $\matrixComponent{\cuckooFilter}{\matrixRowIdx}{\matrixColIdx}
  \not\in \cuckooFingerprintRange$ uniformly at random from
  $\residues{\fieldOrder} \setminus \cuckooFingerprintRange$, ensuring
  that any linear test \eqnref{eqn:linear-congruence} involving
  \matrixComponent{\cuckooFilter}{\matrixRowIdx}{\matrixColIdx}
  (i.e., for which
  $\plaintext{\matrixRowIdx}{\matrixColIdx} \neq \fieldAddIdentity$,
  using the notation in the proof of \propref{prop:cuckooExtraction})
  succeeds with probability only
  $\frac{1}{\setSize{\residues{\fieldOrder} \setminus
      \cuckooFingerprintRange}}$.
\item The \responderTerm can randomly permute the elements of
  \resultMatrix before returning it, since the result computed by a
  correct \requesterTerm will be the same
  (\lineref{prot:requester:return}).  In doing so, the \requesterTerm
  is deprived of knowing which of its linear tests
  \eqnref{eqn:linear-congruence} were satisfied (if any were).
\end{itemize}

\subsection{Security Against a Malicious \ResponderTerm}
\label{sec:pmt:security:requester}

We now prove security for the \requesterTerm against a malicious
\responderTerm.  To do so, we define a malicious \responderTerm to be
a triple $\responderAdversary{} = \langle \responderAdversary{1},
\responderAdversary{2}, \responderAdversary{3}\rangle$ of algorithms
that participates in the experiment
\responderExperiment{\passwordChoiceBit}{\cuckooPMT} defined as
follows:
\[
\begin{array}{ll}
\mathrlap{\codeExpt \responderExperiment{\passwordChoiceBit}{\cuckooPMT}(\langle\responderAdversary{1},\! \responderAdversary{2},\! \responderAdversary{3}\rangle)} \\
\hspace{3em}
& \langle\pmtSetElmt{0}, \pmtSetElmt{1}, \responderAdversaryState{1} \rangle \gets \responderAdversary{1}() \\
& \langle \langle\pubKey, \cuckooFingerprintCtext, \queryMatrix\rangle, \privKey\rangle\gets \requesterAlgorithm{\ref{prot:requester:init}}{\ref{prot:requester:calculateQueryMatrix}}(\pmtSetElmt{\passwordChoiceBit}) \\
& \langle \resultMatrix, \responderAdversaryState{2} \rangle \gets \responderAdversary{2}(\langle\pubKey, \cuckooFingerprintCtext, \queryMatrix\rangle, \responderAdversaryState{1}) \\
& \requesterOutput \gets \requesterAlgorithm{\ref{prot:requester:checkCiphertexts}}{\ref{prot:requester:return}}(\privKey, \pubKey, \resultMatrix) \\
& \responderAdversaryBit \gets \responderAdversary{3}(\responderAdversaryState{2}, \requesterOutput) \\
& \codeReturn \responderAdversaryBit
\end{array}
\]
In this experiment,
\requesterAlgorithm{\ref{prot:requester:init}}{\ref{prot:requester:calculateQueryMatrix}}
denotes steps
\ref{prot:requester:init}--\ref{prot:requester:calculateQueryMatrix}
in \figref{fig:protocol}, producing the message \ref{prot:msg:request}
and the private key \privKey.  Similarly,
\requesterAlgorithm{\ref{prot:requester:checkCiphertexts}}{\ref{prot:requester:return}}
denotes steps
\ref{prot:requester:checkCiphertexts}--\ref{prot:requester:return}.
In the experiment, \responderAdversary{1} chooses two elements
\pmtSetElmt{0}, \pmtSetElmt{1}, and \passwordChoiceBit determines
which of the two that is used in the experiment.
\responderAdversary{2} is given message \ref{prot:msg:request} and
produces the response matrix \resultMatrix.  Finally,
\responderAdversary{3} is given the final result \requesterOutput of
the protocol from \lineref{prot:requester:return}, and outputs a bit
\responderAdversaryBit.  Note that though \figref{fig:protocol} does
not disclose \requester's result explicitly to \responder{}, we allow
it to be disclosed to \responder{} (i.e., \responderAdversary{3}) in
this analysis, to permit \cuckooPMT to be used in other contexts
(e.g.,~\cite{wang2019:reuse}).  We define the \responderTerm-adversary
advantage as
\begin{align*}
  \responderAdvantage{\cuckooPMT}(\responderAdversary{}) & = \prob{\responderExperiment{0}{\cuckooPMT}(\responderAdversary{}) = 0} \\
  & \hspace{2em} - \prob{\responderExperiment{1}{\cuckooPMT}(\responderAdversary{}) = 0} \\[6pt]
  \responderAdvantage{\cuckooPMT}(\timeBound) & = \max_{\responderAdversary{}} \responderAdvantage{\cuckooPMT}(\responderAdversary{})
\end{align*}
where the maximum is taken over all adversaries \responderAdversary{}
that run in time \timeBound.  Intuitively,
\responderAdvantage{\cuckooPMT}(\timeBound) captures the
ability of any adversary running in time \timeBound to
differentiate which of two passwords of its choice
the \requesterTerm uses to run the protocol.

We reduce security of \cuckooPMT against a \responderTerm adversary
\responderAdversary{} to IND-CPA
security~\cite[\dfnrefstatic{5.8}]{bellare2005:introduction} of the
encryption \encScheme.  The IND-CPA experiment
\indcpaExperiment{\indcpaBit}{\encScheme} is
\[
\begin{array}{ll}
\mathrlap{\codeExpt \indcpaExperiment{\indcpaBit}{\encScheme}(\indcpaAdversary)} \\
\hspace{3em}
& \langle \pubKey, \privKey \rangle \gets \keygen(1^\secParam) \\
& \indcpaAdversaryBit \gets \indcpaAdversary^{\indcpaLROracle{\pubKey}{\cdot}{\cdot}}(\pubKey) \\
& \codeReturn \indcpaAdversaryBit
\end{array}
\]
Here, the IND-CPA adversary \indcpaAdversary is given access to a
``left-or-right'' oracle \indcpaLROracle{\pubKey}{\cdot}{\cdot} that
takes two plaintexts $\plaintext{0}, \plaintext{1} \in
\residues{\fieldOrder}$ as input and returns
$\encrypt{\pubKey}(\plaintext{\indcpaBit})$.  Finally,
\indcpaAdversary returns a bit \indcpaAdversaryBit, which the
experiment returns.  We define
\begin{align*}
  \indcpaAdvantage{\encScheme}(\indcpaAdversary) & = \prob{\indcpaExperiment{0}{\encScheme}(\indcpaAdversary) = 0} - \prob{\indcpaExperiment{1}{\encScheme}(\indcpaAdversary) = 0} \\
  \indcpaAdvantage{\encScheme}(\timeBound, \queryBound) & =
  \max_{\indcpaAdversary} \indcpaAdvantage{\encScheme}(\indcpaAdversary)
\end{align*}
where the maximum is taken over all IND-CPA adversaries
\indcpaAdversary running in time \timeBound and making up to
\queryBound oracle queries.

\begin{prop}
  \label{prop:compProtocol:requester-security}
  $\responderAdvantage{\cuckooPMT}(\timeBound) \le
  2\indcpaAdvantage{\encScheme}(\timeBoundAlt, \queryBound)$ for
  $\queryBound = 2\cuckooNmbrBuckets + 1$ and some $\timeBoundAlt \le
  2\timeBound$.
\end{prop}

\begin{proof}
Given a \responderTerm adversary $\responderAdversary{} = \langle
\responderAdversary{1}, \responderAdversary{2}, \responderAdversary{3}
\rangle$, we construct an IND-CPA adversary \indcpaAdversary as
follows.  \indcpaAdversary first invokes \responderAdversary{1} to
obtain \pmtSetElmt{0} and \pmtSetElmt{1}.  Let
\plaintext{0\ciphertextIdx} denote the \ciphertextIdx-th plaintext
that \requester would encrypt in an execution of the protocol on
\pmtSetElmt{0}, and similarly let \plaintext{1\ciphertextIdx} denote
the \ciphertextIdx-th plaintext that \requester would encrypt in an
execution of the protocol on \pmtSetElmt{1}.  Then, \indcpaAdversary
simulates \requester exactly, except using its oracle to obtain the
\ciphertextIdx-th ciphertext $\ciphertext{\ciphertextIdx} \gets
\indcpaLROracle{\pubKey}{\plaintext{0\ciphertextIdx}}{\plaintext{1\ciphertextIdx}}$.
Note that because \indcpaAdversary does not have \privKey, it cannot
compute \requesterOutput as \requester would, and so it chooses
$\requesterOutput \getsr \{0, 1\}$ randomly and provides it to
\responderAdversary{3}.  When \responderAdversary{3} outputs
\responderAdversaryBit, \indcpaAdversary copies this bit as its output
\indcpaAdversaryBit.

The value \requesterOutput provided to \responderAdversary{3} is
correct, i.e., $\requesterOutput = \bigvee_{\matrixRowIdx,\matrixColIdx}
\encZeroTest{\privKey}(\matrixComponent{\resultMatrix}{\matrixRowIdx}{\matrixColIdx})$,
with probability $1/2$.  In this case, the simulation provided by
\indcpaAdversary to \responderAdversary{} is perfectly
indistinguishable from a real execution, and so
\begin{align*}
  \prob{\indcpaExperiment{0}{\encScheme}(\indcpaAdversary) = 0} &
  = \prob{\responderExperiment{0}{\cuckooPMT}(\responderAdversary{}) = 0} &
    \mbox{and} \\
  \prob{\indcpaExperiment{1}{\encScheme}(\indcpaAdversary) = 0} &
  = \prob{\responderExperiment{1}{\cuckooPMT}(\responderAdversary{}) = 0}
\end{align*}
As such, $\indcpaAdvantage{\encScheme}(\indcpaAdversary) \ge
\frac{1}{2} \responderAdvantage{\cuckooPMT}(\responderAdversary{})$.
\indcpaAdversary makes $\queryBound = 2\cuckooNmbrBuckets + 1$ queries
to construct \queryMatrix and \cuckooFingerprintCtext, and consumes
time at most $2\timeBound$ due to the time needed to construct both
\plaintext{0\ciphertextIdx} and \plaintext{1\ciphertextIdx} for each
\ciphertextIdx.
\end{proof}

\section{Performance and Scalability}
\label{sec:performance}

In this section, we describe an implementation of our framework and
evaluate its performance and scalability. The goals of our evaluation
are:
\begin{itemize}[nosep,leftmargin=1em,labelwidth=*,align=left]
\item To demonstrate the performance of our framework with varying
  parameters that could be used for real-world scenarios, e.g.,
  different numbers of participating websites for different users,
  and various sizes of suspicious password sets maintained at a
  \responderTerm;
\item To explore the potential performance degradation brought by
  adopting Tor to ensure the \requesterTerm's login privacy when the
  \directoryTerm is untrusted in this sense (i.e.,
  \untrustedForLoginPrivacy, as discussed in \secref{sec:directory});
  and
\item To evaluate the scalability of our prototype and to interpret
  its scalability in a real-world context.
\end{itemize}

\subsection{Implementation}
\label{sec:performance:impl}

Here, we give the salient details of our prototype implementation of
our framework.

\paragraph*{The PMT implementation}
We implemented \cuckooPMT (\secref{sec:pmt}) in Go.  We instantiated
the exponential Elgamal cryptosystem (\appref{sec:pmt:elgamal}) on a
prime-order elliptic curve group, secp256r1 (NIST
P-256)~\cite{certicom2000:sec2v1,nist2013:fips186-4}, which ensures
approximately 128-bit symmetric encryption security or 3072-bit RSA
security. For the cuckoo filter, we chose the bucket size
$\cuckooBucketCapacity = 16$, which permits an occupancy of 98\%.
That is, to accommodate a set of size \pmtSetSize, we need to build a
cuckoo filter with capacity at least $\pmtSetSize / 0.98$.

We leveraged precomputation on the \requesterTerm side for
\lineref{prot:requester:init} and
\lineref{prot:requester:calculateQueryMatrix} in
\figref{fig:protocol}. Specifically, the \requesterTerm can precompute
$\langle\pubKey, \privKey\rangle$, $2\cuckooNmbrBuckets-2$ ciphertexts
of \fieldAddIdentity, and two ciphertexts of \fieldMultIdentity such
that the online part of the computation in
\lineref{prot:requester:calculateQueryMatrix} is simply to assemble
the matrix \queryMatrix.

\paragraph*{The \directoryTerm}
We implemented the \directoryTerm in Go, leveraging multi-threading to
support parallel message processing. For each PMT query, the
\directoryTerm shuffles the intended \respondersTerm' addresses before
forwarding the \requesterTerm's query to those \respondersTerm and
shuffles all responses from \respondersTerm before returning them back
to the \requesterTerm. The first shuffling is to avoid evaluation bias
due to the networking or computation differences among multiple
\respondersTerm. The second shuffling further weakens the linkability
between responses and source \respondersTerm, as an extra layer of
security protection against a malicious \requesterTerm (see
\secref{sec:directory}).

\subsection{Experimental Setup}
\label{sec:performance:setup}

We set up one \requesterTerm, one \directoryTerm, and up to 256
\respondersTerm.  The \requesterTerm and the \directoryTerm ran on two
machines in our department, both with 2.67\gigahertz $\times$ 8
physical cores, 72\gibibytes RAM, and Ubuntu 18.04 x86\_64.  The 256
\respondersTerm were split evenly across eight Amazon EC2 instances in
the Eastern North American region, each with 3.2\gigahertz $\times$ 32
physical cores, 256\gibibytes RAM and Ubuntu 18.04 x86\_64. Each
\responderTerm was limited to one physical core, and had its own
exclusive data files, processes, and network sockets.

To test scenarios where the \directoryTerm is trusted for login
privacy (\trustedForLoginPrivacy) and where it is not
(\untrustedForLoginPrivacy), we set up two different types of
communication channels between the \requesterTerm and the
\directoryTerm.  For the \trustedForLoginPrivacy \directoryTerm, the
\requesterTerm and the \directoryTerm communicated directly.
For the \untrustedForLoginPrivacy \directoryTerm, we set up a private
Tor network, through which the \requesterTerm communicated with the
\directoryTerm via a newly built two-hop (i.e., with two Tor nodes)
circuit for each new query to hide its identity from the
\directoryTerm.  These two Tor nodes were chosen by Tor's default
selection algorithm from eight Tor nodes running in eight different
Amazon datacenters in North America and Europe.  In both the
\trustedForLoginPrivacy and \untrustedForLoginPrivacy cases, the
\directoryTerm communicated with \respondersTerm directly.

Each reported datapoint is the average of 50 runs. The relative
standard deviations of each datapoint for the \trustedForLoginPrivacy
\directoryTerm and \untrustedForLoginPrivacy \directoryTerm scenarios
were less than 4\% and 8\%, respectively.

\begin{figure}[t]
  \captionsetup[subfigure]{font=normalsize,labelfont=normalsize}
  \begin{subfigure}[t]{.1\columnwidth}
    \setlength\figureheight{2in}
    \begin{minipage}[t]{1\columnwidth}
      \centering
      \hspace*{4.5em}
      \resizebox{!}{2.5em}{\newenvironment{customlegend}[1][]{%
    \begingroup
    \csname pgfplots@init@cleared@structures\endcsname
    \pgfplotsset{#1}%
}{%
    \csname pgfplots@createlegend\endcsname
    \endgroup
}%

\def\addlegendimage{\csname pgfplots@addlegendimage\endcsname}

\begin{tikzpicture}

\definecolor{color0}{rgb}{0.129411764705882,0.380392156862745,0.549019607843137}

\begin{customlegend}[
	legend style={{font={\fontsize{10pt}{12}\selectfont}},{draw=none}},
	legend cell align={left},
	legend columns=3,
	legend entries={{\nmbrResponders{\accountId{}} = $1\quad$},{\nmbrResponders{\accountId{}} = $64\quad$},{\nmbrResponders{\accountId{}} = $128\quad$},{\nmbrResponders{\accountId{}} = $192\quad$},{\nmbrResponders{\accountId{}} = $256$}}]
\addlegendimage{line width=1pt, color0}
\addlegendimage{line width=1pt, color0, dotted}
\addlegendimage{line width=1pt, color0, dash pattern=on 1pt off 3pt on 3pt off 3pt}
\addlegendimage{line width=1pt, color0, dashed}
\addlegendimage{line width=1.5pt, color0, dotted}

\end{customlegend}

\end{tikzpicture}}
    \end{minipage}
  \end{subfigure}
  
  \hspace*{0.75em}
  \begin{subfigure}[b]{.49\columnwidth}
    \setlength\figureheight{2.3in}
    \begin{minipage}[b]{1\textwidth}
      \centering
      \vspace*{0em}\resizebox{!}{11em}{\begin{tikzpicture}

\definecolor{color0}{rgb}{0.129411764705882,0.380392156862745,0.549019607843137}

\pgfplotsset{every axis/.append style={
					xlabel={},
					ylabel={Response time (s)},
					compat=1.3,
                    label style={font=\small},
                    tick label style={font=\small}  
                    }}

\begin{axis}[
xmin=128, xmax=4096,
ymin=0, ymax=6,
xmode=log,
log basis x={2},
xtick={128,256,512,1024,2048,4096},
ytick={0,1,2,3,4,5,6},
yticklabels={0,1,2,3,4,5,6},
width=1.1\figurewidth,
height=1.1\figurewidth,
tick align=outside,
tick pos=left,
xmajorgrids,
minor tick num=1,
x grid style={lightgray!92.026143790849673!black},
ymajorgrids,
y grid style={lightgray!92.026143790849673!black},
]
\addplot [line width=1pt, color0]
table {%
128  0.078624
256  0.129329
512  0.209837
1024 0.382315
2048 0.741059
4096 1.502319
};
\addplot [line width=1pt, color0, dotted]
table {%
128  0.189479
256  0.261052
512  0.380885
1024 0.630404
2048 1.121129
4096 2.181854
};
\addplot [line width=1pt, color0, dash pattern=on 1pt off 3pt on 3pt off 3pt]
table {%
128  0.333078
256  0.427641
512  0.572476
1024 0.880904
2048 1.553705
4096 3.010296
};
\addplot [line width=1pt, color0, dashed]
table {%
128  0.37329
256  0.480669
512  0.665309
1024 1.055505
2048 1.885773
4096 3.724219
};
\addplot [line width=1.25pt, color0, dashed]
table {%
128  0.47435
256  0.59649
512  0.806687
1024 1.26532
2048 2.23752
4096 4.521978
};
\end{axis}
\end{tikzpicture}

\hspace*{-0.5em}\begin{tikzpicture}
\definecolor{color0}{rgb}{0.129411764705882,0.380392156862745,0.549019607843137}

\pgfplotsset{every axis/.append style={
					xlabel={},
					ylabel={},
					compat=1.3,
                    label style={font=\small},
                    tick label style={font=\small}  
                    }}
                   
\begin{axis}[
xmin=128, xmax=4096,
ymin=0, ymax=6,
xmode=log,
log basis x={2},
xtick={128,256,512,1024,2048,4096},
ytick={0,1,2,3,4,5,6},
yticklabels={},
width=1.1\figurewidth,
height=1.1\figurewidth,
tick align=outside,
tick pos=left,
minor tick num=1,
xmajorgrids,
x grid style={lightgray!92.026143790849673!black},
ymajorgrids,
y grid style={lightgray!92.026143790849673!black},
]
\addplot [line width=1pt, color0]
table {%
128  0.573917
256  0.661347
512  0.754593
1024 0.900982
2048 1.306185
4096 2.277571
};
\addplot [line width=1pt, color0, dotted]
table {%
128  0.800866
256  0.905682
512  1.038728
1024 1.280822
2048 1.942558
4096 3.129792
};
\addplot [line width=1pt, color0, dash pattern=on 1pt off 3pt on 3pt off 3pt]
table {%
128  0.948811
256  1.13057
512  1.310208
1024 1.639897
2048 2.436169
4096 3.892989
};
\addplot [line width=1pt, color0, dashed]
table {%
128  1.15235
256  1.363042
512  1.559455
1024 2.009603
2048 2.932855
4096 4.567939
};
\addplot [line width=1.25pt, color0, dashed]
table {%
128  1.361295
256  1.501054
512  1.737416
1024 2.290179
2048 3.233663
4096 5.658453
};
\end{axis}
\end{tikzpicture}}
      \hspace*{2.5em}
      \begin{minipage}[t]{8em}
        \vspace*{-0.01em}
        \caption{\trustedForLoginPrivacy \directoryTerm}
        \label{fig:responseTime:woTor}
      \end{minipage}%
      \hspace*{2.75em}
      \begin{minipage}[t]{8em}
        \vspace*{-0.01em}
        \caption{\untrustedForLoginPrivacy  \directoryTerm}
        \label{fig:responseTime:Tor}
      \end{minipage}%
    \end{minipage}
  \end{subfigure}%
  
  \vspace*{-2.65em}
  \begin{subfigure}[b]{.43\columnwidth}
    \setlength\figureheight{2in}
    \begin{minipage}[b]{1\textwidth}
      \centering
      \hspace*{12.5em}
      \resizebox{!}{1.4em}{\begin{tikzpicture}
\node at (0,0)[
  scale=1,
  anchor=south,
  text=black,
  rotate=0
]{$\pmtSetSize$};
\end{tikzpicture}}
      \vspace*{-0.4em}
    \end{minipage}
  \end{subfigure}
  \vspace*{0.5em}
  \caption{Response time with varying \pmtSetSize and
    \nmbrResponders{\accountId{}}}
  \label{fig:responseTime}
\end{figure}
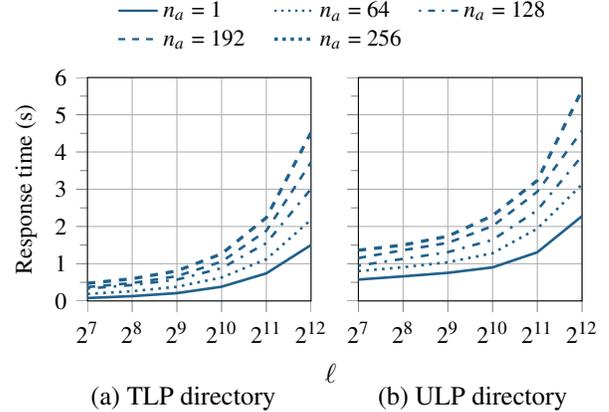

\subsection{Response Time}
\label{sec:performance:results:latency}

We first report the results of our response-time evaluation
experiments for our implementation and setup above.  In these
experiments, the \requesterTerm issued one \cuckooPMT query via the
\directoryTerm to \nmbrResponders{\accountId{}}
\respondersTerm, and awaited the \nmbrResponders{\accountId{}}
responses from the \respondersTerm.  The response time is the duration
observed by the \requesterTerm between starting to generate a PMT
query (after precomputation) and receiving all responses and
outputting the result.  \figref{fig:responseTime:woTor} shows the
response time when \directoryTerm is trusted for login privacy
(\trustedForLoginPrivacy), while \figref{fig:responseTime:Tor} shows
the response time when the \directoryTerm is untrusted for login
privacy (\untrustedForLoginPrivacy). As mentioned in
\secref{sec:performance:setup}, in the former case the \requesterTerm
directly connected to the \directoryTerm, while in the latter case, the
\requesterTerm communicated with the \directoryTerm via a Tor
circuit. Tor circuit setup is included in the response-time
measurement.  In both cases, the \directoryTerm had direct connections
with all \respondersTerm, with no Tor circuits involved.

The main takeaway from \figref{fig:responseTime} is that when the
capacity \pmtSetSize of the suspicious password set at each
\responderTerm was relatively small, say $\pmtSetSize \le 2^9$, the
response time was less than 1\secs with a \trustedForLoginPrivacy
\directoryTerm and less than 2\secs with a \untrustedForLoginPrivacy
\directoryTerm, even for users with a large number of web accounts,
say $\nmbrResponders{\accountId{}} = 256$.  Since the average user has
far fewer accounts ($\nmbrResponders{\accountId{}} \approx
26$~\cite{pearman2017:habitat}), and since modern password-management
recommendations would allow suspicious-password sets to be capped at a
size $\pmtSetSize < 2^7$ (see \secref{sec:stuffing:algorithm}), we can
expect the response time for an isolated request to be less than
1\secs even in the \untrustedForLoginPrivacy \directoryTerm scenario.

\subsection{Scalability}
\label{sec:performance:results:scalability}

To evaluate the scalability of our framework, we measured the maximum
qualifying response rate that our prototype can achieve.  Here, a
\textit{qualifying} response is one for which the response time falls
within a certain allowance, which we specified as 5\secs in the
\trustedForLoginPrivacy \directoryTerm case and 8\secs in the
\untrustedForLoginPrivacy \directoryTerm case.  For each query,
\nmbrResponders{\accountId{}} \respondersTerm were chosen uniformly at
random from all 256 \respondersTerm.  To produce a conservative
\untrustedForLoginPrivacy estimate, we required the \requesterTerm to
communicate with the \untrustedForLoginPrivacy \directoryTerm via a
new Tor circuit for each new query, to account for the potential
scalability degradation brought by building Tor connections.

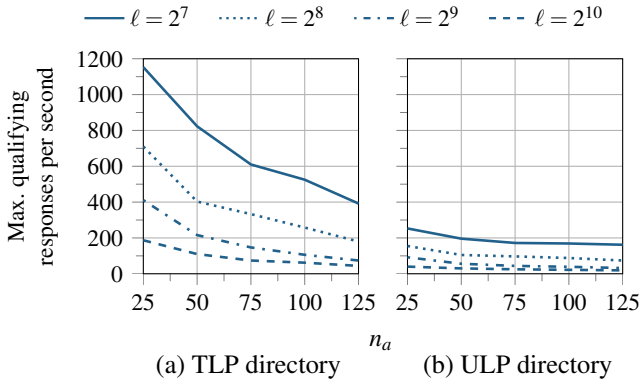
\begin{figure}[h]
  \captionsetup[subfigure]{font=normalsize,labelfont=normalsize}
  \begin{subfigure}[t]{.1\columnwidth}
    \setlength\figureheight{2in}
    \begin{minipage}[t]{1\columnwidth}
      \hspace*{2.5em}
      \resizebox{!}{1.6em}{\newenvironment{customlegend}[1][]{%
    \begingroup
    \csname pgfplots@init@cleared@structures\endcsname
    \pgfplotsset{#1}%
}{%
    \csname pgfplots@createlegend\endcsname
    \endgroup
}%

\def\addlegendimage{\csname pgfplots@addlegendimage\endcsname}

\begin{tikzpicture}

\definecolor{color0}{rgb}{0.129411764705882,0.380392156862745,0.549019607843137}

\begin{customlegend}[
	legend style={{font={\fontsize{10pt}{12}\selectfont}},{draw=none}},
	legend cell align={left},
	legend columns=4,
	legend entries={{$\pmtSetSize = 2^7\quad$},{$\pmtSetSize = 2^8\quad$},{$\pmtSetSize = 2^9\quad$},{$\pmtSetSize = 2^{10}\quad$}}]
\addlegendimage{line width=1pt, color0}
\addlegendimage{line width=1pt, color0, dotted}
\addlegendimage{line width=1pt, color0, dash pattern=on 1pt off 3pt on 3pt off 3pt}
\addlegendimage{line width=1pt, color0, dashed}

\end{customlegend}

\end{tikzpicture}}
    \end{minipage}
  \end{subfigure}
  
  \begin{subfigure}[b]{.48\columnwidth}
    \setlength\figureheight{2.4in}
    \begin{minipage}[b]{1\textwidth}
      \centering
      \vspace*{0em}\resizebox{!}{10.5em}{\begin{tikzpicture}

\definecolor{color0}{rgb}{0.129411764705882,0.380392156862745,0.549019607843137}

\pgfplotsset{every axis/.append style={
					xlabel={},
					compat=1.3,
                    label style={font=\small},
                    tick label style={font=\small}  
                    }}

\begin{axis}[
ylabel style={align=center}, 
ylabel=Max. qualifying\\responses per second,
xmin=25, xmax=125,
ymin=0, ymax=1200,
xtick={25,50,75,100,125},
ytick={0,200,400,600,800,1000,1200},
yticklabels={0,200,400,600,800,1000,1200},
width=1.1\figurewidth,
height=1.1\figurewidth,
tick align=outside,
tick pos=left,
xmajorgrids,
minor tick num=1,
x grid style={lightgray!92.026143790849673!black},
ymajorgrids,
y grid style={lightgray!92.026143790849673!black},
]
\addplot [line width=1pt, color0]
table {%
25		1153
50		823
75		610
100		525
125		392
};
\addplot [line width=1pt, color0, dotted]
table {%
25		710
50		403
75		333
100		258
125		180
};
\addplot [line width=1pt, color0, dash pattern=on 1pt off 3pt on 3pt off 3pt]
table {%
25		411
50		216
75		147
100		106
125		74
};
\addplot [line width=1pt, color0, dashed]
table {%
25		187
50		110
75		74
100		62
125		44
};
\end{axis}
\end{tikzpicture}

\hspace*{-0.5em}\begin{tikzpicture}
\definecolor{color0}{rgb}{0.129411764705882,0.380392156862745,0.549019607843137}

\pgfplotsset{every axis/.append style={
					xlabel={},
					ylabel={},
					compat=1.3,
                    label style={font=\small},
                    tick label style={font=\small}  
                    }}
                   
\begin{axis}[
xmin=25, xmax=125,
ymin=0, ymax=1200,
xtick={25,50,75,100,125},
ytick={0,200,400,600,800,1000,1200},
yticklabels={},
width=1.1\figurewidth,
height=1.1\figurewidth,
tick align=outside,
tick pos=left,
minor tick num=1,
xmajorgrids,
x grid style={lightgray!92.026143790849673!black},
ymajorgrids,
y grid style={lightgray!92.026143790849673!black},
]
\addplot [line width=1pt, color0]
table {%
25		253
50		196
75		172
100		169
125		162
};
\addplot [line width=1pt, color0, dotted]
table {%
25		155
50		105
75		97
100		88
125		74
};
\addplot [line width=1pt, color0, dash pattern=on 1pt off 3pt on 3pt off 3pt]
table {%
25		92
50		56
75		44
100		39
125		32
};
\addplot [line width=1pt, color0, dashed]
table {%
25		40
50		30
75		25
100		22
125		19
};
\end{axis}
\end{tikzpicture}}
      \hspace*{5.5em} 
      \begin{minipage}[t]{7.5em}
	     \vspace*{-0.01em}
        \caption{\trustedForLoginPrivacy  \directoryTerm}
        \label{fig:throughput:woTor}
      \end{minipage}%
      \hspace*{2.5em}
      \begin{minipage}[t]{7.5em}
	     \vspace*{-0.01em}
        \caption{\untrustedForLoginPrivacy  \directoryTerm}
        \label{fig:throughput:Tor}
      \end{minipage}%
    \end{minipage}
  \end{subfigure}%
  
  \vspace*{-2.5em}
  \begin{subfigure}[b]{.43\columnwidth}
    \setlength\figureheight{2in}
    \begin{minipage}[b]{1\textwidth}
      \centering
      \hspace*{13.5em}
      \resizebox{!}{1.3em}{\begin{tikzpicture}
\node at (0,0)[
  scale=1,
  anchor=south,
  text=black,
  rotate=0
]{\nmbrResponders{\accountId{}}};
\end{tikzpicture}}\vspace*{-0.4em}
    \end{minipage}
  \end{subfigure}
  \vspace*{1.5em}
  \caption{Maximum qualifying responses per second}
  \label{fig:throughput}
\end{figure}

The results of these tests are shown in \figref{fig:throughput}.  To
put these numbers in context, consider that there are $\approx$ 369.4M
credential-stuffing login attempts per day for the four U.S.\ industries
listed in \tblref{tbl:stats}.  According to the reported success rates
of credential stuffing, 0.83M of these login attempts are with correct
passwords.  Moreover, there are $\approx$ 187.6M legitimate login
attempts per day in these industries; for a conservative estimate
here, we assume that they all provide the correct passwords.  With the
baseline \textit{phishing} ADS configuration used in
\secref{sec:stuffing:eff}, i.e., $(\adsFPR{\adsCount},
\adsTPR{\adsCount}) = (0.30, 0.95)$, there would be 57.07M ($=
187.6\textrm{M} \times 0.3 + 0.83\textrm{M} \times 0.95$) login
attempts per day that induce PMT queries or, in other words, about 660
PMT queries per second.  With the baseline \textit{researching} ADS
configuration used in \secref{sec:stuffing:eff}, i.e.,
$(\adsFPR{\adsCount}, \adsTPR{\adsCount}) = (0.10, 0.99)$, an
analogous calculation suggests 19.58M PMT queries per day or 227 per
second.  Our experiments suggest that our prototype could achieve
these throughputs with just one \directoryTerm server for a range of
configurations.  For example, configured for \textit{phishing}
attackers, our \trustedForLoginPrivacy \directoryTerm should support
the requisite throughput when $\pmtSetSize < 2^7$ for up to
$\nmbrResponders{\accountId{}} \approx 69$ \respondersTerm.
Configured for \textit{researching} attackers, even the
\untrustedForLoginPrivacy configuration could support the expected
throughput when $\pmtSetSize < 2^7$ for up to
$\nmbrResponders{\accountId{}} \approx 36$ \respondersTerm, and the
\trustedForLoginPrivacy configuration would support the needed
throughput when $\pmtSetSize < 2^7$ for \nmbrResponders{\accountId{}}
as large as $\nmbrResponders{\accountId{}} = 125$.

\begin{table}[h]
  \centering
  \small
  \begin{tabular}{@{}lrcc@{}}
    \toprule
    & \parbox[m]{0.25\columnwidth}{\centering Credential-stuffing login attempts per day}
    & \parbox[m]{0.24\columnwidth}{\centering Proportion that succeed}
    & \parbox[m]{0.2\columnwidth}{\centering Proportion of all login attempts}\\
    Industry
    & \cite[\tblsrefstatic{3}{6}]{shape2018:spill}
    & \cite[\tblsrefstatic{3}{6}]{shape2018:spill}
    & \cite[\figrefstatic{13}]{shape2018:spill} \\
    \midrule
    Airline &   1.4M\hspace{2em} & 1.00\% & 60\% \\[4pt]
    Hotel   &   4.3M\hspace{2em} & 1.00\% & 44\% \\[4pt]
    Retail  & 131.5M\hspace{2em} & 0.50\% & 91\% \\[4pt]
    \parbox[m]{0.125\columnwidth}{Consumer banking}
            & 232.2M\hspace{2em} & 0.05\% & 58\% \\
    \bottomrule
  \end{tabular}
  \caption{Credential-stuffing estimates for U.S.\ industries}
  \label{tbl:stats}
\end{table}

Though already encouraging, these results leveraged only one
\requesterTerm machine and one \directoryTerm machine, and
concentrated all PMT queries to be served by (a randomly chosen subset
of size \nmbrResponders{\accountId{}} of) the same 256
\respondersTerm, each allocated only a single CPU core.
\ResponderTerm CPU was the bottleneck in the \trustedForLoginPrivacy
experiments.  Tor was the bottleneck in the \untrustedForLoginPrivacy
experiments, being the only difference from the
\trustedForLoginPrivacy experiments.  With a more dispersed query
pattern launched from more \requestersTerm, with more capable
\respondersTerm, and with a distributed directory, our design could
scale even further.

\section{Conclusion}
\label{sec:conclusion}

In this paper we have proposed a novel framework by which websites can
coordinate to detect active credential-stuffing attacks on individual
accounts.  Our framework accommodates the tendencies of human users to
reuse passwords, to enter their passwords into incorrect sites, etc.,
while still providing good detection accuracy across a range of
operating points.  The framework is built on a new private
membership-test protocol that scales better than previous
alternatives and/or ensures a higher \textit{extraction
  complexity}, which captures the ability of a \requesterTerm in the
protocol to extract enough information to search elements of the set
offline.  Using probabilistic model checking applied to novel
experiments designed to capture both usability and security, we
quantified the benefits of our framework.  Finally, we showed through
empirical results with our prototype implementation that our design
should scale easily to accommodate the login load of large sectors of
the U.S. economy, for example.

\paragraph{Acknowledgments}
We are grateful to the anonymous reviewers and to our shepherd,
Prof.\ Stephen Checkoway, for their constructive feedback.

\bibliographystyle{plain}
\bibliography{tight,main}

\appendix

\section{Exponential ElGamal Encryption}
\label{sec:pmt:elgamal}

A cryptosystem that can be used to instantiate the specification of
\secref{sec:pmt:crypto} is a variant of ElGamal
encryption~\cite{elgamal1985:public-key} commonly referred to as
``exponential ElGamal'' and implemented as follows (see,
e.g.,~\cite{cramer1997:election}).  It uses an algorithm \elgGen that,
on input $1^{\secParam}$, outputs a multiplicative abelian group
\elgGroup of order \fieldOrder for a \secParam-bit prime \fieldOrder.
\begin{itemize}[nosep,leftmargin=1em,labelwidth=*,align=left]
\item $\keygen(1^\kappa)$ generates $\elgGroup \gets
  \elgGen(1^\secParam)$; selects $\elgPrivKey \getsr
  \residues{\fieldOrder}$; and returns a private key $\privKey =
  \langle \elgPrivKey\rangle$ and public key $\pubKey = \langle
  \elgGroup, \elgGroupGenerator, \elgPubKey \rangle$, where
  \elgGroupGenerator is a generator of \elgGroup, and $\elgPubKey
  \gets \elgGroupGenerator^{\elgPrivKey}$.

\item $\encrypt{\langle\elgGroup,
  \elgGroupGenerator,\elgPubKey\rangle}(\plaintext)$ returns
  $\langle\elgEphemeralPubKey{}, \elgCiphertext{} \rangle$ where
  $\elgEphemeralPubKey{} \gets
  \elgGroupGenerator^{\elgEphemeralPrivKey{}}$,
  $\elgEphemeralPrivKey{} \getsr \residues{\fieldOrder}$, and
  $\elgCiphertext{} \gets \elgGroupGenerator^{\plaintext}
  \elgPubKey^{\elgEphemeralPrivKey{}}$.

\item $\langle \elgEphemeralPubKey{1}, \elgCiphertext{1} \rangle
  \encAdd{\langle \elgGroup, \elgGroupGenerator, \elgPubKey\rangle}
  \langle \elgEphemeralPubKey{2}, \elgCiphertext{2} \rangle$ returns
  $\langle \elgEphemeralPubKey{1} \elgEphemeralPubKey{2}
  \elgGroupGenerator^\elgGroupExponent, \elgCiphertext{1}
  \elgCiphertext{2} \elgPubKey^\elgGroupExponent\rangle$ for
  $\elgGroupExponent \getsr \residues{\fieldOrder}$ if
  $\{\elgEphemeralPubKey{1}, \elgCiphertext{1},
  \elgEphemeralPubKey{2}, \elgCiphertext{2}\} \subseteq \elgGroup$ and
  returns $\bot$ otherwise.

\item $\encZeroTest{\langle \elgPrivKey \rangle}(\langle
  \elgEphemeralPubKey{}, \elgCiphertext{}\rangle)$ returns \boolTrue
  if $\{\elgEphemeralPubKey{}, \elgCiphertext{}\} \subseteq
  \elgGroup$ and $\elgCiphertext{} =
  \elgEphemeralPubKey{}^{\elgPrivKey}$, and returns \boolFalse
  otherwise.
\end{itemize}

To use this cryptosystem in \cuckooPMT, it is necessary to test for
ciphertext validity (\lineref{prot:responder:checkCiphertexts} and
\lineref{prot:requester:checkCiphertexts}).  The next proposition
shows that it suffices to test membership in \elgGroup.

\begin{prop}
  \label{prop:elgValidCtexts}
  For the exponential ElGamal cryptosystem,
  $\ciphertextSpace{\langle \elgGroup, \elgGroupGenerator,
    \elgPubKey\rangle}$ $=$ $\elgGroup \times \elgGroup$.
\end{prop}

\begin{proof}
$\ciphertextSpace{\langle \elgGroup, \elgGroupGenerator,
    \elgPubKey\rangle}$ $\subseteq$ $\elgGroup \times \elgGroup$
  follows by construction, and so we focus on proving $\elgGroup
  \times \elgGroup$ $\subseteq$ $\ciphertextSpace{\langle \elgGroup,
    \elgGroupGenerator, \elgPubKey\rangle}$.  For any $\langle
  \elgEphemeralPubKey{}, \elgCiphertext{}\rangle$ $\in$ $\elgGroup
  \times \elgGroup$, consider the group element
  $\elgCiphertext{}\elgEphemeralPubKey{}^{-\elgPrivKey}$ where
  $\elgGroupGenerator^{\elgPrivKey} = \elgPubKey$.  Since
  \elgGroupGenerator is a generator of \elgGroup, there is a plaintext
  $\plaintext \in \residues{\fieldOrder}$ such that
  $\elgGroupGenerator^{\plaintext} =
  \elgCiphertext{}\elgEphemeralPubKey{}^{-\elgPrivKey}$ and so
  $\langle \elgEphemeralPubKey{}, \elgCiphertext{}\rangle$ $\in$
  $\ciphertextSpace{\langle \elgGroup, \elgGroupGenerator,
    \elgPubKey\rangle}(\plaintext)$.
\end{proof}

\end{document}